\PassOptionsToPackage{prologue,dvipsnames}{xcolor}
\documentclass[acmsmall,screen]{acmart}
\usepackage{style}

\AtBeginDocument{%
  }

\setcopyright{cc}
\copyrightyear{2026}
\acmYear{2026}
\acmDOI{10.1145/3797874}

\makeatletter
\renewcommand{\@journalNameShort}{\@journalName}
\makeatother
\acmJournal{TOPLAS}
\acmVolume{48}
\acmNumber{1}
\acmArticle{5}
\acmMonth{3}

\makeatletter
\renewcommand{\@journalNameShort}{\text{ACM Transactions on Programming Languages and Systems}}
\makeatother

\begin{document}

\title{Denotation-based Compositional Compiler Verification}


\author{Zhang Cheng}
\orcid{0009-0005-3833-6612}
\email{mooc3535@gmail.com}
\affiliation{%
  \institution{Shanghai Jiao Tong University}
  \state{Shanghai}
  \country{China}
}

\author{Jiyang Wu}
\orcid{0009-0002-6638-4839}
\email{wujiyang@sjtu.edu.cn}
\affiliation{%
  \institution{Shanghai Jiao Tong University}
  \state{Shanghai}
  \country{China}
}

\author{Di Wang}
\orcid{0000-0002-2418-7987}
\email{wangdi95@pku.edu.cn}
\affiliation{%
  \institution{Peking University}
  \state{Beijing}
  \country{China}
}

\author{Qinxiang Cao}
\orcid{0000-0002-5678-6538}
\authornote{Corresponding author}
\email{caoqinxiang@gmail.com}
\affiliation{%
  \institution{Shanghai Jiao Tong University}
  \state{Shanghai}
  \country{China}
}

\renewcommand{\shortauthors}{Z. Cheng et al.}

\begin{abstract}
  A desired but challenging property of compiler verification is compositionality, in the sense that the compilation correctness of a program can be deduced incrementally from that of its substructures ranging from statements, functions, and modules.
    This article proposes a novel compiler verification framework based on denotational semantics for better compositionality, compared to previous approaches based on small-step operational semantics and simulation theories.
    Our denotational semantics is defined by semantic functions that map a syntactic component to a semantic domain composed of multiple behavioral \emph{sets}, with compiler correctness established through behavior refinement between the semantic domains of the source and target programs.
    The main contributions of this article include proposing a denotational semantics for open modules, a novel semantic linking operator, and a refinement algebra that unifies various behavior refinements, making compiler verification structured and compositional.
	Furthermore, our formalization captures the full meaning of a program and bridges the gap between traditional power-domain-based denotational semantics and the practical needs of compiler verification.
We apply our denotation-based framework to verify the front-end of CompCert and typical optimizations on simple prototypes of imperative languages. Our results demonstrate that the compositionality from sub-statements to statements, from functions to modules, and from modules to the whole program can be effectively achieved.

\end{abstract}




\received{05 December 2024}
\received[revised]{27 October 2025}
\received[accepted]{23 January 2026}

\maketitle

\section{Introduction}
\label{sec:intro}

Research on compiler verification has a long history. The state-of-the-art achievement in this field is CompCert \cite{compcert}.
Specifically,
CompCert translates programs written in a large subset of C into optimized assembly code, going through multiple intermediate languages.
For each of them, program behavior is formulated by a labeled transition system (LTS) according to the small-step operational semantics. Compiler correctness is then achieved by showing a backward simulation\footnote{
In fact, CompCert verifies almost all compilation phases using forward simulation, and then derives the backward simulation by certain side conditions after vertically composing the compilation correctness of each phase.}
asserting that every execution step of the target program can be simulated by zero or more execution steps of the source program in a behaviorally consistent way, as shown in Fig. \ref{subfig:bsim-compcert}.
\begin{figure}[bht]
    \begin{subfigure}[b]{.32\linewidth}
    \centering
       \includegraphics[width=0.6\linewidth]{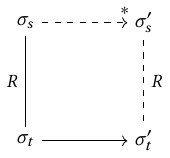}
    \caption{The backward simulation\protect\footnotemark}
    \label{subfig:bsim-compcert}
    \end{subfigure} 
    \begin{subfigure}[b]{.32\linewidth}
    \centering
    \includegraphics[width=0.6\linewidth]{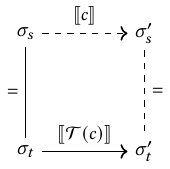}
    \caption{Simple behavior refinement}
    \label{subfig:simpl_bref}
    \end{subfigure} 
    \begin{subfigure}[b]{.32\linewidth}
        \flushright
    \includegraphics[width=0.6\linewidth]{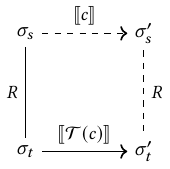}
    \centering        
    \caption{General behavior refinement}
    \label{subfig:general_bref}
    \end{subfigure}   
    \caption{The backward simulation of CompCert and the behavior refinement between denotations.}
    \Description{This figure illustrates the backward simulation of CompCert and the behavior refinement between denotations. (a) The first subfigure shows the backward simulation in CompCert, where each execution step of the target program can be simulated by zero or more execution steps of the source program, ensuring behavior consistency. (b) The second subfigure illustrates simple behavior refinement, where the target program's behavior is contained within the source program's behavior. (c) The third subfigure depicts general behavior refinement, where the source and target may use different state spaces.}
    \label{fig:behavior_refinement}
\end{figure}

\footnotetext{
The relation $R$ can be indexed by a natural number that strictly decreases on source stuttering steps, ensuring progress.}

This article proposes a different framework for compiler verification---a denotation-based approach that focuses on the overall properties of programs and enjoys better proof compositionality.
For example, considering type-safe and nondeterministic programs with the set of program states {\tts state}, a textbook denotational semantics \citep[Chapter 8]{plotkin1983domains} %
for a program statement $c$ can be defined by a binary relation between lifted program states
(i.e., $\llbracket c \rrbracket \subseteq \ttt{state} \times \ttt{state}_{\bot}$\footnote{
$\ttt{state}_{\bot}$ represents $\ttt{state} \cup \{\bot\}$, where $\bot$ represents nontermination.}),
 meaning that for any $(\sigma_1, \sigma_2) \in \llbracket c \rrbracket$,  executing statement $c$ from state $\sigma_1$ may eventually terminate at state $\sigma_2$ if $\sigma_2 \neq \bot$, or otherwise ($\sigma_2 = \bot$) executing $c$ from state $\sigma_1$ may not terminate.
    Thus,
    \begin{itemize}
        \item the sequential statement's denotation can be defined by \emph{composition} of binary relations, namely
    $\llbracket c_1; c_2 \rrbracket \triangleq
    \llbracket c_1 \rrbracket \circ \llbracket c_2 \rrbracket
    $; and
        \item transformation correctness can be defined by set inclusion, i.e., a transformation $\mathcal{T}$ is sound iff. $\forall c, \llbracket \mathcal{T}(c) \rrbracket \subseteq \llbracket c \rrbracket$, which allows the transformation result $\mathcal{T}(c)$ to have fewer possible behaviors than the original program, but no extra behaviors.
    \end{itemize}   
    In this setting, if a transformation $\mathcal{T}$ satisfies $\mathcal{T}(c_1; c_2) = (\mathcal{T}(c_1); \mathcal{T}(c_2))$, then its transformation correctness is compositional in terms of the composition relation, i.e.,
\begin{align}
    \text{if }\llbracket \mathcal{T}(c_1) \rrbracket \subseteq \llbracket c_1 \rrbracket
\text{ and }
   \llbracket \mathcal{T}(c_2) \rrbracket \subseteq \llbracket c_2 \rrbracket
\text{, then }\llbracket \mathcal{T}(c_1) \rrbracket \circ \llbracket \mathcal{T}(c_2) \rrbracket
      \subseteq
      \llbracket c_1 \rrbracket \circ \llbracket c_2 \rrbracket
      \label{eq:compose_mono}
\end{align}
This compositional proof style may look unrealistic for verifying real compilation passes, since it is based on extremely simplified semantic assumptions. We fill the gap between this ideal proof paradigm and reasonably practical compiler verification by overcoming the following difficulties:

\paragraph{General refinement.}
In practice, $\bracket{\mathcal{T}(c)} \subseteq \bracket{c}$ is not general enough for describing all transformation correctness, especially for transformations that may alter how states are represented.
In such cases, a general behavior refinement, based on a state-matching relation $R$ ($\subseteq \ttt{state}_s \times \ttt{state}_t$), is required. We define $\bracket{\mathcal{T}(c)} \sqsubseteq_{R} \bracket{c}$ as follows (we omit the subscript $R$ when there is no ambiguity):
\begin{equation}
  \forall \sigma_s\ \sigma_t\ \sigma_t',  (\sigma_s, \sigma_t)\in R \Rightarrow (\sigma_t,\sigma_t') \in \bracket{\mathcal{T}(c)} \Rightarrow \exists \sigma_s', (\sigma_s', \sigma_t')\in R \wedge (\sigma_s, \sigma_s') \in \bracket{c} 
\end{equation}
This formula is illustrated by Fig. \ref{subfig:general_bref}, and obviously \(\sem{\compile(c)} \subseteq \sem{c}\) is the special case when \(R\) is defined as the equal relation ``\(=\)'', as illustrated in Fig.~\ref{subfig:simpl_bref}. We find that such a general definition of refinement is also compositional w.r.t. common semantic combinators like Formula~(\ref{eq:compose_mono}).

\paragraph{More program features.}
Realistic programming languages involve a richer set of features that must be accounted for in the semantics. To this end, we use input-output event traces to express a program's observable behavior as in CompCert, and use separate \emph{behavior sets} to handle different styles of control flow for CompCert intermediate languages and especially to describe terminating, diverging, and aborting behaviors that traditional powerdomains \cite{DBLP:journals/jcss/Smyth78,DBLP:journals/siamcomp/Plotkin76,DBLP:journals/tcs/Winskel85} in domain theory are unable to capture simultaneously\footnote{That is, semantics of recursion cannot be defined by taking a single Kleene least fixed point.}. 

\paragraph{Semantics of open modules and semantic linking.}
A practical compiler should support separate compilation, i.e., different modules or different source code files can be compiled separately first and linked in the end. The correctness of separate compilation is an important research problem which a series of extensions to CompCert such as CompComp \cite{DBLP:conf/popl/StewartBCA15}, CompCertM \cite{DBLP:journals/pacmpl/SongCKKKH20}, CompCertO \cite{DBLP:conf/pldi/KoenigS21} and recent work by \citeauthor{DBLP:journals/pacmpl/ZhangWWKS24} [\citeyear{DBLP:journals/pacmpl/ZhangWWKS24}] tries to tackle. This article proposes a denotational semantics for open
modules, defines a novel semantic linking operator, and shows that this operator preserves the general behavior refinement, {analogous to how the composition preserves the set inclusion} shown in Formula~(\ref{eq:compose_mono}).
\begin{figure}[t!]
    \centering
\begin{tabular}{|c c|}
\hline
    \begin{lstlisting}[language=Coq]
 S1. while $e$ do $s$   $\ttt{/\!/}$ Clight while loop
 T1. block {      $\ttt{/\!/}$ Csharpminor loop
       loop {
         block { if $e$ then $s$ else exit 1;
         } $\ttt{/\!/}$continue of $s$ branches here
       } 
     } $\ttt{/\!/}$break of $s$ branches here
 S2. switch ($e$) {  $\ttt{/\!/}$ Csharpminor switch
       case N1: $s_1$; 
       case N2: $s_2$;
       default: $s$;
     }
     \end{lstlisting} &
     \begin{lstlisting}[language=Coq]
 T2. block {   $\ttt{/\!/}$ Cminor switch
       block {
         block {
           switch ($e$) {
             N1: exit 0;
             N2: exit 1;
             default: exit 2;
           }
         }; $t_1$ $\ttt{/\!/}$from $s_1$ with exits shifted by 2$\;\;$
       }; $t_2$ $\ttt{/\!/}$from $s_2$ with exits shifted by 1 
     }; $t$ $\ttt{/\!/}$from $s$ with exits unchanged
     $\;$
     \end{lstlisting}
     \\
\hline
\end{tabular}
    \caption{ %
      Clight loop \kwd{S1} and Csharpminor switch \kwd{S2} are respectively translated to Csharpminor loop \kwd{T1} and Cminor switch \kwd{T2}, where ({\tts exit} $n$) will prematurely terminate ($n+1$) layers of nested blocks. In program \kwd{T2}, the exiting number in $s_i$  is properly shifted according to the level of blocks it resides in. In Cminor, unlike Csharpminor, the branches of switch statements cannot be general statements but must be \textbf{exit} statements. }
    \label{fig:transl_exam}
    \Description{This figure presents two examples of recursive control-flow translation in CompCert. On the left, a Clight while loop \texttt{while e do s} is translated into a Csharpminor program with nested \texttt{block} and \texttt{loop} constructs, where \texttt{continue} from the loop body branches to the inner block boundary and \texttt{break} branches to the outer block boundary. The same panel also shows a Csharpminor switch statement with two numbered cases and a default branch. On the right, that switch is translated into Cminor as a cascade of three nested blocks surrounding a switch whose branches are only \texttt{exit} statements: case \(N1\) exits 0 levels, case \(N2\) exits 1 level, and the default exits 2 levels. After each exit, the corresponding translated branch body is executed outside the appropriate number of enclosing blocks, so the exit numbers are shifted according to block nesting. The figure illustrates how structured loops and multi-branch switches are compiled into lower-level control flow using nested blocks and exits.}
\end{figure}

Many compilation passes in CompCert are recursively defined transformations, such as, as shown in Fig. \ref{fig:transl_exam}, the transformation of \textbf{while}, \textbf{do .. while} and \textbf{for} loops into infinite loops with appropriate \textbf{block} and \textbf{exit} constructs, along with statements \textbf{break} and \textbf{continue} into early exits, and
the transformation of multi-branch \textbf{switch} from Csharpminor into a Cminor \textbf{switch} wrapped by the statements associated with the various branches in a cascade of nested Cminor blocks.
In our framework, the denotational semantics of a statement is recursively derived from its substructures with certain semantic operators (e.g., the composition ``$\circ$'').
	  For example, the denotation of a sequential statement $(c_1;c_2)$ is defined by the denotations of $c_1$ and $c_2$, and the denotation of a loop statement is defined based on its loop body's.
When verifying the correctness of statement transformations, we employ induction on the syntax of source program statements, which produces proof obligations asserting that for each syntactic component, its corresponding semantic operators preserve transformation correctness. As illustrated by (\ref{eq:compose_mono}), such obligations should be easy to prove by the favorable algebraic properties inherent in denotational semantics.

We have successfully applied our denotation-based verification framework to the front-end of CompCert and typical optimizations on simple prototypes of imperative languages.
When building this framework, we discover that the following two algebraic structures are highly beneficial.
\begin{itemize}
    \item Denotational semantics can be modeled in an extended Kleene Algebra with Tests (KAT) \cite{DBLP:journals/toplas/Kozen97} with operators $\cup$, $\circ$, the empty relation $\varnothing$ and the identity relation $\idrel$ satisfying certain axioms (further discussed in \S\ref{sec:semantics_nocall}). 
Thus, the properties of programs can be carried out in a purely equational subsystem using the axioms of Kleene Algebras.
    \item On top of extended KATs, we propose a notion of \emph{refinement algebra} to unify various behavior refinements and make refinement proofs structured and automated.
    That is, automated proof tactics are designed to mechanically verify behavior refinement for composition statements (such as if-branch, loop, and sequential statements) and module-level compositionality.
\end{itemize}

\paragbf{Main contribution.}  {This article proposes a novel framework to formalize the denotational semantics of first-order imperative languages, and we apply such framework to the compositional verification of realistic compilers.}
Our denotation-based framework is featured with (i) precise denotational semantics for {full} open modules and realistic languages (putting denotational semantics into practical scenarios and achieving better compositionality), 
(ii) a novel notion of refinement algebra for behavior refinement abstraction and proof automation,
and (iii) successful application to the front-end of CompCert and typical optimizations on {first-order imperative} programs.
 The refinement algebra integrates well with Kleene Algebra with Tests (KAT), and works for both the verification of partial program compilation and module-level compositionality. All results of this article are formalized in \Coq (formerly Coq). 
Compared to previous extensions of CompCert, such as CompComp and CompCertM, our \Coq formalization is notably more lightweight in terms of lines of code. 

\begin{figure}[t]
    \centering 
    \begin{tabular}{|l|c|c|c|c|c|c|c|}
        \hline
        \multirow{2}{*}{Languages}
        & \multicolumn{4}{c|}{Toy languages} & \multicolumn{3}{c|}{CompCert front-end} \\
        \cline{2-8} 
        & WHILE & PCALL & GOTO & CFG & Clight & Csharpminor & Cminor
        \\
        \hline
         Section & \S\ref{subsec:while_lang} &
         \S\ref{subsec:standard_pcall} $\sim$
         \S\ref{subsec:pcall_dvg}&
         \S\ref{subsec:goto_sem} &
         \S\ref{sec:semantics_cfg} &
         \multicolumn{3}{c|}{
         \S\ref{subsec:control_flow} $\sim$ 
         \S\ref{subsec:loop_trace}
          \quad and \quad \S\ref{sec:semantics_front}} 
         \\
         \hline
    \end{tabular}
    \caption{Formal languages for demonstrating our methodology.}
    \label{fig:formal_lang}
    \Description{}
\end{figure}

Our presentation starts with smaller toy languages and later extends the ideas to the more elaborate CompCert front-end languages. The languages used in each section are shown in Fig. \ref{fig:formal_lang}.
\paragbf{Structure of the paper.} We briefly introduce CompCert and the challenges we are facing in \S\ref{sec:background}. We then use the WHILE language to illustrate the basic denotational semantics on KATs and extend it to realistic languages in \S\ref{sec:semantics_nocall}. We further take procedure calls into account and build the semantics of open modules and a novel semantic linking operator via the PCALL languages in \S\ref{sec:semantics_pcall}. We introduce the denotation of unstructured branches using the GOTO language in \S\ref{subsec:goto_sem} and that of control flow graphs (CFG) with the CFG language in \S\ref{sec:semantics_cfg}. We present the full semantics of Clight in \S\ref{sec:semantics_front}.

The highlight of this article is developing novel refinement algebras to establish behavior refinement between denotations in \S\ref{sec:refinement}, and applying our theoretical framework to verify the CompCert front-end, a classical CFG generation, and optimizations on toy languages in \S\ref{sec:compose_correctness}. We further compare our approach with typical related efforts in \S\ref{sec:comparison}, followed by code evaluations in \S\ref{sec:evaluation} and further discussion in \S\ref{sec:discussion}. Finally, we discuss related work in \S\ref{sec:related_work} and conclude the article in \S\ref{sec:conclusion}.

\section{Background and Challenges of Compiler Verification}
\label{sec:background}
\subsection{CompCert Verification Framework}
\label{subsec:compcert_bg}

 The small-step semantics of a language $\mathcal{L}$ in CompCert is modeled as a tuple ($\mathcal{I}$, $\hookrightarrow$, $\mathcal{F}$) where $\mathcal{I}$ denotes the set of initial states loaded by a function {\tts Load}: {$\mathcal{L} \rightarrow \ttt{state}$}, $\mathcal{F}$ denotes the set of final states of program execution, and $\hookrightarrow \subseteq \ttt{state} \times \ttt{trace} \times \ttt{state}$ is a stepping relation. For each single step of execution $(\sigma, \tau, \sigma') \in \hookrightarrow$ (written as $\sigma \stackrel{\tau}{\hookrightarrow} \sigma'$), state $\sigma$ transitions to state $\sigma'$ with a finite sequence of triggered observable events (typically, I/O and other system calls). An observable behavior is defined as a (finite or infinite) trace of observable events occurring in a sequence of execution steps.
The CompCert compilation correctness theorem ensures that the target program preserves the observable behavior of the source program and does not introduce any new behavior that wasn't present in the source program. Such compilation correctness is justified by a backward simulation.
We say a program $P_s$ is backward simulated by a program $P_t$ if given a state matching relation $R$, their initial states $\ttt{Load}(P_s)$ and $\ttt{Load}(P_t)$ is related by $R$, and (simplified for presentation)
\begin{align*}
  \forall \sigma_s\ \sigma_t\ \sigma_t'\ \tau,  (\sigma_s, \sigma_t)\in R \Rightarrow \sigma_t \stackrel{\tau}{\hookrightarrow} \sigma_t' \Rightarrow \exists \sigma_s', (\sigma_s', \sigma_t')\in R \wedge \sigma_s \stackrel{\tau}{\hookrightarrow}{\!^{\!*}}\sigma_s'
\end{align*}
 It reads that for every step taken by the compiled target program, there exists zero or more consistent steps in the source program, such that the two programs remain in a related state.
 Despite the general success of operational semantics in CompCert, the simulation technique and {small-step based proofs of CompCert} face the following inconvenience compared to denotational semantics.
 
\begin{figure}[htb]
    \begin{subfigure}[b]{.4\linewidth}
    \centering
    \includegraphics[width=\linewidth]{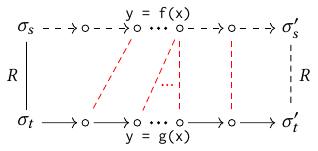}

    \caption{The backward simulation for \kwd{f} and \kwd{g}}
    \label{subfig:fsim-eg1}
    \end{subfigure} 
    \begin{subfigure}[b]{.514\linewidth}
        \flushright
        \includegraphics[width=\linewidth]{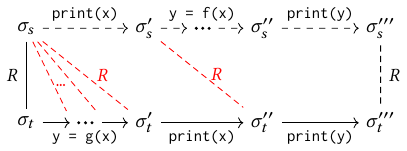}
    \centering
    \caption{The backward simulation for whole programs}
    \label{subfig:fsim-eg2}
    \end{subfigure}   
    \caption{Backward simulation diagrams used in CompCert.}
    \label{fig:f2bsim-eg1}
    \Description{This figure describes two backward-simulation diagrams in CompCert. Subfigure (a) shows a previously established simulation between the internal functions \texttt{f} and \texttt{g}: execution steps of \texttt{g} are matched by corresponding execution steps of \texttt{f} through a sequence of related intermediate states, with no observable events involved. Subfigure (b) shows the larger source and target programs obtained after placing \texttt{f} and \texttt{g} into surrounding code with \texttt{print} statements. The observable outputs introduced by the reordered \texttt{print} operations change the simulation structure, so the intermediate correspondence used for \texttt{f} and \texttt{g} cannot be reused directly for the whole programs. The figure illustrates why backward simulation in CompCert is sensitive to surrounding context and observable events.}
\end{figure}

\paragbf{Intermediate correspondence.}
Consider the following transformation, where \kwd{x} and \kwd{y} are temporary variables,
\kwd{f} and \kwd{g} are pure internal functions, which return identical results for the same input and generate no observable events, and the built-in function \kwd{print} outputs the value of its argument.
{\begingroup
\allowdisplaybreaks[0]
\begin{align*}
  &\text{Source program } S:\ \ \kwd{print(x); y = f(x); print(y);} \\
  &\text{Transformed to } T: \ \ \kwd{y = g(x); print(x); print(y);}
\end{align*}
\endgroup
}
\noindent
Suppose that a (backward) simulation between \kwd{f} and \kwd{g} has been established in advance, where nontrivial intermediate correspondence may be constructed by a series of proof steps. In CompCert, however, these proof steps cannot be directly reused to establish simulation relations between the larger programs $S$ and $T$, since the rearrangement of \kwd{print} statements introduces observable events that alter the simulation structure, as illustrated in Fig.~\ref{fig:f2bsim-eg1}. Consequently, subroutine equivalence proofs must be revisited when establishing correspondence of larger programs.
In comparison, if using denotational semantics, the established semantic equivalence between \kwd{f} and \kwd{g} could potentially be reused in proving the equivalence of the programs $S$ and $T$.

\paragbf{Nondeterminism correspondence}
More importantly, it is hard to use CompCert’s (backward) simulation techniques to verify transformations
where (demonic) nondeterminism straddles event synchronization and thereby impacts later observable behavior. Consider the following example: 
       \begin{align*}
        &\text{Source program: } \ \ \ttt{choice(x = 0, x = 1); print(0); print(x);} \\
        &\text{Transformed to:  }\ \ \ \ttt{print(0); choice(x = 0, x = 1); print(x);}
       \end{align*}
    where the statement ``{\tts choice}'' nondeterministically chooses to execute either ``{\tts x = 0}'' or ``{\tts x = 1}''. 
    The transformation swaps the first two statements while preserving the set of possible observable behaviors (i.e., traces \kwd{00} and \kwd{01}). Intuitively, the denotational semantics of the programs before and after the transformation are equivalent. However, establishing this equivalence is difficult using CompCert’s current simulation theories.
    Because CompCert simulation theories constrain a one-to-one match of observable events (i.e., strict event synchronization), the observable event produced by the first step of the target program must correspond to the observable event produced by the first two steps of the source program. At this point, the source program must first choose to execute either ``{\tts x = 0}'' or ``{\tts x = 1}'', and whichever is chosen, the subsequent choice in the target program may differ, potentially leading to inconsistent results for ``{\tts print(x)}'' in the end.

\paragbf{Program-state correspondence}
 CompCert models the semantics of high-level languages using {mixed program states} integrating continuations (to represent control flow and nested calls) together with 
the definitions of statements and functions currently under execution. 
As a result, establishing correspondences between program states necessarily involves 
all components of these enriched states. 
In comparison,
the standard denotational semantics of a statement~$c$ with procedure calls 
(as presented, for example, in Winskel’s textbook~\cite{DBLP:books/daglib/0070910}) 
interprets~$c$ as a function mapping the behaviors of callees to the overall behavior of~$c$. 
This naturally leads to a formulation of transformation correctness: 
a transformation~$\mathcal{T}$ is correct if it preserves such semantic mappings.
\begin{align*}
&\bracket{c}: \text{mapping from ``behaviors of source callee procedures'' to ``behavior of }c\text{''}
\\
&\bracket{\mathcal{T}(c)}:  \text{mapping from ``behaviors of compiled callee procedures'' to ``behavior of }\mathcal{T}(c)\text{''}
\\
&\mathcal{T}\text{ is correct iff for any callee behaviors }  \chi_s\ \chi_t,
\text{ if }\chi_t \sqsubseteq \chi_s\text{, then }\bracket{\mathcal{T}(c)}(\chi_t) \sqsubseteq \bracket{c}(\chi_s)
\end{align*}
This formalization benefits from two aspects: (i) we can simplify the definition of program states and their correspondence, and further decompose the refinement proof into smaller and simpler conclusions by splitting the original correspondence of CompCert into static environment correspondence and dynamic memory correspondence; (ii) we can define semantic linking operators by taking semantic fixed points, and we further find that the equivalence between semantic linking and syntactic linking is reduced to the concise Beki\'{c}’s theorem \cite{DBLP:conf/ibm/Bekic84e} about fixed points.

To sum up,
the above statements demonstrate the {potential benefits} of denotational semantics for composing proofs, handling nondeterministic correspondence, and lightening program-state correspondence.
Moreover,
we also observe that recent research in simulation theories, such as FreeSim~\cite{FreeSim}, as well as more advanced approaches like prophecy variables~\cite{DBLP:journals/tcs/AbadiL91,DBLP:journals/iandc/LynchV95}, proposes mechanisms intended to mitigate these challenges. Some of these approaches, such as FreeSim, have been formalized and partially integrated into existing proof frameworks, and some of them have not yet been formalized and used for realistic compiler verification like prophecy variables. Particularly, they face their own inherent difficulties in tackling the issues discussed above. We will examine these related works in greater detail in \S\ref{sec:related_work}.

\subsection{Denotational Semantics: Achievements and Challenges}
\label{subsec:denote_bg}
Denotational semantics enjoys good compositionality, meaning that the semantics of a syntactic construct can be systematically derived from the semantics of its substructures using semantic operators. A central concern in this approach is handling loops and recursion effectively. Existing research on deterministic programs demonstrates that the semantics of recursive constructs can be defined on the basis of domain theory \cite{scott1970outline,scott1971toward} or Kleene Algebras \cite{DBLP:journals/toplas/Kozen97}. In domain theory, the semantics of programs are modeled as elements of domains (typically CPOs), and the key tool in domain theory is the Kleene fixed-point theorem where recursive constructs are given meaning by finding the least fixed point of a continuous function. The concepts of CPOs, continuous functions and Kleene fixed points are outlined as follows.

\begin{definition} [Complete Partial Order]
    A \emph{Directed Complete Partial Order (DCPO)} is a partially ordered set \( P \) such that every \emph{directed} subset \( D \subseteq P \) has a supremum (i.e. \(\sup D\)) in \( P \). A subset \( D \) of \( P \) is directed if it is nonempty and for any elements \( x, y \in D \), there exists an element \( z \in D \) such that \( x \leq z \) and \( y \leq z \). A DCPO with a least element \(\bot\) is usually called a \emph{pointed} DCPO. 
    
\end{definition}
\begin{definition} [Scott-Continuous Function]
    A function \( f \) between two DCPOs is called \emph{Scott-continuous} if it preserves the directed suprema, i.e., for any directed subset \( D \) of the domain,
    \[ f(\sup D) = \sup f(D)\]
\end{definition}

\begin{theorem} [Kleene fixed point]
\label{thm:kleene_fixpoint}
Suppose $(A,\leq_A)$ is a directed-complete partial order (DCPO) with a least element $\bot$, and let $F: A \rightarrow A$ be a Scott-continuous (and therefore monotone) function. Then $F$ has a least fixed point $\mu F = \sup(\{F^{n}(\bot) \mid n \in \mathbb{N}\})$, where $\sup$ means taking the supremum.
\end{theorem}
Kleene algebras are algebraic structures with operators $\circ$, $\cup$, $^*$, $\emptyset$, $\idrel$ that satisfy certain axioms. Typically, the family of binary relations with operators $\cup$, $\circ$, the empty relation $\varnothing$, the identity relation $\idrel$ and the reflexive transitive closure $^*$ constitute a Kleene Algebra, where iteration is captured algebraically through the star operator.
The following example shows the definition of the textbook denotational semantics on a toy language, by domain theory or by Kleene Algebras.
\begin{example} \label{ex:domain_vs_ka}
    Consider a simple deterministic and type-safe toy language whose syntactic constructs (extremely simplified for presentation purpose) and semantics can be defined as follows.
\begin{align*}
&\begin{aligned}
   c \triangleq \textbf{skip}\ |\ c_1; c_2\ |\ \textbf{if}\ b\ \textbf{then}\ c_1\ \textbf{else}\ c_2\
                                   |\ \textbf{while}\ b\ \textbf{do}\ c   
\end{aligned} \\
&\begin{aligned}
\bracket{\textbf{skip}} \triangleq
  \idrel
\quad
\bracket{c_1; c_2} \triangleq
  \bracket{c_1} \circ \bracket{c_2}
\quad
 \textbf{test}(X) \triangleq \{ (\sigma, \sigma) \mid \sigma \in X\}\text{, for } X \subseteq \ttt{state}
\end{aligned}
\\
& \begin{aligned}
    \llbracket \textbf{if}\ b\ \textbf{then}\ c_1\ \textbf{else}\ c_2 \rrbracket \triangleq
    \mathbf{test}(\bracket{b}.\ttt{(tts)}) \circ
      \llbracket c_1 \rrbracket \cup
    \mathbf{test}(\bracket{b}.\ttt{(ffs)}) \circ
      \llbracket c_2 \rrbracket  
\end{aligned}
\\
&\begin{aligned}
    \llbracket \textbf{while}\ b\ \textbf{do}\ c \rrbracket \triangleq
      \mu x. \textbf{test}(\bracket{b}.\ttt{(ffs)}) \cup
      \big(\textbf{test}(\bracket{b}.\ttt{(tts)}) \circ
      \bracket{c} \circ x\big), \text{ by domain theory}
\end{aligned}
\\
&\begin{aligned}
    \llbracket \textbf{while}\ b\ \textbf{do}\ c \rrbracket &\triangleq
      \big(\textbf{test}(\bracket{b}.\ttt{(tts)}) \circ
      \bracket{c}\big)^* \circ \textbf{test}(\bracket{b}.\ttt{(ffs)}),  \text{ by Kleene Algebras} \\
    \text{where } \fof{}{b}{tts} &\triangleq \{\sigma \mid \sigma(b) = \textbf{T}\},\ 
    \fof{}{b}{ffs} \triangleq \{\sigma \mid \sigma(b) = \textbf{F}\},\ 
    \idrel = \{(\sigma, \sigma)\mid \sigma \in \ttt{state}\}
\end{aligned}
\end{align*}
\end{example}

Although domain theory and Kleene algebras have a long research history and a strong mathematical foundation, it remains challenging to precisely characterize the full semantics of practical programming languages, which may exhibit termination, divergence, or even abortion.
We may naively construct a domain by treating abortion as the least element in the partial order and regarding normal termination and divergence as incomparable (i.e., without any ordering between them). However, in a nondeterministic program, a state may nondeterministically transition to either a terminating or a diverging state. If these outcomes are incomparable, the semantics cannot represent their “join” (supremum) in a manner consistent with the domain’s ordering. This inconsistency prevents the definition of a continuous function required by the Kleene fixed-point theorem.
In order to handle nondeterministic and possible nonterminating behaviors\footnote{Unless otherwise specified, nondeterminism means demonic nondeterminism.}, computer scientists proposed three classical power domain constructions, known as Hoare powerdomain, Smyth powerdomain \cite{DBLP:journals/jcss/Smyth78}, and Plotkin powerdomain \cite{DBLP:journals/siamcomp/Plotkin76}, in the 1970s.
Power domains in our context are constructed from a flat domain $(D, \leqslant_D)$, which may stand for possible results of executions. Suppose that for any $x$ and $y$ in $D$,
  $$ x \leqslant_D y \Leftrightarrow (x = \bot \vee x = y)$$
 a power domain of $D$, written as $\mathcal{M}(D)$, involves the subsets of the ground domain $D$ and a new partial order $\sqsubseteq$, such that ($\mathcal{M}(D)$, $\sqsubseteq$) forms a CPO and one can take a fixed point over it. For any $X$ and $Y$ in $\mathcal{M}(D)$, there are generally three natural ways to construct the new ordering:
\begin{align*} 
  & X \sqsubseteq_0 Y \Leftrightarrow \forall x \in X, \exists y \in Y, x \leqslant_D y \\
  & X \sqsubseteq_1 Y \Leftrightarrow \forall y \in Y, \exists x \in X, x \leqslant_D y \\
  & X \sqsubseteq_2 Y \Leftrightarrow 
   X \sqsubseteq_0 Y  \wedge X \sqsubseteq_1 Y 
\end{align*}

\paragbf{The Hoare powerdomain.} The Hoare powerdomain is established with the first ordering $\sqsubseteq_0$ meaning that everything $X$ can do, $Y$ can do better. However, this ordering is just a preorder but not a partial order since it fails antisymmetry.
The issue is solved by defining an equivalence relation $\simeq_0$ ($X \simeq_0 Y$ iff $X \sqsubseteq_0 Y \text{ and } Y \sqsubseteq_0 X$) and restricting that every set $X$ in $\mathcal{M}(D)$ is a \emph{downset}. Specifically,
  $$ X \text{ is a \emph{downset} iff } X \simeq_0\ \downarrow X, \text{ where } \downarrow X \triangleq \{x_0\in D \mid \exists x\in X, x_0 \leqslant_D x\}$$
 Consequently, with the Hoare powerdomain, one gets an angelic semantics for nondeterminism since the downsets always include the least element $\bot$. That means, any program that could either terminate or diverge will have the same meaning with its terminating behavior. For instance, given the following four programs, program 1 and 3 will have the same semantics under this power domain (where``{\tts x = ?}'' in program 4 means ``nondeterministically set {\tts x} to any natural number'').

  \begin{tabular}{ll}
          \begin{lstlisting} [language = Coq]
            1. $\textbf{skip}$
            2. $\textbf{while}$ true $\textbf{do skip}$
            3. $\textbf{choice}$($\textbf{while}$ true $\textbf{do skip}$, $\textbf{skip}$) 
          \end{lstlisting}   & 
          \begin{lstlisting} [language = Coq]
            4. $\textbf{while}$ (y $==$ 2 $||$ x > 0) $\textbf{do}$
                 $\textbf{if}$ (y $==$ 2) $\textbf{then}$ y = 1; x = $\ttt{?}$
                 $\textbf{else}$ x = x - 1
          \end{lstlisting}
  \end{tabular}
\paragbf{The Smyth powerdomain.} The Smyth powerdomain is obtained with the preorder $\sqsubseteq_1$, which says that everything $Y$ can do is approximated by some behavior of $X$. Since this ordering is still a preorder but not a partial order, the Smyth powerdomain defines the preorder's natural equivalence $\simeq_1$ ($X \simeq_1 Y$ iff $X \sqsubseteq_1 Y \text{ and } Y \sqsubseteq_1 X$) and constrains every element $X$ in $\mathcal{M}(D)$ to be an \emph{upset}, i.e., 

   $$ X \text{ is an \emph{upset} iff } X \simeq_1\ \uparrow X, \text{ where } \uparrow X \triangleq \{x_0\in D \mid \exists x\in X, x \leqslant_D x_0\}$$  
As a result, one will achieve a demonic semantics in the sense that
 any program that can diverge has the semantics $\uparrow \{\bot\} = D$. Thus, for instance, programs 2 and 3 will have the same semantics. 

\paragbf{The Plotkin powerdomain.} The Plotkin powerdomain uses the well-known \emph{Egli-Milner ordering} $\sqsubseteq_2$, with which every set $X$ in $\mathcal{M}(D)$ is restricted to be equivalent with its \emph{convex} closure:
 $$ conv (X) \triangleq \{x_1\in D \mid \exists x_0, x_2\in X, x_0 \leqslant_D x_1 \leqslant_D x_2\}$$
By this powerdomain, all the first three programs will have different meanings as one would expect. However, some constructs like program 4 that can produce infinitely many different results and yet be certain to terminate, i.e., unbounded nondeterminism, are excluded (see \citeN{DBLP:journals/tcs/Back83} for details). 

To sum up, power domains often fail to capture the full range of possible program outcomes. Such inability is known as an intrinsic defect of domain theory. 

\paragbf{Park's approach} \citeN{DBLP:conf/ac/Park79} proposed an alternative method to characterize the possible nonterminating behavior of nondeterministic programs, by introducing an additional set to capture nonterminating behaviors.
However, when observable events are considered, it is still challenging to explicitly distinguish reactive divergence and silent divergence for realistic compiler verification.

 In a word, a notable gap exists between denotational semantics and realistic programming languages. %
Additionally, practical programs often include structured (e.g., \textbf{break}, \textbf{continue}) and even unstructured branches (e.g., \textbf{goto}), which may prematurely terminate the execution of subsequent statements and hence disrupt the elegant property shown in Formula~(\ref{eq:compose_mono}). Precisely capturing the full meaning of realistic programs while still keeping good compositionality turns out to be a big challenge.
We show the solution of this problem in \S \ref{sec:semantics_nocall} and then extend it to open modules in \S \ref{sec:semantics_pcall}.

\section{Denotational Semantics of Statements (No Procedure Calls)}
\label{sec:semantics_nocall}

\subsection{The Formal Framework for Defining Denotations}
\label{subsec:while_lang}
Our solution is to separately define different program behaviors of interest using different sets, with each behavior set defined by corresponding semantic operators in the form of an extended Kleene Algebra with Tests (KAT).
This manipulation is somewhat similar to Park's approach. In comparison, we design new mechanisms to handle divergence with event traces, capture the full meaning of realistic programs, and especially concisely formalize them in \Coq.

Consider a type-safe toy language WHILE whose set of statements {\tts Com}, ranged over by $c$, is parameterized on the set of Boolean expressions {\tts Exp} ranged over by $b$ (no side effects):
$$ c \triangleq \textbf{skip}\ |\ atom \ |\ \textbf{choice} (c_1, c_2)\ |\ c_1; c_2\ |\ \textbf{if}\ b\ \textbf{then}\ c_1\ \textbf{else}\ c_2\ |\ \textbf{while}\ b\ \textbf{do}\ c  
$$
Nondeterminism can be introduced either by atomic statements (ranged over by $atom$) like nondeterministic assignments or by \textbf{choice} statements that unpredictably choose $c_1$ or $c_2$ to execute.
We formalize the denotation of statements and Boolean expressions with the following signature.
\begin{lstlisting}[language=Coq]
    Record CDenote: Type := {                          
      nrm: state -> state -> Prop;            (* nrm $\subseteq$ state $\times$ state *)
      dvg: state -> Prop                     (* dvg $\subseteq$ state *)
    }.
    Record BDenote: Type := {
      tts: state -> Prop;                    (* tts $\subseteq$ state *)
      ffs: state -> Prop                     (* ffs $\subseteq$ state *)
    }.
\end{lstlisting}
Specifically, $(\sigma_1, \sigma_2) \in \bracket{c}.(\ttt{nrm})$ iff executing $c$ from initial state $\sigma_1$ could terminate on state $\sigma_2$, and $\sigma \in \bracket{c}.(\ttt{dvg})$ iff executing $c$ from $\sigma$ could diverge.
An element $\sigma \in  \bracket{b}$.({\tts tts}) iff state $\sigma$ satisfies $b$, and $\sigma \in  \bracket{b}$.({\tts ffs}) iff state $\sigma$ does not satisfy $b$.
We can use set operators (like union) and relation operators (like composition) to define program semantics as follows.
\begin{align*}
&\begin{aligned}
&\llbracket \textbf{skip} \rrbracket.(\ttt{nrm}) \triangleq
  \idrel \quad
    \llbracket \textbf{skip} \rrbracket.(\ttt{dvg}) \triangleq \mathbf{\varnothing} \; \; \; \!
&\llbracket \textbf{choice} (c_1, c_2) \rrbracket.( \ttt{nrm}) \triangleq
  \llbracket c_1 \rrbracket.(\ttt{nrm}) \cup
  \llbracket c_2 \rrbracket.(\ttt{nrm}) \\
&\llbracket c_1; c_2 \rrbracket.(\ttt{nrm}) \triangleq
  \llbracket c_1 \rrbracket.(\ttt{nrm}) \circ \llbracket c_2 \rrbracket.(\ttt{nrm})
&\llbracket \textbf{choice} (c_1, c_2) \rrbracket.(\ttt{dvg}) \triangleq
  \llbracket c_1 \rrbracket.(\ttt{dvg}) \cup
  \llbracket c_2 \rrbracket.(\ttt{dvg})
\end{aligned}
\\
& \begin{aligned}
    \llbracket \textbf{if}\ b\ \textbf{then}\ c_1\ \textbf{else}\ c_2 \rrbracket.(\ttt{nrm}) \triangleq
    \mathbf{test}(\bracket{b}.\ttt{(tts)}) \circ
      \llbracket c_1 \rrbracket.(\ttt{nrm}) \cup
    \mathbf{test}(\bracket{b}.\ttt{(ffs)}) \circ
      \llbracket c_2 \rrbracket.(\ttt{nrm})  
\end{aligned} \\
&\begin{aligned}
    \llbracket \textbf{while}\ b\ \textbf{do}\ c \rrbracket.(\ttt{nrm}) \triangleq
      \big(\textbf{test}(\bracket{b}.\ttt{(tts)}) \circ
      \bracket{c}.(\ttt{nrm})\big)^* \circ \textbf{test}(\bracket{b}.\ttt{(ffs)}),
\end{aligned}
\\
&\begin{aligned}
      \ \ \text{where} \quad &\textbf{test}(X) \triangleq \{ (\sigma, \sigma) \mid \sigma \in X\}\text{, for } X \subseteq \ttt{state}
    &\text{identity relation on } X \\
    &\idrel \triangleq \{ (\sigma,\sigma) \mid \sigma \in \ttt{state}\} 
    &\text{identity relation on } \ttt{state} \\
    &R^* \triangleq \idrel \cup R \cup (R \circ R) \cup \cdots \cup R^n \cup \cdots
    &\text{reflexive transitive closure}\\ 
    &R^+ \triangleq R \cup (R \circ R) \cup \cdots \cup R^n \cup \cdots
    &\text{transitive closure}\\    
    &R^\infty \triangleq R \circ R \circ R \circ \cdots
    &\text{infinitely iterated composition}\footnotemark
\end{aligned}
\end{align*}
As is well known, the family of binary relations with operators $\cup$, $\circ$, the empty relation $\varnothing$, the identity relation $\idrel$ and the reflexive transitive closure $^*$ constitute a Kleene Algebra.
Moreover, if we overload the ``$\circ$'' operator as follows,
\footnotetext{Formally, an element $\sigma_0 \in R^\infty$ iff \( \exists \sigma_1 ...\sigma_n..., \text{ s.t. }  \forall i, (\sigma_i, \sigma_{i+1}) \in R\), where \(\sigma_0 \sigma_1 ...\sigma_n...\) is an infinite sequence of states.} %

\begin{align*}
    &R_1 \circ R_2 \triangleq
    \{ \left(\sigma_1, \sigma_3\right) \mid 
    \exists \sigma_2,
      (\sigma_1, \sigma_2) \in R_1
    \wedge
      (\sigma_2, \sigma_3) \in R_2
    \} & \text{original}
\\
    &R \circ X\triangleq 
    \{ \sigma_1 \mid \exists \sigma_2,
      (\sigma_1, \sigma_2) \in R
    \wedge
    \sigma_2 \in X
    \}  &\text{overloaded}
\end{align*}
these operators form a typed Kleene Algebra \cite{kozen1998typed} with two base types.
 The diverging behavior of sequential, \textbf{if}, and \textbf{while} statements can then be formalized as follows.
\begin{align*}
&\llbracket c_1; c_2 \rrbracket.(\ttt{dvg}) \triangleq
  \llbracket c_1 \rrbracket.(\ttt{dvg}) \cup
    \big(\llbracket c_1 \rrbracket.(\ttt{nrm}) \circ \llbracket c_2 \rrbracket.(\ttt{dvg})\big) \\
& \begin{aligned}
    \llbracket \textbf{if}\ b\ \textbf{then}\ c_1\ \textbf{else}\ c_2\rrbracket.(\ttt{dvg}) \triangleq
    \big(\mathbf{test}(\bracket{b}.\ttt{(tts)}) \circ
      \llbracket c_1 \rrbracket.(\ttt{dvg})\big) \cup
    \big(\mathbf{test}(\bracket{b}.\ttt{(ffs)}) \circ
      \llbracket c_2 \rrbracket.(\ttt{dvg})\big)
\end{aligned}
\\
&\begin{aligned}
    \llbracket \textbf{while}\ b\ \textbf{do}\ c \rrbracket.(\ttt{dvg}) \triangleq \,
    &\big(\mathcolor{dkgreen}{(}\mathbf{test}(\bracket{b}.\ttt{(tts)}) \circ \bracket{c}.(\ttt{nrm})\mathcolor{dkgreen}{)}^*
    \circ
    \mathcolor{dkgreen}{(}\mathbf{test}(\bracket{b}.\ttt{(tts)}) \circ \bracket{c}.(\ttt{dvg})\mathcolor{dkgreen}{)}\big)
    \ \cup \\
      &\big(\mathbf{test}(\bracket{b}.\ttt{(tts)}) \circ \bracket{c}.(\ttt{nrm})\big)^\infty
\end{aligned}
\end{align*}
The interesting case is that the \textbf{while} statement diverges if and only if either the loop body $c$ diverges after a finite number of normally executing $c$, or $c$ is executed an infinite number of times.

The above style of modeling denotational semantics based on a typed Kleene Algebra can be naturally scaled to realistic settings by extended denotations and redefined operators.
In CompCert, the behavior of a C program encompasses not only state transitions but also the observable events (I/O and other system calls) triggered during execution. Therefore, we can formalize the denotation of C-like program statements with the following signature, where {\tts trace} is the set of all finite sequences of observable events, and {\tts itrace} the set of all infinite sequences of observable events.  \label{Clight_denote}
\begin{lstlisting}[language=Coq]
    Record Denote: Type := {
      nrm: state -> trace -> state -> Prop;      (* nrm $\subseteq$ state $\times$ trace $\times$ state *)
      ...                                               
      err: state -> trace -> Prop;              (* err $\subseteq$ state $\times$ trace            *)
      fin_dvg: state -> trace -> Prop;          (* fin_dvg $\subseteq$ state $\times$ trace$\;\,$  *)
      inf_dvg: state -> itrace -> Prop          (* inf_dvg $\subseteq$ state $\times$ itrace *)
    }.
\end{lstlisting}
Specifically, the terminating behavior here is extended to a ternary relation which additionally records the event trace generated during execution.
Diverging behavior is divided into finite and infinite parts for distinguishing behaviors with finite events and infinite events.
In addition, we use a set {\tts err} to capture the aborting behavior of programs.
We then redefine operators $\circ$, \textbf{test}, and $\idrel$ (and therefore the operator for reflexive transitive closure) as follows, where $R \subseteq \ttt{state} \times \ttt{trace} \times \ttt{state}$, $Y \subseteq \ttt{state} \times ( \ttt{trace} \cup \ttt{itrace})$, ``{\tts nil}'' is the empty sequence and “$\cdot$” is the concatenation of sequences.
\begin{align*}
  & \idrel \triangleq \{(\sigma, \ttt{nil}, \sigma) \mid \sigma \in \ttt{state} \} \quad \quad \ \ 
  \textbf{test}(X) \triangleq \{(\sigma, \ttt{nil}, \sigma) \mid \sigma \in X \}, \text{ for } X \subseteq \ttt{state}
  \\
  &R \circ Y \triangleq 
    \{ (\sigma_1, \tau) \mid
    \exists \sigma_2\ \tau_1\ \tau_2,
      (\sigma_1, \tau_1, \sigma_2) \in R
    \wedge
    (\sigma_2, \tau_2) \in Y
    \wedge
    \tau = \tau_1 \cdot \tau_2
    \}
  \\
  &\begin{aligned}
    R_1 \circ R_2 \triangleq
    \{ \left(\sigma_1, \tau, \sigma_3\right) \mid 
    \exists \sigma_2\ \tau_1\ \tau_2, \
    (\sigma_1, \tau_1, \sigma_2) \in R_1
    \wedge
    (\sigma_2, \tau_2, \sigma_3) \in R_2
    \wedge
    \tau = \tau_1 \cdot \tau_2
    \}
  \end{aligned}
\end{align*}

In summary, we use various behavior sets to define different behaviors of interest with typed operators, thus solving the limitations of traditional powerdomains in compiler verification scenarios. In our implementation, these sets operators are provided by a \Coq set library~\cite{DBLP:sets/cao}.
Besides, we address two additional challenges for practical languages: supporting various control flow constructs (in \S\ref{subsec:control_flow}) and defining program divergence in the presence of event traces (in \S\ref{subsec:loop_trace}).

\subsection{Semantics of Control Flow Constructs}
\label{subsec:control_flow}
For handling different control flow statements like \textbf{break} and \textbf{continue} in Clight, or \textbf{block} and \textbf{exit} in Csharpminor, we add corresponding fields of sets to the denotation record.
\paragbf{Clight} The Clight language provides control flow constructs including \textbf{break}, \textbf{continue}, and \textbf{return} statements to prematurely terminate the normal execution of statements. To capture the behavior of these program constructs, we add the following sets to {\tts Denote}.
\begin{lstlisting}[language=Coq]
    brk: state -> trace -> state -> Prop;               (* brk $\subseteq$ state $\times$ trace $\times$ state *)
    ctn: state -> trace -> state -> Prop;               (* ctn $\subseteq$ state $\times$ trace $\times$ state *)
    rtn: state -> trace -> state -> val -> Prop.         (* rtn $\subseteq$ state $\times$ trace $\times$ state $\times$ val *)
\end{lstlisting}
These sets represent that for any element $(\sigma_0, \tau, \sigma_1)$ in $\bracket{s}.$({\tts brk}) or $\bracket{s}.$({\tts ctn}), the execution of Clight statement $s$ from $\sigma_0$ could eventually reach state $\sigma_1$ and then exit because of a \textbf{break} or \textbf{continue} statement, producing a sequence of observable events $\tau$. In other words,  once a \textbf{break} or \textbf{continue} statement is encountered, the execution will end prematurely, and the corresponding set {\tts brk} or {\tts ctn} records how it ends. Therefore, we formalize such semantics as follows (selectively listed).
\begin{align*}
    &\begin{aligned}
    &\fof{}{\textbf{break}}{brk} \triangleq \idrel, \text{ and other fields are assigned the empty set}. \\
    &\fof{}{\textbf{continue}}{ctn} \triangleq \idrel, \text{ and other fields are assigned the empty set} \\
    &\fof{}{s_1;s_2}{brk} \triangleq \fof{}{s_1}{brk} \cup \big(\fof{}{s_1}{nrm} \circ \fof{}{s_2}{brk}\big) \\ 
    &\fof{}{s_1;s_2}{ctn} \triangleq \fof{}{s_1}{ctn} \cup \big(\fof{}{s_1}{nrm} \circ \fof{}{s_2}{ctn}\big)
    \end{aligned}
\end{align*}
The semantics of exiting by a \textbf{return} statement can be defined similarly since the {\tts rtn} set has a similar meaning but additionally with a return value (where {\tts val} is the set of values ranging over integers, floats, pointers, as well as an undefined value {\tts Vundef}).
 Clight uses one special loop (\textbf{loop} $s_1$ $s_2$) to encode all three kinds of C loops. Its semantics is to repeatedly and sequentially execute statements $s_1$ and $s_2$ until a certain break or return statement, and a ``continue jump'' in $s_1$ will branch to $s_2$.
 Therefore, the \textbf{for} statement of C language for example can be defined as:
\[ 
    \textbf{for}(s_1; b; s_2)\ \{s_3\} \triangleq s_1;\ \big(\textbf{loop } \big((\textbf{if } b \textbf{ then skip else} \textbf{ break});s_3\big)\ s_2\big)
\]
We then formalize the terminating behavior of Clight loops as follows, where $N_1$ denotes the behavior that the loop body $s_1$ either ends normally or prematurely due to a continue statement in $s_1$, and $N_{12}$ denotes the behavior of a ``sequential'' execution of loop bodies $s_1$ and then $s_2$.
\begin{align}
& N_1 = \fof{}{s_1}{nrm} \cup \fof{}{s_1}{ctn}
  \quad \quad N_{12} = N_1 \circ
    \fof{}{s_2}{nrm} \label{formula:n1_n2_def}
\\
&\begin{aligned}
  \fof{}{\textbf{loop }s_1 \ s_2}{nrm} \triangleq
    &\ N_{12}^* \circ \big(\fof{}{s_1}{brk}\cup (N_1 \circ
    \fof{}{s_2}{brk})\big)
\end{aligned}
\end{align}

\paragbf{Csharpminor} The control flow of Csharpminor is structured with the \textbf{block} and \textbf{exit} statement. For instance, the execution of statement (block \{ block \{ $u_1$; exit(1)\}; $u_2$\}; $u_3$) is equivalent to the execution of ($u_1; u_3$) since the statement \textbf{exit}($n$) prematurely terminates the execution of the ($n$+1) layers of nested \textbf{block} statements.
We then use a function mapping from natural numbers to relations to capture this behavior, so that an element 
($\sigma_0, \tau, \sigma_1) \in$ $\fof{}{u}{blk}_n$ if and only if executing the Csharpminor statement $u$ from $\sigma_0$ could terminate at $\sigma_1$, and $n$ records the layers of nested blocks to exit.
Thus, the denotation of Csharpminor statements $\bracket{u}$ satisfies the following equations.
\begin{align*}
    &\begin{aligned}
      \fof{}{\textbf{block}\{u\}}{nrm} \triangleq
      \fof{}{u}{nrm} \cup \fof{}{u}{blk}_0
        \quad \quad \,
      \fof{}{\textbf{block}\{u\}}{blk}_n \triangleq
      \fof{}{u}{blk}_{n+1}
    \end{aligned}
    \\
    &\begin{aligned}
      \fof{}{\textbf{exit}(n_0)}{blk}_n \triangleq \idrel
      \text{, if } n_0 = n 
      \quad \quad \quad \quad
      \quad \quad 
      \fof{}{\textbf{exit}(n_0)}{blk}_n \triangleq \varnothing \text{, if } n_0 \neq n 
    \end{aligned}
    \\
    &\begin{aligned}
        \fof{}{u_1;u_2}{blk}_n \triangleq \fof{}{u_1}{blk}_n \cup \big(\fof{}{u_1}{nrm} \circ \fof{}{u_2}{blk}_n\big)
    \end{aligned}
\end{align*}
The loop structure of Csharpminor is an infinite loop (\textbf{Sloop} $u$), whose behavior is to execute the loop body $u$ repeatedly until a certain early exit. The semantics is defined as follows.
\begin{align*}
    \fof{}{\textbf{Sloop } u}{nrm} \triangleq \varnothing \quad \quad 
    \fof{}{\textbf{Sloop } u}{blk}_n \triangleq \fof{}{u}{nrm}^* \circ \fof{}{u}{blk}_n
\end{align*}

\subsection{Divergence with Event Trace}
\label{subsec:loop_trace}
Handling divergence with event traces is more involved. Previous research indicates that simulations between semantics have to address the classic stuttering problem \cite{DBLP:journals/jar/Leroy09}. Therefore, unlike defining divergence directly using $R^\infty$ in simple languages, we must explicitly consider whether the event trace in program execution is empty or not. We demand that $R^\infty$ is only defined if either all traces in $R$ are nonempty (i.e., nonsilent), or all traces in $R$ are empty (i.e., silent).
To this end, we use silent operator $\triangle$ and nonsilent operator $\blacktriangle$ to explicitly filter silent and nonsilent event traces as follows, where $R \subseteq \ttt{state} \times \ttt{trace} \times \ttt{state}$ and $Y \subseteq \ttt{state} \times (\ttt{trace} \cup \ttt{itrace})$.
\begin{align*}
  \blacktriangle R& \triangleq \{ (\sigma, \tau, \sigma') \mid (\sigma, \tau, \sigma') \in R \wedge \tau \neq \text{\tts nil} \} \\
  \triangle R& \triangleq \{ (\sigma, \tau, \sigma') \mid (\sigma, \tau, \sigma') \in R \wedge \tau = \text{\tts nil} \} \\
  \triangle Y& \triangleq \{ (\sigma, \tau) \mid (\sigma, \tau) \in Y \wedge \tau = \text{\tts nil} \}
\end{align*}
Recall that $\fof{}{s}{fin\_dvg}$ and $\fof{}{s}{inf\_dvg}$ denote diverging behavior with finite and infinite event traces generated by executing statement $s$, respectively.
For Clight loops, $\fof{}{\textbf{loop }s_1 \ s_2}{fin\_dvg}$ indicates that after sequentially executing the loop bodies $s_1$ and $s_2$ a finite number of times, a Clight loop could silently diverge since $s_1$ does, or $s_2$ does, or the entire loop itself silently diverges.
Similarly, the nonsilently diverging behavior (i.e, reacting behavior) $\fof{}{\textbf{loop }s_1 \ s_2}{inf\_dvg}$ indicates that after sequentially executing the loop bodies $s_1$ and $s_2$ a finite number of times, a Clight loop could nonsilently diverge since $s_1$ does, or $s_2$ does, or the entire loop itself nonsilently diverges.
They are formalized as follows, where $N_1$ and $N_{12}$ are defined by Formula~(\ref{formula:n1_n2_def}).
\begin{align*}
  &\begin{aligned}
    \fof{}{\textbf{loop }s_1 \ s_2}{fin\_dvg} \triangleq
    N_{12}^* \circ \big(\fof{}{s_1}{fin\_dvg} \cup (N_1 \circ
      \fof{}{s_2}{fin\_dvg})
      \cup (\triangle N_{12})^\infty \big)
  \end{aligned}
  \\
  &\begin{aligned}
    \fof{}{\textbf{loop }s_1 \ s_2}{inf\_dvg} \triangleq
    \Big(\big((\triangle N_{12})^* \circ \blacktriangle N_{12}\big)^* \circ
    (\triangle N_{12})^* \circ \mathbb{D}\Big) \cup
    \big((\triangle N_{12})^* \circ \blacktriangle N_{12}\big)^\infty
  \end{aligned}
  \\
  &\quad \quad \begin{aligned}
    \text{where } \mathbb{D} = \fof{}{s_1}{inf\_dvg} \cup
        \big(N_1 \circ \fof{}{s_2}{inf\_dvg}\big)  
  \end{aligned}
  \end{align*}
Remark that $\mathbb{D}$ denotes the reacting behavior by the loop body $s_1$ or $s_2$,
and $(\triangle N_{12})^* \circ \blacktriangle N_{12}$ denotes that one observable behavior must be triggered after silently looping a finite number of times.

\section{Denotational Semantics of Procedure Calls}
\label{sec:semantics_pcall}
 So far, we have shown how to precisely define the denotational semantics of programs with nondeterministic and nonterminating behaviors without considering procedure calls.
 
Another challenge in applying such semantics to realistic compiler verification is how to model the semantics of open modules and make it compositional, in the sense that the semantics of every separately compiled module can be linked together and keep behavioral consistency with the linked semantics of source modules. Our key idea, extended from the standard way of defining denotational semantics for \emph{whole programs} with procedure calls, is to define an open module's denotation under a semantic ``oracle'' that describes the terminating behavior of external calls. Then when linking, the semantics is given by merging the denotation of modules and then taking fixed points, in which external calls may now be ``handled'' by one of the linked modules.
In this section, we begin with a simple language for introducing the standard denotational semantics for programs with procedure calls in \S\ref{subsec:standard_pcall}, and then propose a denotational semantics for open modules and a novel semantic linking operator in \S\ref{subsec:pcall_module}. Finally, their divergence is defined in \S\ref{subsec:pcall_dvg}.

For presentation, this section ignores observable events and focuses solely on procedure calls.

\subsection{Standard Denotational Semantics of Procedure Calls (Whole Program)}
\label{subsec:standard_pcall}

Consider a language, PCALL, defined as follows: a whole program in PCALL consists of a list of procedures. Each procedure \( p \), belonging to a set \(\kwd{Pdef}\), is characterized by its procedure name \( id_p \), a list of local variables \(x_1;...; x_n \), and its procedure body \( c_p \), which is represented as a statement.
\[
\begin{array}{l@{\qquad}l}
  \begin{aligned}
  \text{Statements: } 
    c \triangleq \textbf{skip} \sep atom \sep \textbf{choice}(c_1,c_2) \sep c_1;c_2 
    \sep &
    \\
    \textbf{if}\ b\ \textbf{then}\ c_1\ \textbf{else}\ c_2 
    \sep \textbf{while}\ b\ \textbf{do}\ c 
    \sep \mathcolor{red}{\textbf{call}\ id} \;& 
  \end{aligned}
  & \begin{aligned}
    &\text{Procedures: } p \triangleq (id_p, x_1;\dots;x_n, c_p) 
    \\
    &\text{Whole programs: } P \triangleq p_1;\dots;p_n    
  \end{aligned}
\end{array}
\]
In this toy language, procedure calls have no arguments and no return values.
Before defining its semantics, we partition program states into local and global components: {\tts state} $\triangleq \ttt{lstate} \times \ttt{gstate}$. Let the symbol $l$ range over the set of local states {\tts lstate} and $g$ over the set of global states {\tts gstate}.

The standard way \cite{DBLP:journals/tcs/Back83,DBLP:conf/ac/Park79} to define the denotational semantics of programs with procedure calls is parameterizing the denotation of procedures with the behavior of invoked procedures, and then taking the Kleene fixed point (see Thm. \ref{thm:kleene_fixpoint}) on recursive domain equations.
\label{subsec:pcall_terminate}

Specifically,
let {\tts ident} be the set of procedure names. Consider a semantic oracle $\chi$ ($\subseteq \ttt{ident} \times \ttt{gstate} \times \ttt{gstate}$) such that an element $(id, g_0, g_1) \in \chi$ if and only if executing the procedure named $id$ from global state $g_0$ could normally terminate at global state $g_1$.
The terminating behavior of PCALL statements is still a binary relation $\fofchi{c}{nrm}$ ($\subseteq \ttt{state} \times \ttt{state}$) so that an element $(\sigma_0, \sigma_1) \in\fofchi{c}{nrm}$ if and only if
executing statement $c$ from state $\sigma_0$ could terminate at state $\sigma_1$, in which the semantics of procedure calls will be given by $\chi$, i.e.,
\begin{align*}
  &\fofchi{\textbf{call } id}{nrm} \triangleq \{\left((l, g_0), (l, g_1)\right)\mid (id, g_0, g_1) \in \chi\}
\end{align*}
Notably, the local state remains unchanged after executing the procedure call.
The semantics of other PCALL statements  can refer to the semantic definitions of WHILE statements, for example:
\begin{align*}
    \fofchi{c_1;c_2}{nrm} \triangleq 
    \fofchi{c_1}{nrm} \circ
    \fofchi{c_2}{nrm} 
\end{align*}
We further use $\fofchi{p}{nrm}$ ($\subseteq \{id_p\} \times \ttt{gstate} \times \ttt{gstate}$) to represent the terminating behavior of the procedure $p$ such that an element $(id_p, g_0, g_1) \in \fofchi{p}{nrm}$ if and only if executing procedure $p$ at global state $g_0$ could eventually terminate at global state $g_1$. It's easy to see $\fofchi{p}{nrm}$ can be derived from the semantics of its body $\fofchi{c_p}{nrm}$ as illustrated in Fig. \ref{fig:procedure_sem}, and the formal definition is shown in Formula~(\ref{eq:procedure_sem}), where the function {\tts Init}: {\tts Pdef} $\times$ {\tts gstate} $\rightarrow$ {\tts state} gives the initial local and global state after allocating local variables for the execution of the procedure body \(c_p\), and the function {\tts Free}: {\tts lstate} \(\times\) {\tts gstate} $\rightarrow$ {\tts gstate} gives the final global state after freeing local addressed variables.
\begin{figure}[t]
    \centering
    \includegraphics[width=0.7\linewidth]{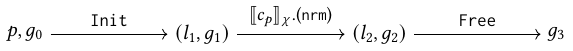}
    \caption{An illustration of the semantic definition of \(\fofchi{p}{nrm}\).}
    \label{fig:procedure_sem}
    \Description{This figure illustrates how the terminating semantics of a procedure \(p\) is derived from the terminating semantics of its body \(c_p\). Starting from an initial global state \(g_0\), the function \texttt{Init} allocates the procedure’s local variables and produces an initial local-global state \((l_1, g_1)\). The body semantics \(\fofchi{c_p}{nrm}\) then relates this input state to a final local-global state \((l_2, g_2)\) after normal termination of the procedure body. Finally, the function \texttt{Free} releases the local addressed variables and produces the resulting global state \(g_3\). }
\end{figure}
\begin{equation}
 \label{eq:procedure_sem}
   \fofchi{p}{nrm} \triangleq \Set{(id_p, g_0, g_3) | \exists l_1\ l_2\ g_1\ g_2,\
  \begin{aligned}
    &\!\big((l_1, g_1), (l_2, g_2)\big) \in \fofchi{c_p}{nrm}
     \ \wedge \\
    &(l_1, g_1) = \code{Init}(p, g_0) \wedge g_3 = \code{Free}(l_2, g_2)
  \end{aligned}}   
\end{equation}
Finally, the terminating behavior of the whole program, i.e., \(\bracket{p_1;...;p_n}.(\ttt{nrm})\), can be defined as:
\begin{align}
\label{eq:pcall_nrm}
 &\begin{aligned}
     \bracket{p_1;...;p_n}.(\ttt{nrm}) \triangleq \mu  \chi.\left(\fofchi{p_1}{nrm} \cup \cdots \cup\fofchi{p_n}{nrm}\right)
 \end{aligned}
\end{align}
 According to the Kleene fixed point theorem, Formula~(\ref{eq:pcall_nrm}) means that a procedure call could terminate if and only if there exists a natural number $k$ such that the procedure call terminates with at most $k$ layers of nested calls.
We proceed to discuss how we extend this approach to support semantic linking of open modules in \S\ref{subsec:pcall_module}, and how we define the diverging behavior in \S\ref{subsec:pcall_dvg}.

\subsection{Semantics of Open Modules and Semantic Linking}
\label{subsec:pcall_module}
A more realistic program is usually composed of multiple modules (or partial programs), each of which contains a list of procedures and is compiled individually. We find that the traditional approach from~\citeN{DBLP:journals/tcs/Back83} (which we summarize in \S\ref{subsec:pcall_terminate}) can be easily extended to this setting by parameterizing the semantics of open modules with the terminating behavior of external procedures\footnote{
 External procedures are those not defined in the current module.
 }.
\subsubsection{Semantics of Open Modules}
Formally, we use $\fofchi{M}{nrm}$ ($\subseteq \ttt{ident} \times \ttt{gstate} \times \ttt{gstate}$) to represent the terminating behavior of an open module. It means that an element $(id, g_0, g_1)\in\fofchi{M}{nrm}$ if and only if there exists a procedure named $id$ in module $M$ and its execution beginning with global state $g_0$ could finally terminate at global state $g_1$, in which the terminating behavior of external procedures is given by an oracle $\chi$ ($\subseteq\ttt{ident} \times \ttt{gstate} \times \ttt{gstate}$). The terminating behavior of an open module $M$ that consists of a list of procedures $p_1;...;p_n$ can then be defined as:
\begin{figure}[t]
    \begin{subfigure}[b]{.9\linewidth}
    \centering
    \includegraphics[width=0.75\linewidth]{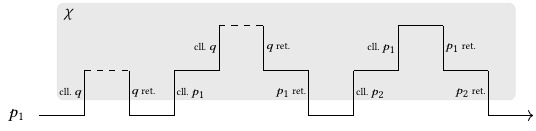}
    \caption{Semantics of a single procedure $p_1$}
    \label{subfig:procedure_sem}
    \end{subfigure} 
    \begin{subfigure}[b]{.9\linewidth}
        \flushright
    \includegraphics[width=0.75\linewidth]{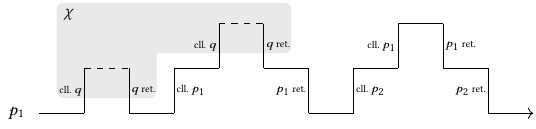}
    \centering
    \caption{Semantics of module $M$ containing procedure $p_1$}
    \label{subfig:moudle_sem}
    \end{subfigure}
    \caption{Comparison between the semantics of procedures and that of modules. Here raising edges denote procedure calls (abbr. cll.), falling edges denote procedure returns (abbr. ret.), and the dashed line represents the behavior of external procedure $q$.
    The shaded areas show that for $\fofchi{p_1}{nrm}$, the behavior of every procedure call of $p_1$ is interpreted by $\chi$, while for $\fofchi{M}{nrm}$, only the behavior of external calls are interpreted by $\chi$.}
    \label{fig:modproc}
    \Description{This figure compares the terminating semantics of a single procedure with that of an open module containing that procedure. Subfigure (a) shows the semantics of procedure \(p_1\): starting from \(p_1\), all procedure calls made during its execution, including recursive calls to \(p_1\) itself and calls to other procedures such as the external procedure \(q\), are interpreted by the semantic oracle \(\chi\). Subfigure (b) shows the semantics of a module \(M\) containing \(p_1\): calls between procedures defined inside the module are resolved internally by the module semantics, while only calls to procedures outside the module, represented by the dashed edge to \(q\), are interpreted by \(\chi\).}
\end{figure}

\begin{align}
     \fofchi{M}{nrm}& \triangleq 
        \mu  \chi_0.\left(\fof{\chi_0 \cup \chi}{p_1}{nrm} \cup \cdots \cup\fof{\chi_0 \cup \chi}{p_n}{nrm}\right),\text{ where } M = p_1;\dots;p_n
        \label{eq:module_terminate}
 \end{align}
 It is trivial to see that the standard whole-program denotational semantics given in Formula~(\ref{eq:pcall_nrm}) is a special case of Formula~(\ref{eq:module_terminate}). In Formula~(\ref{eq:module_terminate}), the semantic oracle $\chi_0$ ($\subseteq\ttt{ident} \times \ttt{gstate} \times \ttt{gstate}$) represents the terminating behavior of internal procedures, and the Kleene least fixed-point operator is used to capture arbitrary finite layers of internal calls.
It is worth noting that although both the terminating behavior of a single procedure and that of an open module are subsets of $\ttt{ident} \times \ttt{gstate} \times
	\ttt{gstate}$, they depend on {different} semantic oracles.
 As shown in Fig. \ref{fig:modproc}, the semantics of a procedure depends on the semantics of all invoked procedures, even including recursive calls to itself. In comparison, the semantics of an open module depends on the semantics of external procedures only.

\subsubsection{Semantic Linking}
\label{subsec:semantic_linking}
Semantic linking means composing the behavior of individual modules as a whole so that cross-module calls between module $M_1$ and module $M_2$ before linking become “internal” calls after linking.
That is, after semantic linking, only external calls outside module $M_1$ and module $M_2$ need to be interpreted by a semantic oracle.
Therefore, the definition of semantic linking is straightforward:
assume that the terminating behavior of procedures outside module $M_1$ and $M_2$ are given by an oracle $\chi$, the terminating behavior of semantic linking is defined as:
\begin{align}
\label{eq:module_semantic_linking}
  \linkofchi{\bracket{M_1}\oplus \bracket{M_2}}{nrm} \triangleq
       \mu \chi_0.\left(\fof{\chi_0 \cup \chi}{M_1}{nrm} \cup \fof{\chi_0 \cup \chi}{M_2}{nrm}\right)  
\end{align}
As we can see, similar to the semantic definition of modules, the semantic linking between two modules is defined by merging their denotations as if all the procedures reside in one module and then taking the fixed point.
Based on these definitions, the equivalence between
semantic and syntactic linking is shown as follows, which is reduced to Beki\'{c}'s theorem \cite{DBLP:conf/ibm/Bekic84e}.
Therefore, proving horizontal compositionality in our framework is straightforward.

\begin{theorem} [Equivalence between semantic and syntactic linking]
\label{thm:ss_equiv}
    For any modules $M_1$ and $M_2$, 
    $\bracket{M_1} \oplus \bracket{M_2} = \bracket{M_1 + M_2}$, where “$+$” is syntactic linking, physically putting modules together.
\end{theorem}
\begin{theorem}[Beki\'{c}'s theorem]
\label{lemma:lemma_co1}
    Given CPOs $A_1$ and $A_2$, for any monotonic and continuous functions $f: A_1 \times A_2 \rightarrow A_1$ and $g: A_1 \times A_2 \rightarrow A_2$,
    $$ \mu (x, y). \big(\mu x_0.f(x_0, y), \mu y_0.g(x,y_0)\big) = \mu (x, y). \big(f(x, y), g(x, y)\big)
    $$
\end{theorem}
{Note that} Thm.~\ref{thm:ss_equiv} is not specific to PCALL. It generally holds for settings where the semantics of open modules and their semantic linking are defined in the form of Formulae (\ref{eq:module_terminate}) and (\ref{eq:module_semantic_linking}).

\subsection{Diverging Behavior of PCALL}
\label{subsec:pcall_dvg}
The diverging behavior of a procedure can be induced either by its internal statements or by recursive procedure calls.
We define the diverging behavior of a PCALL program as follows: when defining statement denotation $\bracket c_\chi$ and procedure denotation $\bracket p_\chi$,
 we only consider internal divergence (caused by dead loops) first. Those caused by procedure calls (including dead loops in callee procedures and infinite layers of nested calls) are captured at the moment of defining module denotation $\bracket M_\chi$. For this purpose, we further add the following set {\tts cll} to {\tts CDenote} %
 to record the point of procedure calls, so that an element $(\sigma_0, (id, g)) \in \fofchi{c}{cll}$ if and only if executing statement $c$ from state $\sigma_0$ could reach a point of procedure call named $id$ at global state $g$.
\begin{lstlisting}[language=Coq]
    cll: state -> query -> Prop, $\text{\normalsize where }$query $\triangleq$ ident $\times$ gstate.
\end{lstlisting}
The calling behavior
$\fofchi{c}{cll}$ is also recursively defined on the syntax of statements, for example,
\begin{align*}
    &\fofchi{c_1;c_2}{cll} \triangleq 
\fofchi{c_1}{cll} \cup \big(\fofchi{c_1}{nrm} \circ \fofchi{c_2}{cll}\big) \\
&\fofchi{\textbf{call } id}{cll} \triangleq \big\{\big((l, g), (id, g)\big) \mid (l, g) \in \text{\tts state} \big\}
\end{align*}
Correspondingly, the calling behavior of a procedure $\fofchi{p}{cll}$ is derived from that of its body, i.e., 
  \begin{align*}
      \fofchi{p}{cll} \triangleq \big\{ \big((id_p, g_0), (id, g_2)\big) \mid 
      \exists l_1\ g_1, (l_1, g_1) = \ttt{Init}(p, g_0) \wedge
      \big((l_1, g_1), (id, g_2)\big) \in \fofchi{c_p}{cll}\big\}
  \end{align*}
\subsubsection{Diverging Behavior of Open Modules} We use $\fofchi{M}{dvg}$ ($\subseteq \ttt{ident} \times \ttt{gstate}$) to denote the diverging behavior of open modules. Specifically, $(id, g) \in \fofchi{M}{dvg}$ if and only if there exists a procedure named $id$ in module $M$ and its execution beginning with global state $g$ could diverge due to internal statements after a finite number of procedure calls, or an infinite number of internal procedure calls inside module $M$.
Thus, assume that the terminating behavior of external procedures is given by semantic oracle $\chi$, and let $\hat{\chi} = \fofchi{M}{nrm} \cup \chi$ (i.e., $\hat{\chi}$ specifies the terminating behavior of all procedures, including both internals and externals), $\fofchi{M}{dvg}$ can be defined as:
\begin{align*}
     & \parents{M}_{\hat{\chi}} \triangleq \bracket{p_1}_{\hat{\chi}} \cup \cdots \cup \bracket{p_n}_{\hat{\chi}} \text{, where } M = p_1;\dots;p_n
     \\
     &\fofchi{M}{dvg} \triangleq \big(\pop{\hat{\chi}}{M}{cll}^* \circ \pop{\hat{\chi}}{M}{dvg}\big) \cup \pop{\hat{\chi}}{M}{cll}^\infty
\end{align*}
where $ \parents{M}_{\hat{\chi}}$ is defined by simply merging corresponding behavior sets of procedures $p_1, \dots, p_n$.

We further use $\fofchi{M}{cll}$ ($\subseteq \kwd{query} \times \kwd{query} $) to denote the calling behavior of a module such that an internal procedure of module $M$ could eventually call an external procedure after a finite number of internal calls. Let $\fofchi{M}{dom}$ (\(\subseteq \kid\)) be the set of procedure names defined in $M$, and then 
\begin{align*}
      &\kecall(X) \triangleq \Set{(\ivar, g) \mid  id \not\in X }, \text{ for } X \subseteq \kid    
  \\
  &\fofchi{M}{cll} \triangleq
    \pop{\hat{\chi}}{M}{cll}^+ \circ \ktest\bigl(\kecall(\fofchi{M}{dom})\bigr)
\end{align*}
\subsubsection{Diverging Behavior of Semantic Linking}
Based on the definitions above, semantic linking can be easily extended to the {\tts dvg} field and {\tts cll} field:
 assume that the terminating behavior of external procedures outside $M_1$ and $M_2$ is given by semantic oracle $\chi$, and let $\Hat{\chi} = \linkofchi{\bracket{M_1}\oplus \bracket{M_2}}{nrm} \cup \chi$, 
$\linkofchi{\bracket{M_1}\oplus \bracket{M_2}}{dvg}$ can be defined as follows.
\begin{align*}
    &\begin{aligned}
      \linkofchi{\bracket{M_1}\oplus \bracket{M_2}}{dvg} \triangleq \
      &\big(\fof{\hat{\chi}}{M_1}{cll} \cup \fof{\hat{\chi}}{M_2}{cll}\big)^* \circ
      \big(\fof{\hat{\chi}}{M_1}{dvg} \cup \fof{\hat{\chi}}{M_2}{dvg}\big) 
       \ \cup \\
      &\big(\fof{\hat{\chi}}{M_1}{cll} \cup \fof{\hat{\chi}}{M_2}{cll}\big)^\infty
    \end{aligned}
    \\
    &\begin{aligned}
      \linkofchi{\bracket{M_1}\oplus \bracket{M_2}}{cll} \triangleq
      \big(\fof{\hat{\chi}}{M_1}{cll} \cup \fof{\hat{\chi}}{M_2}{cll}\big)^+ \circ
      \ktest\bigl(\kecall(\fofchi{M_1 + M_2}{dom})\bigr)
    \end{aligned}   
\end{align*}

\paragbf{Remark} There may exist traditional approaches~\cite{DBLP:journals/tcs/Back83} that, for example, assume another semantic oracle $\chi^{\prime}(\subseteq$ {\tts ident} $\times$ {\tts gstate}) to interpret the diverging behavior of procedure calls, so that the diverging behavior of a module can be defined by taking appropriate fixed points of recursive equations similar to Formula~(\ref{eq:module_terminate}). However, this approach is not available for practical C modules, since one will meet the same stuttering problem as we do in \S\ref{subsec:loop_trace} when observable events are taken into account. Our approach successfully solves this problem. We can similarly define $\fofchi{M}{fin\_dvg}$ and $\fofchi{M}{inf\_dvg}$ for realistic C modules like how we handle Clight loops.
The key point is to use $\triangle \pop{\chi}{M}{cll}$ and $\blacktriangle\pop{\chi}{M}{cll}$ in the definition so that we can distinguish the cases that zero events happen before a call and the cases that at least one event happens before the call. {We illustrate this in \S\ref{subsec:module_sem}},
where the semantics of function calls and open modules in practical languages are defined similarly to the semantics of PCALL procedures and PCALL modules shown in this section, and the primary differences lie in the additional manipulation w.r.t. arguments and return values.

\subsection{Semantics of Unstructured Branches}
\label{subsec:goto_sem}
We further consider unstructured branches on top of PCALL with the following GOTO language.
\[
    c \triangleq \textbf{skip} \sep
    atom \sep c_1;c_2 \sep
    \textbf{choice}(c_1, c_2)\sep
    \textbf{if}\ b\ \textbf{then}\ c_1\ \textbf{else}\ c_2 \sep
    \textbf{while}\ b\ \textbf{do}\ c \sep
    \textbf{call} \ id \sep
     \mathcolor{red}{j: c} \sep
     \mathcolor{red}{\textbf{goto}\ j}
\]
where $j$ ranges over the set of labels. For handling \textbf{goto} statements, we add the following {\tts gto} set to {\tts CDenote}, so that an element $(\sigma_0,  (j, \sigma_1)) \in \fofchi{c}{gto}$ if and only if executing statement $c$ from state $\sigma_0$ could reach a point of jumping by \textbf{goto} statement at state $\sigma_1$ for the first time.
\begin{lstlisting}[language=Coq]
 gto: state -> trace -> goto_info -> Prop, $\text{\normalsize where }$goto_info $\triangleq$ label $\times$ state, $\text{\normalsize and }$label $\text{\normalsize is the set of labels}$.
\end{lstlisting}
We additionally define the behavior of \textbf{goto} statements as follows, where the normal execution of a program will be interrupted by a \textbf{goto} statement, so $\fofchi{\textbf{goto } j}{nrm}$ is assigned the empty set.
\begin{align*}
    \fofchi{\textbf{goto } j}{nrm} \triangleq \emptyset
    \quad \quad
    \fofchi{\textbf{goto } j}{gto} \triangleq \{(\sigma, (j, \sigma)) \mid \sigma \in \ttt{state}\}
\end{align*}
Existing semantic definition of other statements is the same as PCALL. The ``goto'' exiting behavior of other statements is recursively defined similarly to the calling behavior in \S\ref{subsec:pcall_dvg}, for example,
\begin{align*}
    \fofchi{c_1; c_2}{gto} \triangleq \fofchi{c_1}{gto} \cup \fofchi{c_1}{nrm} \circ \fofchi{c_2}{gto}
\end{align*}
We then use the following auxiliary relations to denote the terminating, ``goto'' exiting, and diverging behavior of executing $c$ from given labels, respectively.
\begin{align*}
    &\lol{\chi}{c}{nrm} \subseteq \ttt{goto\_info} \times \ttt{state} \\
    &\lol{\chi}{c}{gto} \subseteq \ttt{goto\_info} \times \ttt{goto\_info} \\
    &\lol{\chi}{c}{dvg} \subseteq \ttt{goto\_info}
\end{align*}
They are defined by evaluating part of statement $c$ from a given label, such that, for example, an element ($(j, \sigma_0), \sigma_1) \in \lol{\chi}{c}{nrm}$ if and only if we can find a labeled position ``$j: ...$'' in $c$ and executing $c$ from $j$ at state $\sigma_0$ could terminate at $\sigma_1$. 
For instance, given statements $s_1$ and $s_2$,
\[ 
  \lol{\chi}{c_1; j_1: c_2}{nrm} \triangleq \{((j_1, \sigma_0), \sigma_1) \mid (\sigma_0, \sigma_1) \in \fofchi{c_2}{nrm} \}
\]
Finally, the behavior of a procedure $p$ whose body $c$ supports \textbf{goto} constructs is extended as follows.
\begin{align*}
&\begin{aligned}
  \fofchi{p}{nrm} \triangleq \Set{(id_p, g_0, g_3) | \exists l_1\ g_1\ l_2\ g_2,\
  \begin{aligned}
    &\!\big((l_1, g_1), (l_2, g_2)\big) \in
            \mathcolor{red}{\fofchi{c}{nrm} \cup \big(\fofchi{c}{gto} \circ \mathbb{N} \big)}
     \\
    &\wedge (l_1, g_1) = \ttt{Init}(p, g_0) \wedge g_3 = \ttt{Free}(l_2, g_2)
  \end{aligned}}  
\end{aligned}
\\
&\begin{aligned}
  \fofchi{p}{dvg} \triangleq \Set{(id_p, g_0) | \exists l_1\ g_1,(l_1, g_1) = \ttt{Init}(p, g_0)\wedge 
  \begin{aligned}
    \!(l_1, g_1) \in
            \mathcolor{red}{\fofchi{c}{dvg} \cup \big(\fofchi{c}{gto} \circ \mathbb{F} \big)}
  \end{aligned}}  
\end{aligned}
\\
&\begin{aligned}
      \quad \text{ where } \mathbb{N} & = \lol{\chi}{c}{gto}^* \circ \lol{\chi}{c}{nrm}
  \quad \quad 
  \mathbb{F} = \big(\lol{\chi}{c}{gto}^* \circ \lol{\chi}{c}{dvg} \big) \cup
           (\lol{\chi}{c}{gto})^\infty
\end{aligned}
\end{align*}
Here, \(\mathbb{N}\) denotes that executing \(c\) from a given label could terminate normally after a finite number of ``goto'' iterations, and \(\mathbb{F}\) denotes that executing \(c\) from a given label could diverge, either by a certain sub-statement that could diverge after a finite number of ``goto'' iterations or by an infinite number of ``goto'' iterations. 
The procedure would then terminate if its body $c$ normally terminate, or the body meets a \textbf{goto} statement (i.e., \textbf{goto} $j$) and then execution from the given label \(j\) would finally terminate. The diverging behaviors are handled similarly, with additional consideration of infinite iterations by \textbf{goto}. Intuitively, unstructured branchings resembles a loop with multiple entry points, where each labeled statement acts as an entry, and each \textbf{goto} statement specifies which entry to jump to.
Finally,
the way of defining the semantics of open modules and their semantic linking is equally applicable to programs with \textbf{goto} constructs, as further illustrated in \S\ref{sec:semantics_front}.

\section{Denotational Semantics of Control Flow Graphs}
\label{sec:semantics_cfg}
Most compiler back-ends, including those of LLVM, GCC, and CompCert, are based on Control Flow Graphs (CFGs). We use the following CFG language to illustrate how the semantic linking shown in \S\ref{sec:semantics_pcall} can be applied to these settings. We here do not consider observable events for simplicity.
\begin{align*}
    &\text{Instructions}&&I\triangleq \mathbf{Do}(atom, j) \sep 
    \mathbf{Call}(i, j) \sep
    \mathbf{Cond} (b, j_1, j_2)
    \\
    &\text{Control flow graph}&& G\triangleq j \mapsto I
    \\
    &\text{Procedures}  &&p\triangleq (i_p, j_\ttt{entry}, j_\ttt{exit}, G) \\
    &\text{Modules} &&M\triangleq p_1;...;p_n
\end{align*}
A CFG is formulated as a finite mapping from labels to instructions.
 Each instruction carries the labels of its possible successors. Specifically, the execution of \(\mathbf{Do}(atom, j)\) involves executing the atomic instruction \(atom\) (e.g., an assignment) and then branching to node \(j\). The instruction \(\mathbf{Call}(i, j)\) means invoking procedure \(i\) and then branching to node \(j\). Finally, the instruction \(\textbf{Cond}(b, j_1, j_2)\) denotes a conditional branch determined by the truth value of the pure condition \(b\), which is assumed to have no side effects. 
 We model the nodes of CFGs with instructions rather than basic blocks since this simplifies semantics and reasoning over static analyses \cite{DBLP:conf/cc/KnoopKS98}.

The set of program states is defined as {\tts state} $\triangleq \ttt{label} \times \ttt{lstate} \times \ttt{gstate}$.
We again use a semantic oracle $\chi$ ($\subseteq \ttt{ident} \times \ttt{gstate} \times \ttt{gstate}$) to interpret the semantics of callee procedures, so as to apply the method presented in \S\ref{sec:semantics_pcall} to define the semantics of open modules and semantic linking in this context. The denotational semantics of typical instructions can then be defined as follows.
\begin{align*}
    &\begin{aligned}
        \mathbf{jump}(j) \triangleq
           \{\big((j_\ttt{pc}, l, g), (j'_\ttt{pc}, l, g)\big) \mid
            j'_\ttt{pc} = j
           \}, \text{ where } j_\ttt{pc} \text{ is the label of the current node.}
    \end{aligned}
     \\
    &\begin{aligned}
    \fofchi{\mathbf{Do}(atom, j)}{nrm} \triangleq
        \fofchi{atom}{nrm} \circ \mathbf{jump}(j)
    \end{aligned}
    \\
    &\begin{aligned}
    \fofchi{\mathbf{Call}(i, j)}{nrm} \triangleq
        \{\big((j_\ttt{pc}, l, g), (j_\ttt{pc}, l, g')\big) \mid
            (i, g, g') \in \chi
           \} \circ \mathbf{jump}(j)
    \end{aligned}
    \\
    &\begin{aligned}
        \fofchi{\mathbf{Cond}(b, j_1, j_2)}{nrm} \triangleq
        \mathbf{test}(\fof{}{b}{tts}) \circ \mathbf{jump}(j_1) \cup
        \mathbf{test}(\fof{}{b}{ffs}) \circ \mathbf{jump}(j_2)
    \end{aligned}
\end{align*}
We next use $\fofchi{G}{nrm}$ to denote the general union of denotations of all nodes. Because a procedure could normally terminate if and only if executing instructions from the entry node could finally branch to the exit node after a finite number of branches between nodes,
the terminating behavior of procedures and open modules (i.e., $\fofchi{p}{nrm}$ and $\fofchi{M}{nrm}$) can then be defined as follows. 
 \begin{align*}
     &\fofchi{G}{nrm} \triangleq {\textstyle \bigcup_j}\big(\fofchi{G(j)}{nrm}\big)
     \\
     &\fofchi{p}{nrm} \triangleq \{
        (i_p, g, g') \mid \exists l',
        \big((j_\ttt{entry}, l_\ttt{empty}, g), (j_\ttt{exit}, l', g')\big) \in \mathcolor{red}{\fofchi{G}{nrm}^*}
     \}
     \\
     &\fofchi{M}{nrm} \triangleq 
        \mu  \chi_0.\left(\fof{\chi_0 \cup \chi}{p_1}{nrm} \cup \cdots \cup\fof{\chi_0 \cup \chi}{p_n}{nrm}\right),\text{ where } M = p_1;\dots;p_n
 \end{align*}
The semantics of open modules is defined the same way as that in \S\ref{subsec:pcall_module}, so is the semantic linking. In essence, the denotational semantics of control-flow graphs (CFGs) does not fundamentally differ from their small-step semantics, except for the treatment of procedure calls. The purpose of presenting the semantics of CFGs in this section is to illustrate that this semantic construction naturally scales to the CompCert back-end languages. Accordingly, the semantic linking operator defined in \S\ref{subsec:pcall_module} can also be applied to the CompCert back-end, in particular to its assembly languages.

\section{Semantics of CompCert Front-end Languages}
\label{sec:semantics_front}
\subsection{Languages and Semantic Domain}
\label{subsec:dom_front}

\paragbf{Syntax of Clight.} We write \(s\) for the syntax of Clight statements, \(F\) for Clight functions, \(e\) for Clight expressions, and $x, y, z$ for program variables. 
 We recall the key definitions as follows, where $?term$ denotes the optional occurrence of a syntactic term, and $\Vec{term_k}$ is a list of syntactic terms with the representative element \(term_k\). The full syntax of Clight can be found in \cite{DBLP:conf/fm/BlazyDL06}.
\begin{align*}
  &\text{(Statements)} &&s \triangleq\
  \begin{aligned}[t]
     & \Sskip \mid
      \Sseq \mid
      \Sif \mid
      \Sloop \mid
      \Scall \mid
      \Sset \mid
       \\ 
      & %
      \Sbreak \mid
      \Scontinue \mid
      \Sreturn \mid
      \Sgoto \mid 
      \Slabel \mid
      \Sswitch
      \\
      &\text{where }\,  ?term \DEF\; ? \mid term
    \qquad \lslist{k} \DEF \kLSnil \mid (n_1, s_1);\lslist{2}
  \end{aligned} \numberthis \label{eqn:clit_stmt_syntax}
  \\
  & \text{(Functions)}
  &&\begin{aligned}
    F \triangleq (\Vec{(x_k, \type_{k})}, \Vec{(y_k, \type_{k})}, \Vec{(z_k, \type_{k})}, s_F, \return_F)
  \end{aligned} \numberthis \label{eqn:clit_funct_syntax}
\end{align*}
A Clight function is defined as a quintuple consisting, from left to right, of 
the parameter list, the list of local variables, the list of temporary variables, 
the function body, and the return type of \(F\).

\paragbf{CompCert memory model} CompCert models memories as a collection of disjoint blocks, 
each identified by a block identifier (ranged over by \(b\) of a countable set $\block$). Within a block, memory contents 
are addressed by an integer offset (ranged over by \(o\) of the integer set \(\offset\)), and operations (e.g., load and 
store) are defined in terms of these block-offset pairs (i.e., memory locations). 
A memory state (of set \kmem) can be simply viewed as a mapping from memory locations to program values\footnote{We refer to \citeN{compcert-memory-v2} for the detailed definition of the CompCert memory model.} (of set \kval) that range over 32-bit integers, 64-bit floating-point numbers, pointers, and a special value \kundef.
This memory model is employed for all intermediate languages of CompCert.

\begin{table}[t]
\centering
\caption{Program states and the query–reply interface for function calls.}
\label{tab:clight-sem-component}
\renewcommand{\arraystretch}{1.1}
\begin{tabular}{|l|l|l|}
\hline
\text{Component}\qquad & \text{Formal definition} \qquad& \text{Representative element} \\ 
\hline
Program states \qquad & ${\kwd{state}\color{white}{\kwd{c}}} \DEF \ktmp \times \kmem$ 
             & $\quad (te, m) \in \kstate$ \\
C queries \qquad & $\kcquery \DEF \kid \times \kvalist \times \kmem$ 
             & $\,(id, \vec{v}, m) \in \kcquery$ \\
C replies \qquad & $\kcreply \DEF \kmem \times \kval$ 
             & $\,\,\,\, (m', v') \in \kcreply$ \\
\hline
\end{tabular}
\end{table}

\paragbf{Clight semantic infrastructures} As shown in Table \ref{tab:clight-sem-component}, a Clight program state is defined as a pair consisting of a temporary environment \(te\) and a memory state \(m\), where a temporary environment is a mapping from variables to its values. A C function call provides a query consisting of the invoked function identifier \(id\), the argument list \(\vec{v}\), and the memory state \(m\). After the function call, the corresponding reply consisting of a memory state \(m'\) and a return value \(v'\) may be received.

\paragbf{Clight statement denotations.}  A Clight statement denotation is defined by the following signature.
\begin{lstlisting}[language=Coq]    
    Record Clit.Denote: Type := {
      nrm: state -> trace -> state -> Prop;              (* nrm $\subseteq$ state $\times$ trace $\times$ state ----------$\,$*)
      brk: state -> trace -> state -> Prop;              (* brk $\subseteq$ state $\times$ trace $\times$ state ----------$\,$*)
      ctn: state -> trace -> state -> Prop;              (* ctn $\subseteq$ state $\times$ trace $\times$ state ----------$\,$*)
      rtn: state -> trace -> state -> val -> Prop;        (* rtn $\subseteq$ state $\times$ trace $\times$ state $\times$ val ----*)
      err: state -> trace -> Prop;                      (* err $\subseteq$ state $\times$ trace -------------------*)
      cll: state -> trace -> cquery -> Prop;             (* cll $\subseteq$ state $\times$ trace $\times$ cquery ---------$\,$*)
      gto: state -> trace -> (label $\times$ state) -> Prop;     (* gto $\subseteq$ state $\times$ trace $\times$ (label $\times$ state) *)
      fin_dvg: state -> trace -> Prop;                  (* fin_dvg $\subseteq$ state $\times$ trace ---------------*)
      inf_dvg: state -> itrace -> Prop                  (* inf_dvg $\subseteq$ state $\times$ itrace --------------*)
    }.
\end{lstlisting}
CompCert classifies program behaviors into four categories: {normal termination}, {abortion}, {silent divergence}, and {reacting divergence}.   
They are captured in our denotation by the sets \(\knrm\), \(\kerr\), \(\kfin\), and \(\kinf\), respectively.  
Additionally, the sets \(\kbrk\), \(\kctn\), and \(\krtn\) are introduced to account for Clight’s structured control flow, corresponding to the \textbf{break}, \textbf{continue}, and \textbf{return} statements.  
Furthermore, the set \(\kcll\) is used to conveniently capture abortion and divergence arising from recursive calls, while the set \(\kgto\) is designed to capture possible unstructured branching by \textbf{goto}.  

\paragbf{Function/Open module denotations} 
Functions or open modules communicate with each other by language interfaces~\cite{DBLP:conf/pldi/KoenigS21,DBLP:journals/pacmpl/ZhangWWKS24}. Formally, a \emph{language interface} is a tuple \(\pair{\kquery, \kreply}\), where \kquery is a set of queries and \kreply is a set of replies. The signature of function/open module denotations is parameterized by a language interface, as defined below.
\begin{lstlisting}[language=Coq]    
    Record FDenote {query reply: Type}: Type := {
      dom: ident -> Prop;                               (* dom $\subseteq$ ident ----------------------------$\,$*)
      nrm: query -> trace -> reply -> Prop;               (* nrm $\subseteq$ query $\times$ trace $\times$ reply -----------$\,$*)
      err: query -> trace -> Prop;                       (* err $\subseteq$ query $\times$ trace --------------------$\!$*)
      cll: query -> trace -> query -> Prop;               (* cll $\subseteq$ query $\times$ trace $\times$ query -----------$\,$*)
      fin_dvg: query -> trace -> Prop;                   (* fin_dvg $\subseteq$ query $\times$ trace ----------------$\!$*)
      inf_dvg: query -> itrace -> Prop                   (* inf_dvg $\subseteq$ query $\times$ itrace ---------------$\!$*)
    }.
\end{lstlisting}
At the function level, the denotation ignores statement-level details (captured by the sets $\kbrk$, $\kctn$, $\krtn$, and $\kgto$ of \ClitDenote), and instead focuses on the four resulting behaviors from a given incoming query 
(i.e., termination, abortion, silent divergence, and reacting divergence), each capturing the original observable behaviors in CompCert (by the \ktrace or \kitrace field of each behavior set).

\subsection{Denotational Semantics of Clight Statements}
\label{subsec:clight_sem}

\paragbf{Overview} We compositionally build the semantics of Clight statements from broken semantics, labeled broken semantics, and finally to the intact semantics, as shown in Table~\ref{tab:clight-semantics}. \emph{Broken semantics} captures the behavior of 
a statement $s$ under the convention that execution halts immediately upon 
encountering a \goto statement, thereby disregarding the semantics of any 
subsequent instructions that the \goto might target. Building upon this, 
\emph{labeled broken semantics} interprets $s$ from a specified label, but still 
within the broken-semantics discipline. Finally, the \emph{intact semantics} of Clight statements is obtained by 
systematically composing the (labeled) broken semantics, thus recovering the full 
control-flow behavior of \goto statements.

These semantics are parameterized by an environment \(\ENV\) and a semantic oracle \(\chi\) (\(\subseteq \Tnee\)), where \(\mathbb{E}\) is a tuple of global environment \(\gvar\) and local environment \(\evar\), and the oracle \(\chi\)  is used for specifying the terminating behavior of function calls from the Clight statement \(s\).
\begin{table}[t]
\centering
\caption{Three related semantic notions of Clight statements.}
\label{tab:clight-semantics}
    \begin{tabular}{|l|c|c|l|}
    \hline
    \text{Semantics} & \text{Symbol} & \text{Denotation} & \text{Brief Description} \\
    \hline
     \rule{0pt}{2.5ex} Broken semantics & $\ssec{s}$ & $\ClitDenote$ & Halt at \textbf{goto} statements. \\
     \rule{0pt}{2.5ex} Labeled broken semantics & $\llec{s}$ & $\klabel \to \ClitDenote$ & Start from a given label. \\
     \rule{0pt}{2.5ex} Intact semantics & $\hhec{s}$ & $\ClitDenote$ & Full semantics via composition. \\
    \hline
    \end{tabular}
\end{table}
\subsubsection{Broken Semantics of Clight Statements}

\paragbf{Atomic statements} The denotation of atomic statements satisfy the following equations, where $\eeval(\ENV, \sigma, e)$ means trying to derive the value of \(e\) under environment $\ENV$ and state \(\sigma\).
\begin{align}
    &\sosec{\Sskip}{nrm} \DEF \idrel 
      \qquad
      \ssec{\Sskip}.(\kX) \DEF \oset, \kforx \in \Set{\kbrk, \kctn, \krtn, \kerr, \kcll, \kgto, \kfin, \kinf} \nonumber
    \\
    &\sosec{\Sset}{nrm} \DEF \Set{\big((te, m),\knil, (te', m)\big)  |
       \begin{aligned}
        &\exists v, \eeval(\ENV, (te, m), e) = \ksome(v)
        \\
        &\quad  \land te' = te[x \mapsto v]   
       \end{aligned}
     }
    \\
    & \sosec{\Sset}{err} \DEF \Set{\big(\sigma,\knil\big) \mid
    \eeval(\ENV, \sigma, e) = \knone}
    \\ 
    & \ssec{\Sset}.(\kX) \DEF \oset, \kforx \in \Set{\kbrk, \kctn, \krtn, \kcll, \kgto, \kfin, \kinf }
    \\
    & \sosec{\Sreturn}{rtn} \DEF 
        \Set{(\sigma, \knil, \sigma, v) | \eeval(\ENV, \sigma, ?e) = \ksome(v)} 
    \\
    & \sosec{\Sreturn}{err} \DEF 
        \Set{(\sigma, \knil) | \eeval(\ENV, \sigma, ?e) = \knone}
    \\ 
    & \ssec{\Sreturn}.(\kX) \DEF \oset, \kforx \in \Set{\knrm, \kbrk, \kctn, \kcll, \kgto, \kfin, \kinf }
\end{align}

\paragbf{Clight function calls} The denotation of Clight function calls satisfies the following equations, where $\qeval(\ENV, \sigma, e_1, \Vec{e_2})$ means trying to derive the query of a call from \(e_1\) and \(\Vec{e_2}\) under environment \(\ENV\) and state \(\sigma\). 
\begin{align*}
&\begin{aligned}
    \sosec{\Scall}{nrm} \triangleq
      \Set{ ((te, m), \tau, (te', m')) |  
        \begin{aligned}
            & \exists q, \, \qeval(\ENV, (te, m), e_1, \Vec{e_2}) = \ksome(q)
            \\
            & \;\; \land \exists v',\, (q, \tau, (m', v')) \in \chi \land
              te' = te[?x\mapsto v']  
        \end{aligned}
      }
\end{aligned}
\\
&\begin{aligned}
    \sosec{\Scall}{cll} \triangleq
      \Set{(\sigma, \ttt{nil}, q) \mid \qeval(\ENV, \sigma, e_1, \Vec{e_2}) = \ksome(q) }
\end{aligned} \numberthis
\\
&\begin{aligned}
    \sosec{\Scall}{err} \DEF \Set{(\sigma, \knil) \mid \qeval(\ENV, \sigma, e_1, \Vec{e_2}) = \knone}
\end{aligned} \numberthis
\\
&\begin{aligned}
    \ssec{\Scall}.(\kX) \DEF \oset, \kforx \in \Set{\kbrk, \kctn, \krtn, \kgto, \kfin, \kinf}
\end{aligned} \numberthis
\end{align*}

\paragbf{Operators for composing denotations} Before showing the semantics of composite statements, we define the following operators for semantic composition. Specifically, for given denotations \(D_1\) and \(D_2\) of \ClitDenote, we use \(D_1 \Cup D_2\) to denote the union of them, and use \(D_1 \fatsemi D_2\) to denote the sequential composition of them, in which the composition operator $\circ$ is overloaded as follows.
\begin{align*}
    & (D_1 \Cup D_2).(\kX) \DEF D_1.(\kX) \cup D_2.(\kX), \kforx \in \Full
    \\
    & (D_1 \fatsemi D_2).(\knrm) \DEF D_1.(\knrm) \circ D_2.(\knrm)
    \\
    &\begin{aligned}
      (D_1\fatsemi D_2).(\kX) \DEF D_1.(\kX) \cup (D_1.(\knrm) \circ D_2.(\kX)),\;
      \kX \in \{\kbrk, \kctn, \krtn, \kerr, \kcll, \kgto, \kfin, \kinf \}
    \end{aligned}
    \\
&\qquad 
\begin{aligned} \text{where}\quad\!
  &R_1 \circ Y &&\!\!\!\!\! \triangleq 
    \{ (\sigma_1, \tau) \mid
    \exists \sigma_2\ \tau_1\ \tau_2,
      (\sigma_1, \tau_1, \sigma_2) \in R_1
    \wedge
    (\sigma_2, \tau_2) \in Y
    \wedge
    \tau = \tau_1 \cdot \tau_2
    \}
  \\
  &R_1 \circ R_2 &&\!\!\!\!\! \triangleq
    \{ \left(\sigma_1, \tau, \sigma_3\right) \mid 
    \exists \sigma_2\ \tau_1\ \tau_2, \
    (\sigma_1, \tau_1, \sigma_2) \in R_1
    \wedge
    (\sigma_2, \tau_2, \sigma_3) \in R_2
    \wedge
    \tau = \tau_1 \cdot \tau_2
    \}
  \\
  &R_1 \circ V &&\!\!\!\!\! \triangleq
    \{ \left(\sigma_1, \tau, \sigma_3, v\right) \mid 
    \exists \sigma_2\ \tau_1\ \tau_2, \
    (\sigma_1, \tau_1, \sigma_2) \in R_1
    \wedge
    (\sigma_2, \tau_2, \sigma_3, v) \in V
    \wedge
    \tau = \tau_1 \cdot \tau_2
    \}
\end{aligned}
\end{align*}

\paragbf{Clight control flow} The denotations of Clight control flow statements (i.e., conditional branching, sequential statements, \textbf{break}, \textbf{continue}, \textbf{return}, and Clight loops) satisfy the following equations, where \(\beval(\ENV, \sigma, e)\) meaning trying to derive the Boolean value of the Clight expression \(e\) under environment \(\ENV\) and state \(\sigma\). We write $\Uparrow X$ to lift a set of states $X$ into a binary relation with empty trace, 
and $\nabla R$ to project a ternary relation $R$ onto a binary one by canceling its target state.
\begin{align*}
&\begin{aligned} %
    &\fof{\ENV}{e}{tss}\; (\subseteq \kstate) \DEF \Set{\sigma | \beval(\ENV, \sigma, e) = \ksome(\textbf{T})  }
    \\
     &\fof{\ENV}{e}{fss}\; (\subseteq \kstate) \DEF \Set{\sigma | \beval(\ENV, \sigma, e) = \ksome(\textbf{F})  }
     \\
    &\fof{\ENV}{e}{err}\; (\subseteq \kstate) \DEF \Set{\sigma |\beval(\ENV, \sigma, e) = \knone  }
\end{aligned}
\\
&\Uparrow X \;\DEF\; \{\,(\sigma, \knil) \mid \sigma \in X\,\},  
\qquad \quad\, \text{ for }  X \subseteq \kstate,
\\ 
&\nabla R \;\DEF\; \{\,(\sigma, \tau) \mid \exists \sigma',\, (\sigma, \tau, \sigma') \in R \},  
\text{ for } R \subseteq \kstate \times \ktrace \times \kstate.
\\    
    & \ssec{\Sseq} \DEF \ssec{s_1}\fatsemi \ssec{s_2}
\\
    &\begin{aligned} %
        &\ssec{\Sif}.(\kX) \DEF \ktest(\fof{\ENV}{e}{tss}) \circ \ssec{s_1}.(\kX) \cup \ktest(\fof{\ENV}{e}{fss}) \circ \ssec{s_2}.(\kX), \\
        &\qquad \qquad\;\;
        \kforx \in \Set{\knrm, \kbrk, \kctn, \krtn, \kcll, \kgto, \kfin, \kinf }
    \end{aligned} 
    \\
    &\begin{aligned}
        &\sosec{\Sif}{err} \DEF\;
        \\
        &\qquad \qquad\;\; (\Uparrow \fof{\ENV}{e}{err}) \cup \ktest(\fof{\ENV}{e}{tss}) \circ \sosec{s_1}{err} \cup \ktest(\fof{\ENV}{e}{fss}) \circ \sosec{s_2}{err}
    \end{aligned} 
\\
&\begin{aligned}
&\sosec{\textbf{break}}{brk} \triangleq \idrel, \text{ and other fields are assigned the empty set}. \\
&\sosec{\textbf{continue}}{ctn} \triangleq \idrel, \text{ and other fields are assigned the empty set}.
\end{aligned}
\\
& N_1 = \sosec{s_1}{nrm} \cup \sosec{s_1}{ctn}
  \quad \quad N_{12} = N_1 \circ
    \sosec{s_2}{nrm}
\\
&\begin{aligned}
  \sosec{\Sloop}{nrm} \triangleq
    &\ N_{12}^* \circ(\sosec{s_1}{brk}\cup N_1 \circ
    \sosec{s_2}{brk}) 
\end{aligned}
\\
& \sosec{\Sloop}{brk} = 
  \sosec{\Sloop}{ctn} \triangleq \oset 
\\
&\begin{aligned}
  \ssec{\Sloop}.(\kX) \triangleq
    &\ N_{12}^* \circ(\ssec{s_1}.(\kX)\cup N_1 \circ
    \ssec{s_2}.(\kX)), \kforx \in \Set{\krtn, \kcll, \kgto}
\end{aligned}
\\
&\begin{aligned}
  \sosec{\Sloop}{err} \triangleq
    &\ N_{12}^* \circ \big(\sosec{s_1}{err}\cup N_1 \circ
    \sosec{s_2}{err} \cup N_1 \circ (\nabla \sosec{s_2}{ctn})\big)
\end{aligned}
\\
&\begin{aligned}
  \sosec{\Sloop}{\ttfin} \triangleq
  N_{12}^* \circ \big(\fof{}{s_1}{\ttfin} \cup N_1 \circ
  \fof{}{s_2}{\ttfin}
  \cup (\triangle N_{12})^\infty \big)
\end{aligned}
\\
&\begin{aligned}
    \sosec{\Sloop}{\ttinf} \DEF\; 
    &\text{ let }
    \mathbb{D} = \sosec{s_1}{\ttinf} \cup
    \big(N_1 \circ \sosec{s_2}{\ttinf}\big)
    \text{ in} 
    \\
    & \left((\triangle N_{12})^* \circ \blacktriangle N_{12}\right)^* \circ
    (\triangle N_{12})^* \circ \mathbb{D} \cup
    \left((\triangle N_{12})^* \circ \blacktriangle N_{12}\right)^\infty
\end{aligned}
\end{align*}

\paragbf{Switch statement} 
In Clight, a \textbf{switch} statement evaluates a case expression~$e$ and then selects 
the execution branch from a list of case statements~$\lslist{k}$. Each case statement 
consists of an optional case label~$?n_k$ and its associated statement~$s_k$. 
If $e$ matches a case label, execution begins at that case; if no match is found, 
execution starts from the first default case (denoted $(?, s)$); and if neither a 
matching case nor a default exists, the switch terminates immediately, behaving as 
\(\Sskip\). 
Formally, we use $\ssec{\lslist{k}}$ to denote the semantics of sequentially executing 
the list of case statements $\lslist{k}$ one by one, as defined by 
Equations~(\ref{eqn:broken_nil}) and~(\ref{eqn:broken_sl}). 
We further define $\wwec{\lslist{k}}(n)$ to denote execution beginning from the 
case labeled $n$, if it exists, or from the default case otherwise.

\begin{align*}
   &\seqec{\kLSnil} \DEF \ssec{\Sskip} \numberthis \label{eqn:broken_nil}
   \\ 
   &\seqec{(?n_1,s_1);\lslist{2}} \DEF \ssec{s_1}\fatsemi \seqec{\lslist{2}} \numberthis \label{eqn:broken_sl}
   \\
   \\
   &\wwec{\kLSnil}(n) \DEF \ssec{\Sskip} \numberthis \label{eqn:sl_nil}
   \\
   &\begin{aligned}
        & \wwec{(n_1,s_1);\lslist{2}}(n) \DEF \wwec{\lslist{2}}(n), && \text{ if } n \neq n_1 
   \\
        & \wwec{(n_1,s_1);\lslist{2}}(n) \DEF \ssec{s_1} \fatsemi \seqec{\lslist{2}}, && \text{ if } n = n_1 
   \end{aligned} \numberthis \label{eqn:sl_case_some}
   \\
   &\begin{aligned}
   & \wwec{(?,s_1);\lslist{2}}(n) \DEF \ssec{s_1} \fatsemi \seqec{\lslist{2}}, && \ \ \text{ if case } n \text{ not found in } \lslist{2} 
   \\
   & \wwec{(?,s_1);\lslist{2}}(n) \DEF  \wwec{\lslist{2}}(n), && \ \  \text{ if case } n \text{ found in } \lslist{2} 
   \end{aligned} \numberthis \label{eqn:sl_case_none}
   \\
   \\
    &\vvec{\ENV}{n}{e}.(\knrm) \DEF \Set{(\sigma, \knil, \sigma) | \weval(\ENV, \sigma, e) = \ksome(n)} 
    \\
    &\vvec{\ENV}{n}{e}.(\kerr) \DEF \Set{(\sigma, \knil) | \weval(\ENV, \sigma, e) = \knone}
    \\
    &\ssec{\Sswitch}.(\knrm) \DEF  \bcup{n}{
        \Bigl(\vvec{\ENV}{n}{e}.(\knrm) \circ
            \bigl(\wwec{\lslist{k}}(n).(\knrm) \cup
                \wwec{\lslist{k}}(n).(\kbrk)\bigr)
        \Bigr)} \nonumber
    \\
    &\ssec{\Sswitch}.(\kerr) \DEF
    \bcup{n}{
      \Bigl(\vvec{\ENV}{n}{e}.(\knrm)) \circ
        \wwec{\lslist{k}}(n).(\kerr)\Bigr)}
        \cup \vvec{\ENV}{n}{e}.(\kerr)
  \\  
  &\ssec{\Sswitch}.(\kbrk) \DEF \oset 
  \\
  &\begin{aligned}
    \ssec{\Sswitch}.(\kX) \DEF\;
    &\bcup{n}{
    \Bigl(\vvec{\ENV}{n}{e}.(\knrm)) \circ
    \wwec{\lslist{k}}(n).(\kX)\Bigr)}, \\
    & \kforx \in \Set{\krtn, \kcll, \kgto, \kfin, \kinf}
  \end{aligned}
\end{align*}

\paragbf{Unstructured branching} The denotation of a \textbf{goto} statement and a labeled statement satisfy the following equations, where Equation~(\ref{eqn:goto_nrm}) and~(\ref{eqn:goto_other}) reflect the essence of the name \emph{broken} for \(\ssec{s}\).
\begin{align}
  & \sosec{\Sgoto}{gto}  \DEF \Set{(\sigma, \knil, (j, \sigma)) | \sigma \in \kstate} \label{eqn:goto_nrm}
  \\
  & \ssec{\Sgoto}.(\kX)  \DEF \oset, \kforx \in \Set{\knrm, \kbrk, \kctn, \krtn, \kerr, \kcll, \kfin, \kinf } \label{eqn:goto_other}
  \\
  & \ssec{\Slabel}.(\kX) \DEF \ssec{s}.(\kX), \kforx \in \Full
\end{align}

\subsubsection{Labeled Broken Semantics of Clight Statements}
We use $\llec{s}(j)$ to  denote the semantics of executing the statement \(s\) from the given label \(j\). If there is no such a labeled statement inside \(s\), then $\llec{s}(j)$ is assigned the empty denotation (i.e., each field of which is the empty set). 
\begin{align*}
& \llec{\Sseq}(j) \DEF  
\big(\llec{s_1}(j)\fatsemi \ssec{s_2}\big)
\Cup \llec{s_2}(j)
\\
& \llec{\Sif}(j) \DEF \llec{s_1}(j) \Cup \llec{s_2}(j)
\\
& N_1^j = (\llec{s_1}(j).(\knrm) \cup \llec{s_1}(j).(\kctn))
  \qquad N_{12}^j = N_1^j \circ \ssec{s_2}.(\knrm)
\\
&\begin{aligned}
    \llec{\Sloop}(j).(\knrm) \DEF\;
    &\llec{s_1}(j).(\kbrk) \cup  N_1^j \circ \ssec{s_2}.(\kbrk) \cup  N_{12}^j \circ \ssec{\Sloop}.(\knrm) \ \cup
    \\
    & \llec{s_2}(j).(\kbrk) \cup
   \llec{s_2}(j).(\knrm) \circ \ssec{\Sloop}.(\knrm)
\end{aligned}
\\
&\begin{aligned}
    \llec{\Sloop}(j).(\kerr) \DEF\; 
   &  \llec{s_1}(j).(\kerr) \cup N_1^j \circ  \bigl(\nabla \ssec{s_2}.(\kctn) \cup (\ssec{s_2}\fatsemi \ssec{\Sloop}).(\kerr) \bigr)
   \ \cup 
  \\ 
  &(\nabla \llec{s_2}(j).(\kctn))
   \cup \bigl(\llec{s_2}(j) \fatsemi \ssec{\Sloop}\bigr).(\kerr)
\end{aligned}
\\
&\begin{aligned}
    \llec{\Sloop}(j).(\kbrk) = \llec{\Sloop}(j).(\kctn)  \DEF \oset 
\end{aligned}
\\
&\begin{aligned}
  \llec{\Sloop}(j).(\kX) \DEF\;
  &\llec{s_1}(j).(\kX) \cup N_1^j \circ \bigl((\ssec{s_2} \fatsemi \ssec{\Sloop}).(\kX)\bigr)  \ \cup
  \\ 
  & \bigl(\llec{s_2}(j) \fatsemi \ssec{\Sloop}\bigr).(\kX),\;\!\!\,
  \kforx \in \Set{\krtn, \kcll, \kgto, \kfin, \kinf}
\end{aligned}
\\
&\begin{aligned} %
    &\llec{\kLSnil}(j).(\kX) \DEF \oset, \kforx \in \Full
    \\
    &\llec{(n_1, s_1);\lslist{2}}(j) \DEF
     \bigl(\llec{s_1}(j)\fatsemi\ssec{\lslist{2}}\bigr) \Cup \llec{\lslist{2}}(j)
\end{aligned}
\\
& \llec{\Sswitch}(j).(\knrm) \DEF \llec{\lslist{k}}(j).(\knrm) \cup \llec{\lslist{k}}(j).(\kbrk) 
\\
& \llec{\Sswitch}(j).(\kbrk)\DEF \oset
\\
& \llec{\Sswitch}(j).(\kX) \DEF \llec{\lslist{k}}(j).(\kX), \kX \in \{\kctn, \krtn, \kerr, \kcll, \kgto, \kfin, \kinf \}
\\
&\llec{\Slabel}(j') \DEF \ssec{s} \qquad\,\, \text{ if } j' = j 
\\
&\llec{\Slabel}(j') \DEF \llec{s}(j') \ \ \text{ if } j' \neq j 
\\
& \begin{aligned}
   \textbf{Otherwise }
   &(\text{i.e., } s = \Sskip \mid \Scall \mid 
    \Sset \mid \Sbreak \mid \Scontinue \mid \Sreturn \mid \Sgoto)
   \\
   & \llec{s}(j).(\kX) = \oset, \kforx \in \Full
\end{aligned}
\end{align*}

\subsubsection{Intact semantics of Clight statements}We define the intact semantics of Clight statement \(\hhec{s}\) by composing the broken semantics \(\ssec{s}\) and labeled broken semantics \(\llec{s}\) systematically. To achieve this, we use \(\unllec{s}\) to denote the uncurried form of \(\llec{s}(j)\), such that, for instance, 
\begin{align*}
    & \unllec{s}.(\kgto) \subseteq (\klabel \times \kstate) \times \ktrace \times (\klabel \times \kstate) 
    \\
    &\llec{s}(j).(\kgto) \subseteq \kstate \times \ktrace \times (\klabel \times \kstate) 
    \\
    &\unllec{s}.(\kgto) \DEF \Set{\bigl((j, \sigma), \tau, (j', \sigma')\bigr) \mid \bigl(\sigma, \tau, (j', \sigma')\bigr) \in \llec{s}(j).(\kgto))} 
\end{align*}
The intact semantics of Clight statements are then defined by the following equations, where \(\labelsof{s}\) denotes the set of statement labels that have appeared in $s$ in the form of ``\(j: ...\)''.
\begin{align*}
    & \hhec{s}.(\kX) \DEF \ssec{s}.(\kX) \cup \big(\sosec{s}{gto} \circ
         \unllec{s}.{(\kgto)}^* \circ \unllec{s}.{(\kX)} \big),\;
         \kX \!\in\! \{\knrm, \kbrk, \kctn, \krtn, \kcll\}
    \\
    & \hohec{s}{err} \DEF \sosec{s}{err} \cup \Big(\sosec{s}{gto} \circ
         \unllec{s}.{(\kgto)}^* \circ \bigl(\unllec{s}.{(\kwd{err})} \cup \exlabel{s}\bigr) \Big)
    \\
    & \qquad \qquad \quad\;\; \text{ where }\exlabel{s} = \Set{((j, \sigma), \knil) | j \notin \labelsof{s} \land \sigma \in \kstate}
    \\
    & \begin{aligned}
        &\hhec{s}.(\kfin) \DEF
       \\
       &\quad\ \ \ssec{s}{.(\kfin)} \cup \sosec{s}{gto} \circ \unllec{s}.{(\kgto)}^* \circ \big( 
            \unllec{s}.(\kfin) \cup (\triangle \unllec{s}.(\kgto))^\infty \big)
    \end{aligned}
    \\
    &\begin{aligned}
        &\hhec{s}.(\kinf) \DEF \textbf{let } 
        G = \unllec{s}.{(\kgto)}
        \textbf{ in}
        \\
        &\quad\ \ \ssec{s}.(\kinf)  \cup \sosec{s}{gto} \circ \Big(\bigl((\triangle G)^* \circ \blacktriangle G \bigr)^*\circ (\triangle G)^* \circ \unllec{s}.{(\kinf)} \cup \bigl((\triangle G)^* \circ \blacktriangle G\bigr)^\infty\Big)
    \end{aligned}
\end{align*}

\subsection{Denotational Semantics of Clight Functions}
A Clight function \(F\) is defined by Equation~(\ref{eqn:clit_funct_syntax}).
We use \(\ssgc{(\ivar_F, F)}\) (of \FDenote) to denote the semantics of the Clight function \(F\) with the identifier \(\ivar_F\). This function-level semantics is parameterized by a global environment \(\gvar\) and a semantic oracle \(\chi\) (\(\subseteq \Tnee\)) for specifying the terminating behavior of the functions invoked by \(F\), and satisfies the following equations, where $\fentry$ intuitively represents the setup of a function’s internal initial state and local environment at entry, 
while $\varsfree$ represents the release of its local resources at exit.
\begin{align*}
    &\sosgc{(\ivar_F, F)}{dom} \triangleq \Set{\ivar_F}, \text{ i.e., the domain of } F \text{ is a singleton set with the unique element } \ivar_F
    \\ 
    &\begin{aligned} %
       \sosgc{(\ivar_F, F)}{nrm} \triangleq 
       \Set{
       \!((\ivar_F, \vec{v}, m), \tau, (m', v')) \!|\! 
       \begin{aligned}
           & \exists \evar\; \sigma, \fentry(\gvar, m, F, \vec{v}) = \ksome(\evar, \sigma) \land
           \\
           & \quad\;\;\,  \exists \sigma', \ENV = (\gvar, \evar) \land 
             \bolds{\big(}(\sigma, \tau, \sigma', v') \in \hohec{s_F}{rtn} \;\! \vee
           \\
           & \qquad \qquad \; (\sigma, \tau, \sigma') \in \hohec{s_F}{nrm} \land v' = \kundef \bolds{\big)} \;\!\land
           \\ 
           & \quad\;\,\, \qquad \varsfree(\gvar, \evar, \sigma') = \ksome(m') 
       \end{aligned}
       \!}
    \end{aligned}
    \\ 
    &\begin{aligned} %
       \sosgc{(\ivar_F, F)}{err} \triangleq 
       \Set{ 
       ((\ivar_F, \vec{v}, m), \tau) | 
       \begin{aligned}
           &  (1)\,\fentry(\gvar, m, F, \vec{v})  = \knone \land \tau = \knil; \textbf{ or}
           \\
           & (2)\,\exists \evar\, \sigma, \fentry(\gvar, m, F, \vec{v}) = \ksome(\evar, \sigma) \;\, \land
           \\
           & \qquad\quad\;\; \ENV = (\gvar, \evar) \land (\sigma, \tau) \in \hohec{s_F}{err}; \textbf{ or}
           \\ 
           & (3)\,\exists \evar\, \sigma, \fentry(\gvar, m, F, \vec{v}) = \ksome(\evar, \sigma) \;\, \land
           \\
           & \qquad\quad\;\;  \exists \sigma', \ENV = (\gvar, \evar) \land 
             \bolds{\big(}(\sigma, \tau, \sigma') \in \hohec{s_F}{brk} \ \vee
           \\
           & \qquad \qquad \qquad \quad\,  (\sigma, \tau, \sigma') \in \hohec{s_F}{ctn}  \bolds{\big)}; \textbf{ or}
           \\ 
           & (4)\,\exists \evar\, \sigma, \fentry(\gvar, m, F, \vec{v}) = \ksome(\evar, \sigma) \;\, \land
           \\
           & \qquad\quad\;\;  \exists \sigma'\, v', \ENV = (\gvar, \evar) \land 
             \bolds{\big(}(\sigma, \tau, \sigma', v') \in \hohec{s_F}{rtn} \ \vee
           \\
           & \qquad \qquad \qquad \,  (\sigma, \tau, \sigma') \in \hohec{s_F}{nrm} \land v' = \kundef \bolds{\big)} \ \land
           \\ 
           & \qquad \qquad \qquad \; \varsfree(\gvar, \evar, \sigma') = \knone 
       \end{aligned}
       }
    \end{aligned}
    \\ 
    &\begin{aligned} %
       \sosgc{(\ivar_F, F)}{cll} \triangleq 
       \Set{\big((\ivar_F, \vec{v}, m), \tau, q\big) | 
       \begin{aligned}
           &  \exists \evar\, \sigma, \fentry(\gvar, m, F, \vec{v}) = \ksome(\evar, \sigma) \ \land
           \\
           & \quad\quad \; \ENV = (\gvar, \evar) \land (\sigma, \tau, q) \in \hohec{s_F}{cll}  
       \end{aligned}
       }
    \end{aligned}
    \\ 
    &\begin{aligned} %
       \sosgc{(\ivar_F, F)}{\kX} \triangleq 
       \Set{\:\,\big((\ivar_F, \vec{v}, m), \tau\big)\;\, | 
       \begin{aligned}
           &  \exists \evar\, \sigma, \fentry(\gvar, m, F, \vec{v}) = \ksome(\evar, \sigma) \ \land
           \\
           & \qquad\; \ENV = (\gvar, \evar) \land (\sigma, \tau) \in \hohec{s_F}{\kX}  
       \end{aligned}
       }
    \end{aligned}
\end{align*}
 where \(\kX \in \Set{\kfin, \kinf}\). It is worth noting that the aborting behavior \(\ssgc{(\ivar_F, F)}.(\kerr)\) accounts for four possible cases:
(1) abnormal initialization at function entry;
(2) abnormal execution within the function body;
(3) premature exit of the function body via \textbf{break} or \textbf{continue};
(4) abnormal deallocation of local variables at function exit.
Besides, aborting and diverging behaviors of a function do not include those of its callees. Instead, we record the function’s call behavior explicitly, so that such cases can be handled later when defining the semantics of open modules.

\subsection{Semantics of Open Modules}
\label{subsec:module_sem}
 \paragbf{Definition of open module} An open module is modeled as a collection of Clight functions, i.e.,  \(M \DEF (id_1, F_1); ...; (id_n, F_n)\). Each \(F_i\) may call functions internal or external to module \(M\) (i.e., \emph{open}).

\paragbf{Semantic notations}  We use \(\ssgc{M}\) (of \FDenote) to denote the semantics of open Clight modules, where the parameter \(\gvar\) denotes a global environment, and \(\chi\) \((\subseteq \Tnee)\) is a semantic oracle for specifying the terminating behavior of functions external to module \(M\).
  
  Before showing the semantic equations, we overload the symbol ``\(\Cup\)'' to denote the union of function-level denotations, and define \(\ppgc{M}\) to represent the combined denotation of all functions defined in \(M\). Formally, given function-level denotations \(D_1\) and \(D_2\) of \FDenote,  we have:
\begin{align}
    & (D_1 \Cup D_2).(\kX) \DEF D_1.(\kX) \cup D_2.(\kX), \kforx \in \Set{\kdom, \knrm, \kerr, \kcll, \kfin, \kinf}
    \\
    &\ppgc{M} \triangleq \ssgc{(id_1, F_1)} \Cup \cdots \Cup\ssgc{(id_n, F_n)}\text{, where } M = (id_1, F_1);...;(id_n, F_n)
\end{align}
\paragbf{Semantic equations} The semantics of open Clight modules then satisfies the following equations:
\begin{align*}
  & \sosgc{M}{dom} \triangleq \ppgc{M}.(\kdom)
  \\
  & \sosgc{M}{nrm} \triangleq 
    \mu  \chi_0. \popg{\chi_0 \cup \chi}{M}{nrm}
  \\
  &\sosgc{M}{err} \triangleq
     \popg{\hatchi}{M}{cll}^* \circ \popg{\hatchi}{M}{err}    
\\
&\kecall(X) \triangleq \{(id, \vec{v}, m) \mid id \not\in {X} \}   \text{ for } X \subseteq \kid
  \\
  &\sosgc{M}{cll} \triangleq
     \popg{\hatchi}{M}{cll}^+ \circ \ktest \bigl(\kecall(\popgc{M}{dom})\bigr)
  \\
  &\sosgc{M}{\ttfin} \triangleq
      \popg{\hatchi}{M}{cll}^* \circ
	  \popg{\hatchi}{M}{\ttfin} \cup
	  (\triangle\popg{\hatchi}{M}{cll})^\infty
  \\
  &\begin{aligned}
      \sosgc{M}{\ttinf} \triangleq \
      &\bigl((\triangle\popg{\hatchi}{M}{cll})^* \circ \blacktriangle\popg{\hatchi}{M}{cll}\bigr)^* \circ
      (\triangle\popg{\hatchi}{M}{cll})^* \circ \popg{\hatchi}{M}{\ttinf}
       \ \cup \\
       &\bigl((\triangle\popg{\hatchi}{M}{cll})^* \circ \blacktriangle\popg{\hatchi}{M}{cll}\bigr)^\infty,
       \text{ where } \hatchi = \sosgc{M}{nrm} \cup \chi
  \end{aligned}
\end{align*}
 Here, the set \(\kecall(X)\) denote the set of external queries w.r.t. the given identifier set \(X\), so as to capture calls to functions outside the module. Besides, once the normal semantics \(\ssgc{M}.(\knrm)\) is defined, we extend the oracle \(\chi\) (specifying the terminating behavior of external functions) to \(\hatchi = \sosgc{M}{nrm} \cup \chi\) so as to cover the terminating behavior of both internal and external functions.

\subsection{Semantic Linking}
\label{subsec:sem_link}
We use \((\semlink{M_1}{M_2})_{\chi}^\gvar\) to denote the semantic linking of two open Clight modules, where \(\gvar\) is the same global environment, and the semantic oracle \(\chi\) \((\subseteq \Tnee)\) specifies the terminating behavior of function calls external to modules \(M_1\) and \(M_2\). Intuitively, semantic linking merges the domains and behaviors of the two modules while continuing to rely on the oracle $\chi$ for any calls external to both. The denotation of semantic linking satisfies the following equations:
\begin{align*}
   &\kokgc{\semlink{M_1}{M_2}}{dom} \DEF \sosgc{M_1}{dom} \cup \sosgc{M_2}{dom} 
   \\
    &\begin{aligned}
      \kokgc{\semlink{M_1}{M_2}}{nrm} \triangleq
        \mu \chi_0. \bigl(\sosg{\chi_0 \cup \chi}{M_1}{nrm} \cup \sosg{\chi_0 \cup \chi}{M_2}{nrm}\bigr)
    \end{aligned}
    \\ %
    &\begin{aligned}
      \kokgc{\semlink{M_1}{M_2}}{err} \triangleq
      \kokg{\dotchi}{\semcup{M_1}{M_2}}{cll}^* \circ
      \kokg{\dotchi}{\semcup{M_1}{M_2}}{err}
    \end{aligned}  
    \\ %
    &\begin{aligned}
      \kokgc{\semlink{M_1}{M_2}}{cll} \triangleq
        \kokg{\dotchi}{\semcup{M_1}{M_2}}{cll}^+ \circ \ktest\bigl(\kecall\bigl(\sosgc{M_1}{dom} \cup \sosgc{M_2}{dom}\bigr)\bigr)
    \end{aligned}
    \\
    &\begin{aligned}
    \kokgc{\semlink{M_1}{M_2}}{\ttfin}
    \DEF 
     \kokg{\dotchi}{\semcup{M_1}{M_2}}{cll}^* \circ \kokg{\dotchi}{\semcup{M_1}{M_2}}{\ttfin}
    \cup   \bigl( \triangle  C \bigr)^\infty 
    \end{aligned}
    \\
    &\begin{aligned}
    \kokgc{\semlink{M_1}{M_2}}{\ttinf}
    \DEF\; &\big((\triangle C)^* \circ \blacktriangle C\big)^* \circ
      (\triangle C)^* \circ \kokg{\dotchi}{\semcup{M_1}{M_2}}{\ttinf}
       \cup \big((\triangle C)^* \circ \blacktriangle C\big)^\infty
    \\
    & \text{where } C = \kokg{\dotchi}{\semcup{M_1}{M_2}}{cll}  \text{, and }
    \dotchi = \kokgc{\semlink{M_1}{M_2}}{nrm} \cup \chi
    \end{aligned}
\end{align*}

\begin{theorem} [Semantic adequacy]
\label{app:thm_ss_equiv}
    For any Clight modules $M_1$ and $M_2$, 
    $\semlink{M_1}{M_2} = \bracket{M_1 + M_2}$.
\end{theorem}

\subsection{Csharpminor and Cminor}
\label{subsec:Csminor}

\paragbf{Syntax of Csharpminor}  Csharpminor is an \emph{untyped} low-level imperative language featured with infinite loops, control flow blocks, and premature exits from blocks, which are defined as follows.
\begin{align*}
& \text{(Statements)} && u \DEF\
\begin{aligned}[t]
 &\textbf{skip} \mid
  u_1;u_2 \mid
   \Cif{u} \mid
   \Ccall \mid 
   \Cgoto \mid 
   \Clabel{u} \mid
   \Sset \mid
   \\ 
  &\Cloop{u} \mid 
  \textbf{block}\{u\} \mid
  \textbf{exit}(n) \mid
  \textbf{return} \ ?e \mid
  \textbf{switch}\ b \ e\ \Vec{(?n_k, u_k)} \mid ...  
\end{aligned}
\\
&\text{(Functions)} && F \DEF 
\begin{aligned}[t]
(\Vec{x_k}, \Vec{y_k}, \Vec{z_k}, u_F, sig)  
\end{aligned}
\end{align*}
A Csharpminor function \(F\) is composed of a list of parameter names \(\Vec{x_k}\), a list of temporary variables \(\Vec{x_k}\), a list of local variables \(\Vec{z_k}\) (whose addresses are taken in \(F\)), a statement \(u_F\) representing the function body, and a signature $sig$ specifying the types of the arguments and return value of \(F\).

\paragbf{Syntax of Cminor} Cminor functions is defined by the parameter list \(\Vec{x_k}\), the temporary variables \(\Vec{x_k}\), the size of memory block for stack space \(stacksize\), the function body \(t_F\), and the signature of \(F\).  
 In Cminor, Csharpminor local variables become sub-blocks of the Cminor stack memory block. 
\begin{align*}
& \text{(Statements)} && t \triangleq\
\begin{aligned}[t]
 &\textbf{skip} \mid
  t_1;t_2 \mid
  \Cif{t} \mid
  \Ccall  \mid
  \Cgoto  \mid
  \Clabel{t}  \mid
  \Sset \mid
  \\ 
  &\textbf{loop } t \mid 
  \textbf{block}\{t\} \mid
  \textbf{exit}(n) \mid
  \textbf{return} \ ?e \mid
  \textbf{switch}\ b \ e\ table \ \mdefault \mid ...  
\end{aligned}
\\
& \text{(Functions)}
&&\begin{aligned}
  F \DEF (\Vec{x_k}, \Vec{y_k}, stacksize, t_F, sig)
\end{aligned}
\end{align*}

\paragbf{Semantic domain}  Csharpminor and Cminor enjoy the same form of signatures for statement denotations, since they adopt the same block-exit mechanism to model structured control flow.
\begin{lstlisting}[language=Coq]
    Record Cshm/Cmin.Denote: Type := {
      nrm: state -> trace -> state -> Prop;              (* nrm $\subseteq$ state $\times$ trace $\times$ state ----------$\,$*)
      blk: nat -> state -> trace -> state -> Prop;        (* blk $\subseteq$ nat $\times$ state $\times$ trace $\times$ state ----*)
      rtn: state -> trace -> state -> val -> Prop;        (* rtn $\subseteq$ state $\times$ trace $\times$ state $\times$ val ----*)
      err: state -> trace -> Prop;                      (* err $\subseteq$ state $\times$ trace -------------------*)
      cll: state -> trace -> cquery -> Prop;             (* cll $\subseteq$ state $\times$ trace $\times$ cquery ---------$\,$*)
      gto: state -> trace -> (label $\times$ state) -> Prop;     (* gto $\subseteq$ state $\times$ trace $\times$ (label $\times$ state) *)
      fin_dvg: state -> trace -> Prop;                  (* fin_dvg $\subseteq$ state $\times$ trace ---------------*)
      inf_dvg: state -> itrace -> Prop                  (* inf_dvg $\subseteq$ state $\times$ itrace --------------*)
    }.
\end{lstlisting}

The main difference between \kwd{Cshm}/\kwd{Cmin.Denote} and \ClitDenote lies in how they handle structured control flow. While \ClitDenote uses $\kbrk$ and $\kctn$, \kwd{Cshm}/\kwd{Cmin.Denote} 
instead employs $\kblk$; the meanings of the other sets remain the same as in Clight. 
The semantic infrastructure (i.e., states, queries, and replies) of Csharpminor and Cminor builds the same with that of Clight, as summarized in Table~\ref{tab:clight-sem-component}.  

The Cminorgen phase translates Csharpminor to Cminor, where CompCert coalesces the separate memory blocks allocated to local variables into a single block that represents the activation record as a whole.
We write $\ssem{\ENV}{\chi}{u}$ for the semantics of Csharpminor statements, where 
$\ENV = (\gvar, \evar)$; and write $\ssem{\ENV}{\chi}{t}$ for that of Cminor statements, where $\ENV = (\gvar, sp)$ with $sp$ denoting the stack pointer, as shown in Table \ref{tab:Cshm_cmin_notation}. Other notations (i.e., $\chi$, $\gvar$, and $\evar$) carry the same meaning as in Clight.  

The function-level denotational semantics of Csharpminor and Cminor, including those for functions, open modules, and semantic linking, are defined in the same way as in Clight.

\begin{table}[t]
\centering
\caption{Semantic notation for Csharpminor and Cminor.}
\label{tab:Cshm_cmin_notation}
\renewcommand{\arraystretch}{1.1}
\begin{tabular}{|l|c|l|l|l|}
\hline
Language     & Symbol       & Denotation    & Environment                        & Program state \\
\hline
\rule{0pt}{2.7ex} 
Csharpminor  & $\ssec{u}$        & $\CshmDenote$ & $\ENV = (\gvar, \evar)$      & $\sigma = (te, m)$ \\
\rule{0pt}{2.7ex} 
Cminor       & $\ssec{t}$        & $\CminDenote$ & $\ENV = (\gvar, sp)$           & $\sigma = (te, m)$ \\
\hline
\end{tabular}
\end{table}

\section{Behavior Refinement}
\label{sec:refinement}
  Compilation correctness is ubiquitously formulated as a behavior refinement relation, which ensures that every behavior exhibited by the target program is one possible behavior of the source. As previously discussed, we define the semantics of program statements, procedures, and open modules through a series of behavior sets, using extended KATs and suitable fixed-point theorems. The highlight of our theoretical contribution is proposing a novel algebraic structure to unify various behavior refinements between different behavior sets.
  In this section, we begin with typical examples of behavior refinement in \S\ref{subsec:ref_examples}, and then propose the novel refinement algebra in \S\ref{subsec:ref_algebra}$\sim$\S\ref{subsec:ra_instances}. We define Kripke relations for Cminorgen in \S\ref{subsec:krel-cminorgen}, and finally explain our design choices in \S\ref{subsec:design_choice}

\subsection{Examples of Behavior Refinement}
\label{subsec:ref_examples}

\begin{example} [Behavior refinement for a simple transformation]
\label{ex:bref_while}
Consider a simple transformation that converts a WHILE program $c$ into another WHILE program $\mathcal{T}(c)$ with sequenced \textbf{skip} statements removed. Because how states are represented before and after the transformation is unchanged, we can use the set inclusion relation to define the behavior refinement between the programs before and after transformation, i.e.,
\begin{align}
    \sem{\mathcal{T}(c)} \sqsubseteq \sem{c} \DEF 
    \sem{\mathcal{T}(c)}.(\knrm) \subseteq \sem{c}.(\knrm) \text{ and } \sem{\mathcal{T}(c)}.(\kdvg) \subseteq \sem{c}.(\kdvg)
\end{align}
For the transformation of type-safe and nondeterministic WHILE programs, the behavior refinement implies that compiler correctness allows the compiled program to exhibit fewer normal termination possibilities and fewer divergence possibilities. In other words, if the target program terminates normally, the source program must also be able to terminate normally; if the target program diverges, the source program must have a behavior that can diverge. The transformation correctness is then compositional in terms of sequential statements, in the sense that, for any $c_1$ and $c_2$,
\begin{align}
    \text{if } \sem{\mathcal{T}(c_1)} \sqsubseteq \sem{c_1} \text{ and } \sem{\mathcal{T}(c_2)} \sqsubseteq \sem{c_2}\text{, then } \sem{\mathcal{T}(c_1); \mathcal{T}(c_2)} \sqsubseteq \sem{c_1;c_2}.
    \label{bref_seqmono}
\end{align}
Its proof follows from the monotonicity of the composition and set union w.r.t. the set inclusion: 
    \begin{align*}
        \fof{}{\mathcal{T}(c_1);\mathcal{T}(c_2)}{nrm} &= \fof{}{\mathcal{T}(c_1)}{nrm} \circ \fof{}{\mathcal{T}(c_2)}{nrm}
            \subseteq \fof{}{c_1}{nrm} \circ \fof{}{c_2}{nrm} = \fof{}{c_1;c_2}{nrm}
        \\
        \fof{}{\mathcal{T}(c_1);\mathcal{T}(c_2)}{dvg} &= 
          \fof{}{\mathcal{T}(c_1)}{dvg} \cup
          \big(\fof{}{\mathcal{T}(c_1)}{nrm} \circ 
          \fof{}{\mathcal{T}(c_2)}{dvg}\big)
        \\
        &\subseteq\;\!\!
          \fof{}{c_1}{dvg} \cup
          \big(\fof{}{c_1}{nrm} \circ \fof{}{c_2}{dvg}\big) = \fof{}{c_1;c_2}{dvg}
        \tag*{\qed}
    \end{align*} 
The transformation correctness of \textbf{while} loops is compositional as well, in the sense that, for any $c$,
\begin{align}
    \text{if } \sem{\compile(c)} \refines
    \sem{c},
    \text{ then }
    \sem{\Cwhile{\compile(c)}} \refines \sem{\Cwhile{c}} 
\end{align}
Its proof similarly follows from the monotonicity of operators \(\circ\), \(\cup\), \(^*\), and \(^\infty\) w.r.t. the set inclusion.
    \begin{align*}
       \fof{}{\Cwhile{\compile(c)}}{nrm}
       & = \big(\ktest(\sem{b}.\ttt{(tts)}) \circ
         \sem{\compile(c)}.(\knrm)\big)^* \circ \ktest(\sem{b}.\ttt{(ffs)}) 
       \\
       & \subseteq \big(\ktest(\sem{b}.\ttt{(tts)}) \circ
         \sem{c}.(\knrm)\big)^* \circ \ktest(\sem{b}.\ttt{(ffs)})
       \\
       &= \fof{}{\Cwhile{c}}{nrm}
\\
        \fof{}{\Cwhile{\compile(c)}}{dvg}
          & = 
           \bigl(\ktest(\sem{b}.\ttt{(tts)}) \circ \sem{\compile(c)}.(\knrm)\bigr)^*
         \circ
           \bigl(\ktest(\sem{b}.\ttt{(tts)}) \circ \sem{\compile(c)}.(\kdvg)\bigr) 
         \\
         &\quad\; \cup
            \bigl(\ktest(\sem{b}.\ttt{(tts)}) \circ \sem{\compile(c)}.(\knrm)\bigr)^\infty
         \\
          & \subseteq 
           \bigl(\ktest(\sem{b}.\ttt{(tts)}) \circ \sem{c}.(\knrm)\bigr)^*
         \circ
           \bigl(\ktest(\sem{b}.\ttt{(tts)}) \circ \sem{c}.(\kdvg)\bigr) 
         \\
         & \quad\; \cup
            \bigl(\ktest(\sem{b}.\ttt{(tts)}) \circ \sem{c}.(\knrm)\bigr)^\infty 
         \\
         & =  \fof{}{\Cwhile{c}}{dvg}
         \tag*{\qed}
    \end{align*}
\end{example}

\begin{example}
\label{ex:bref_abort}
To take abortion into account, we extend the WHILE language to WHILEE, which is identical to WHILE except that it is not type-safe. In other words, WHILEE programs may abort (e.g., division by zero). We then use the following signature to capture WHILEE's program behavior:
\begin{lstlisting}[language=Coq]
    Record EDenote: Type := {                          
      nrm: state -> state -> Prop;            (* nrm $\subseteq$ state $\times$ state *)
      err: state -> Prop;                    (* err $\subseteq$ state *)
      dvg: state -> Prop                     (* dvg $\subseteq$ state *)
    }.
\end{lstlisting}
To allow for more flexible optimization, modern compilers usually do not guarantee behavioral consistency for programs with undefined behavior. That is, if the source program could abort, the compiled program may either terminate normally, diverge, or abort as well. In such cases, behavior refinement $\sem{\mathcal{T}(c)} \sqsubseteq \sem{c}$ (still consider the simple transformation in Example \ref{ex:bref_while}) can be redefined as follows, where $\fof{}{c}{err}$ ($\subseteq \kstate$) denotes potential aborting behaviors of program $c$.
\begin{align}
  \sem{\compile(c)} \refines \sem{c} \DEF \left\{ \,
  \begin{aligned}
    & \fof{}{\mathcal{T}(c)}{nrm} \subseteq \fof{}{c}{nrm} \cup (\fof{}{c}{err} \times \kstate)
    \text{ and }
    \\
    &\fof{}{\mathcal{T}(c)}{dvg} \subseteq \fof{}{c}{dvg} \cup \fof{}{c}{err} \text{ and }
    \\ 
    &\fof{}{\compile(c)}{err} \subseteq \fof{}{c}{err}
  \end{aligned} \right.
  \label{eq:bref_abort}
\end{align}
The transformation correctness of sequential statements like Formula~(\ref{bref_seqmono}) can be proved as follows.
    \begin{align}
        &\fof{}{\mathcal{T}(c_1);\mathcal{T}(c_2)}{nrm} = \fof{}{\mathcal{T}(c_1)}{nrm} \circ \fof{}{\mathcal{T}(c_2)}{nrm}
        \\
        &\qquad\subseteq
            \big(\fof{}{c_1}{nrm} \cup (\fof{}{c_1}{err} \times \kstate)\big)
             \circ
             \big(\fof{}{c_2}{nrm} \cup (\fof{}{c_2}{err} \times \kstate)\big) \label{proof:line2}
            \\
        &\qquad\subseteq
             \big(\fof{}{c_1}{nrm} \circ \fof{}{c_2}{nrm}\big) \cup
              (\fof{}{c_1}{err} \times \kstate) \cup
              \big(\fof{}{c_1}{nrm} \circ (\fof{}{c_2}{err} \times \kstate)\big)
           \\
        &\qquad\subseteq
             \big(\fof{}{c_1}{nrm} \circ \fof{}{c_2}{nrm}\big) \cup
              \Big(\big(\fof{}{c_1}{err} \cup 
              (\fof{}{c_1}{nrm} \circ \fof{}{c_2}{err})\big) \times \kstate\Big) \label{proof:line4}
        \\
        &\qquad= \fof{}{c_1;c_2}{nrm} \cup (\fof{}{c_1;c_2}{err} \times \kstate) 
    \nonumber
    \\
    &\fof{}{\compile(c_1);\compile(c_2)}{dvg} =
    \fof{}{\compile(c_1)}{dvg} \cup
       \bigl(\fof{}{\compile(c_1)}{nrm} \circ
        \fof{}{\compile(c_2)}{dvg}\bigr)
    \\
    &\qquad\subseteq
      \big(\fof{}{c_1}{dvg} \cup \fof{}{c_1}{err}\big)
         \cup
          \big(\fof{}{c_1}{nrm} \cup (\fof{}{c_1}{err} \times \kstate)\big)
         \circ 
         \big(\fof{}{c_2}{dvg} \cup \fof{}{c_2}{err}\big)
    \\ 
    &\qquad\subseteq 
      \bigl(\fof{}{c_1}{dvg} \cup \fof{}{c_1}{nrm} \circ \fof{}{c_2}{dvg}\bigr)
      \cup 
        \big(\fof{}{c_1}{err} \cup 
            (\fof{}{c_1}{nrm} \circ \fof{}{c_2}{err})\big)
    \\
     &\qquad=
      \fof{}{c_1;c_2}{dvg} \cup \fof{}{c_1;c_2}{err}
      \nonumber
    \\
    &\fof{}{\compile(c_1);\compile(c_2)}{err} =
    \fof{}{\compile(c_1)}{err} \cup
       \bigl(\fof{}{\compile(c_1)}{nrm} \circ
        \fof{}{\compile(c_2)}{err}\bigr)
    \\
    &\qquad\subseteq
    \fof{}{c_1}{err}
      \cup
      \big(\fof{}{c_1}{nrm} \cup (\fof{}{c_1}{err} \times \kstate)\big)
      \circ 
      \fof{}{c_2}{err}
    \\
    &\qquad\subseteq
    \fof{}{c_1}{err} \cup \bigl(\fof{}{c_1}{nrm} \circ \fof{}{c_2}{err} \bigr) = \fof{}{c_1;c_2}{err}
    \tag*{\qed}
\end{align}

\paragraph{Summary of existing proof system.}
In the preceding examples, we showed by hand how relational Kleene algebra can be used to verify the correctness of simple program transformations. This style of reasoning is supported by a well-developed proof system: the equational theory of Kleene algebra with tests (KAT). Beyond such manual reasoning, a rich body of research has developed automated proof techniques for regular expressions, relation algebra, and KAT. Classic approaches build on the notion of derivatives~\cite{DBLP:journals/jacm/Brzozowski64}, which enable the construction of automata directly from algebraic expressions. This idea has been extended to mechanized proofs of regular expression equivalence in Isabelle/HOL \cite{DBLP:journals/jar/KraussN12}, decision procedures for (in-)equivalence of regular expressions in \Coq \cite{DBLP:conf/RelMiCS/MoreiraPS12}, \Coq tools for KAT applied to while programs \cite{DBLP:conf/itp/Pous13}, unified decision procedures for regular expression equivalence \cite{DBLP:conf/itp/NipkowT14}, and symbolic algorithms for language equivalence and KAT \cite{DBLP:conf/popl/Pous15}. Together, these advances demonstrate the maturity of automated algebraic reasoning about relational (in-)equations.

\paragraph{Challenges for realistic compiler verification.}
Despite the expressiveness of relational Kleene algebra and the maturity of its automated reasoning, directly applying these techniques to realistic compiler verification remains challenging. Existing methods focus primarily on equational reasoning over algebraic structures, whereas compiler correctness demands reasoning about simulation relations between source and target programs. In other words, the correctness of a compiler often cannot be expressed simply in terms of set inclusion between source and target denotations. Instead, one must account for changes in state representations across abstraction layers—most notably the transformation of memory models—so that correctness is proved not by inclusion alone but by carefully designed state-matching relations. We illustrate this point with the following examples.

\end{example}

Before showing general definitions, we first consider the compilation correctness of Cshmgen.
\begin{example} [Behavior refinement for Cshmgen of CompCert]
\label{ex:bref_Cshmgen}
Cshmgen is mainly responsible for translating structured C control flow into simpler primitives and erasing Clight types. In this process, the representation of program states remains unchanged. Suppose that a Clight statement \(s\) is transformed to Csharpminor by Cshmgen, we can then define the behavior refinement \(\sem{t} \precsim \sem{s}\):
\begin{align*}
&\begin{aligned} %
    \forall \sigma\ \tau\ \sigma',\ 
    (\sigma, \tau, \sigma') \in \fof{}{t}{nrm} \Rightarrow
    \big((\sigma, \tau, \sigma') \in \fof{}{s}{nrm}\big) \vee
    \left(\exists \tau_0, (\sigma, \tau_0) \in \fof{}{s}{err} \wedge \tau_0 \leqslant_T \tau\right);
\end{aligned}
\\
&\begin{aligned} %
    \forall \sigma\ \tau,\ 
    (\sigma, \tau) \in \fof{}{t}{err} \Rightarrow 
    \left(\exists \tau_0, (\sigma, \tau_0) \in \fof{}{s}{err} \wedge \tau_0 \leqslant_T \tau\right);
\end{aligned}
\\
&\begin{aligned} %
    \forall \sigma\ \tau,\ 
    (\sigma, \tau) \in \fof{}{t}{\ttfin} \Rightarrow 
    (\sigma, \tau) \in \fof{}{s}{\ttfin} \vee
    \left(\exists \tau_0, (\sigma, \tau_0) \in \fof{}{s}{err} \wedge \tau_0 \leqslant_T \tau\right);
\end{aligned}
\\
&\begin{aligned} %
    \forall \sigma\ \tau,\ 
    (\sigma, \tau) \in \fof{}{t}{\ttinf} \Rightarrow 
    (\sigma, \tau) \in \fof{}{s}{\ttinf} \vee
    \left(\exists \tau_0, (\sigma, \tau_0) \in \fof{}{s}{err} \wedge \tau_0 \leqslant_T \tau\right).
\end{aligned}
\end{align*}
where $\tau_0 \leqslant_T \tau$ means that $\tau_0$ is the prefix of $\tau$. 
{This example illustrates the way how observable event traces are preserved for realistic compilation, especially when undefined behaviors (UB) are raised in the source program, which conforms to CompCert's prefix-preserving view of UB.
}

\end{example}

We now generalize the above definition to capture the correctness of more complex phases of CompCert, where source and target program states may differ in representation and are related through a state-matching relation. The following example illustrates this generalization.
\begin{figure}[t]
    \begin{subfigure}[b]{.48\linewidth}
    \centering
    \includegraphics[width=0.8\linewidth]{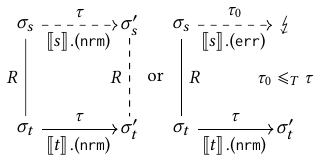}
    \caption{Refinement of terminating behavior}
    \label{subfig:bref-nrm1}
    \end{subfigure} 
    \begin{subfigure}[b]{.48\linewidth}
        \flushright
    \includegraphics[width=0.8\linewidth]{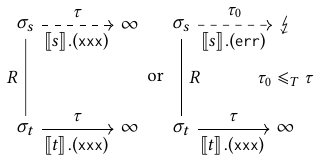}
    \centering
    \caption{Refinement of diverging behavior}
    \label{subfig:bref-dvg1}
    \end{subfigure}   
    \caption{Behavior refinement of termination and divergence for CompCert, where $\kX \in \Set{\kfin, \kinf}$.}
    \label{fig:compcert_bref}
    \Description{This figure illustrates behavior refinement in CompCert for both normal termination and divergence. Subfigure (a) shows refinement of terminating behavior: each terminating execution of the target program must correspond either to a related terminating execution of the source program, preserving the observable trace and the state-matching relation, or to a premature abort of the source program, marked by a lightning symbol. Subfigure (b) shows refinement of diverging behavior: each diverging execution of the target program, with either finite or infinite observable trace, must correspond either to a related diverging execution of the source program or again to a premature source abort.}
\end{figure}

\begin{example} [Behavior refinement for general phases of CompCert]
\label{ex:bref_compcert}
Suppose that a compilation phase of CompCert transforms a source statement \(s\) into the target statement \(t\).
For a given state-matching relation \(R\), we define their behavior refinement, denoted by \(\sem{t} \precsim_R \sem{s}\), as follows:
\begin{align*}
&\begin{aligned} %
    \forall \sigma_t\ \tau\ \sigma'_t\ &\sigma_s,\
    (\sigma_t, \tau, \sigma'_t) \in \fof{}{t}{nrm} \Rightarrow
    (\sigma_s, \sigma_t) \in R \Rightarrow 
    \\
    &\big(\exists \sigma'_s, (\sigma_s, \tau, \sigma'_s) \in \fof{}{s}{nrm}
    \wedge (\sigma'_s, \sigma'_t) \in R\big) \vee
    \left(\exists \tau_0, (\sigma_s, \tau_0) \in \fof{}{s}{err} \wedge \tau_0 \leqslant_T \tau\right);
\end{aligned}
\\
&\begin{aligned} %
    \forall \sigma_t\ \tau\ \sigma_s, \,
    &(\sigma_t, \tau) \in \fof{}{t}{err} \Rightarrow 
    (\sigma_s, \sigma_t) \in R \Rightarrow 
    \left(\exists \tau_0, (\sigma_s, \tau_0) \in \fof{}{s}{err} \wedge \tau_0 \leqslant_T \tau\right);
\end{aligned}
\\
&\begin{aligned} %
    \forall \sigma_t\ \tau\ \sigma_s, \,
    &(\sigma_t, \tau) \in \fof{}{t}{\ttfin} \Rightarrow 
    (\sigma_s, \sigma_t) \in R \Rightarrow 
    \\
    &(\sigma_s, \tau) \in \fof{}{s}{\ttfin} \vee \left(\exists \tau_0, (\sigma_s, \tau_0) \in \fof{}{s}{err} \wedge \tau_0 \leqslant_T \tau\right);
\end{aligned}
\\
&\begin{aligned} %
    \forall \sigma_t\ \tau\ \sigma_s, \,
    &(\sigma_t, \tau) \in \fof{}{t}{\ttinf} \Rightarrow 
    (\sigma_s, \sigma_t) \in R \Rightarrow 
    \\
    &(\sigma_s, \tau) \in \fof{}{s}{\ttinf} \vee \left(\exists \tau_0, (\sigma_s, \tau_0) \in \fof{}{s}{err} \wedge \tau_0 \leqslant_T \tau\right)
\end{aligned}
\end{align*}
Fig.~\ref{fig:compcert_bref} illustrates behavior refinement for both terminating and diverging executions. It states that for every terminating (resp. diverging) behavior of the target program, either there may exist consistent terminating (resp. diverging) behavior in the source end, or the source program may abort (denoted as $\lightning$) prematurely. Once the source program may abort, the compiled target program could do anything unpredictable. 
Furthermore, though not depicted in Fig.~\ref{fig:compcert_bref}, the definition also requires that if the target program may abort, then the corresponding source program may abort as well. This rules out compilers that intend to translate a safe program into one that may aborts.
It is also worth noting that refinement between behavior sets entails refinement between the triggered observable behaviors, thereby demonstrating the same compilation correctness as in CompCert.

Besides, the correctness of each compilation phase can be proved with different state-matching relations and can be vertically composed to form end-to-end correctness of the whole compiler.
This is demonstrated by the following theorem, where \(L_1, L_2\) and \(L_3\) are the semantics of statements.
\begin{theorem}
    For any matching relations \(R_1\) and \(R_2\), if 
   \(L_3 \precsim_{R_2} L_2\) and \(L_2 \precsim_{R_1} L_1\), then \(L_3 \precsim_{R_2 \circ R_1} L_1\).
\end{theorem}

\end{example}

However, the behavior refinement defined in Example~\ref{ex:bref_compcert} is inherently limited to closed settings, i.e., the verification of whole-program compilation, as is the case for the original CompCert\footnote{Extended by \citeN{DBLP:conf/popl/KangKHDV16}, CompCert supports verified separate compilation by the same compiler from version 2.7. }.

\paragbf{Toward verified compositional compilation (VCC)}
Verified compositional compilation requires two dimensions of compositionality in open settings:
(i)  \emph{horizontal compositionality} requires that the compilation correctness of individual modules can be composed to obtain compilation correctness of the linked whole program;
(ii) \emph{vertical compositionality} demands that correctness of each compilation phase can be composed to yield an end-to-end correctness across the compilation chain.

Several works stemming from CompCert have moved in this direction. Early attempts such as CompComp~\cite{DBLP:conf/popl/StewartBCA15} and CompCertM~\cite{DBLP:journals/pacmpl/SongCKKKH20} extend CompCert with interaction semantics~ \cite{DBLP:conf/esop/BeringerSDA14} and employ open simulations to prove compilation correctness of open modules. More recent developments, notably CompCertO~\cite{DBLP:conf/pldi/KoenigS21} and the work of \citeN{DBLP:journals/pacmpl/ZhangWWKS24}, introduce Kripke logical relations to express the correspondence between open module interactions. We will return to these developments in more detail later.

To support VCC in our denotation-based framework, we leverage the notion of Kripke relations and the Kripke interface relations (also called simulation conventions) from CompCertO. We formulate the denotation-based behavior refinement for open functions and modules as follows.

\begin{definition} [Kripke relations]
    Given sets \(A_1\), \(A_2\), and \(W\), a \emph{Kripke relation} 
\(R : W \rightarrow \{X \mid X \subseteq A_1 \times A_2 \}\) 
is a family of relations indexed by the set of \emph{Kripke worlds} \(W\), 
written as \(R: \mathcal{K}_W(A_1, A_2)\).
\end{definition}

 \emph{Kripke interface relations} (KIR) are used to relate the source and target language interfaces.

\begin{definition} [Kripke interface relations, or KIR]
\label{def:kir}
A Kripke interface relation between two language interfaces 
\(\pair{\kquery_1,\kreply_1}\) and \(\pair{\kquery_2,\kreply_2}\) is a tuple
\(
B = \pair{W_B, \mcal{\leadsto}{B}, \msal{R}{q}{B}, \msal{R}{r}{B}},
\)
where $\mcal{\leadsto}{B}$ (\( \subseteq W_B \times W_B\)) is an accessibility 
relation between worlds, and \(\msal{R}{q}{B}: \mathcal{K}_{W_B}(\kquery_1,\kquery_2)\) and 
\(\msal{R}{r}{B}: \mathcal{K}_{W_B}(\kreply_1,\kreply_2)\) are Kripke relations used for relating 
queries and replies, respectively.
\end{definition}

According to CompCertO~\cite{DBLP:conf/pldi/KoenigS21}, Kripke interface relations can be vertically composed as follows: 
given Kripke interface relations 
\(A = \pair{W_A, \mcal{\leadsto}{A}, \msal{R}{q}{A}, \msal{R}{r}{A}}\) and 
\(B = \pair{W_B, \mcal{\leadsto}{B}, \msal{R}{q}{B}, \msal{R}{r}{B}}\), 
$A {\bcirc} B \DEF \pair{W_A \times W_B, \leadsto_{AB},\msal{R}{q}{A} \bcirc \msal{R}{q}{B}, \msal{R}{r}{A} \bcirc \msal{R}{r}{B}}$, 
where for any worlds \( w_A, w_A' \in W_A\), and \( w_B,\,w_B' \in W_B\), 
\begin{gather*}
    (w_A, w_B) \leadsto_{AB} (w_A', w_B') \iff 
    w_A \leadsto_{A} w_A' \wedge
    w_B \leadsto_{B} w_B';
    \\
     \msal{R}{q}{A} \bcirc \msal{R}{q}{B}(w_A, w_B) \DEF \msal{R}{q}{A}(w_A) \circ \msal{R}{q}{B}(w_B); \text{ and }  
     \msal{R}{r}{A} \bcirc \msal{R}{r}{B}(w_A, w_B) \DEF \msal{R}{r}{A}(w_A) \circ \msal{R}{r}{B}(w_B).
\end{gather*}

\begin{example} [Behavior refinement for open functions/modules of CompCert]
\label{ex:bref-open-module} 

Suppose that a source open module \(M_s\) is transformed to a target module \(M_t\).
We use semantic oracles $\chi_s$ (\(\subseteq \TEE{s}\)) and \(\chi_t\) \((\subseteq \TEE{t})\) to specify the terminating behavior of invoked functions external to the source module \(M_s\) and functions external to the target module \(M_t\) respectively. Additionally, we use a semantic oracle \(\chi_e\) (\(\subseteq \kquery_s \times \ktrace\)) to specify the aborting behavior of functions external to the source module \(M_s\).
Given a KIR \(A = \pair{W_A, \mcal{\leadsto}{A}, \msal{R}{q}{A}, \msal{R}{r}{A}}\) between languages interfaces \(\pair{\kquery_s, \kreply_s}\) and \(\pair{\kquery_t, \kreply_t}\), we define \(\chi_t \precsim_A (\chi_s, \chi_e) \) as follows.
\begin{align*}
    &\begin{aligned} %
        \chi_t \precsim_A (\chi_s, \chi_e)
        \DEF
        \forall w_A\ & q_t\ \tau\ r_t\ q_s,\;
        (q_t, \tau, r_t) \in \chi_t \Rightarrow
        (q_s, q_t) \in \msal{R}{q}{A}(w_A) \Rightarrow 
        \\
        &\big(\exists w_A'\,r_s, (q_s, \tau, r_s) \in \chi_s \wedge 
        w_A \leadsto_A w_A' \wedge
         (r_s, r_t) \in \msal{R}{r}{A}(w_A')\big) \ \vee
        \\
        &\big(\exists \tau_0, (q_s, \tau_0) \in \chi_e \wedge \tau_0 \leqslant_T \tau\big);
    \end{aligned}
\end{align*}
Similarly, given a KIR \(B = \pair{W_B, \mcal{\leadsto}{B}, \msal{R}{q}{B}, \msal{R}{r}{B}}\), 
we then define \(\sem{M_t}_{\chi_t} \precsim_B \sem{M_s}_{\chi_s}^{\chi_e}\) as follows, where \(\sem{M_t}_{\chi_t}\) has the denotation signature \FDenote (defined in \S\ref{subsec:dom_front}) with language interface \(\pair{\kquery_t, \kreply_t}\), and \(\sem{M_s}_{\chi_s}\) has the denotation signature \FDenote with language interface \(\pair{\kquery_s, \kreply_s}\).
\begin{align*} %
    &\sem{M_t}_{\chi_t} \precsim_B \sem{M_s}_{\chi_s}^{\chi_e} \DEF
        \text{let }
        \sem{M_s}_{\chi_s}^{\chi_e}.(\kerr) \DEF \fof{\chi_s}{M_s}{err} \cup (\fof{\chi_s}{M_s}{cll} \circ \chi_e)
        \text{ in }
    \\
    &\quad\qquad\begin{aligned}
        \forall w_B\ q_t\ \tau\ r_t\ q_s,\
        &(q_t, \tau, r_t) \in \fof{\chi_t}{M_t}{nrm} \Rightarrow
        (q_s, q_t) \in \msal{R}{q}{B}(w_B) \Rightarrow 
        \\
        &\big(\exists w_B'\ r_s, (q_s, \tau, r_s) \in \fof{\chi_s}{M_s}{nrm}
        \wedge w_B \leadsto_B w_B'
        \wedge (r_s, r_t) \in \msal{R}{r}{B}(w_B')\big)
        \ \vee
        \\
        &\big(\exists \tau_0, (q_s, \tau_0) \in \efof{\chi_s}{\chi_e}{M_s}{err} \wedge \tau_0 \leqslant_T \tau\big);
    \end{aligned}
    \\
    &\quad\qquad\begin{aligned}
        \forall w_B\ q_t\ \tau\ q_t'\ q_s,\
        &(q_t, \tau, q_t') \in \fof{\chi_t}{M_t}{cll} \Rightarrow
        (q_s, q_t) \in \msal{R}{q}{B}(w_B) \Rightarrow 
        \\
        &\big(\exists w_B'\ q_s', (q_s, \tau, q_s') \in \fof{\chi_s}{M_s}{cll}
        \wedge w_B \leadsto_B w_B'
        \wedge (q_s', q_t') \in \msal{R}{q}{B}(w_B')\big)
        \ \vee
        \\
        &\big(\exists \tau_0, (q_s, \tau_0) \in \efof{\chi_s}{\chi_e}{M_s}{err} \wedge \tau_0 \leqslant_T \tau\big);
    \end{aligned}
    \\
    &\quad\qquad\begin{aligned}
        \forall w_B\ q_t\ \tau\ q_s,\:\!
        (q_t, \tau) \in \fof{\chi_t}{M_t}{err} \Rightarrow
        (q_s,&\; q_t) \in \msal{R}{q}{B}(w_B) \Rightarrow 
        \\
        &\!\!\!
        \big(\exists \tau_0, (q_s, \tau_0) \in \efof{\chi_s}{\chi_e}{M_s}{err} \wedge \tau_0 \leqslant_T \tau\big);
    \end{aligned}
    \\
    &\quad\qquad\begin{aligned}
        \forall w_B\ q_t\ \tau\ q_s,\
        &(q_t, \tau) \in \fof{\chi_t}{M_t}{\ttfin} \Rightarrow
        (q_s, q_t) \in \msal{R}{q}{B}(w_B) \Rightarrow 
        \\
        & (q_s, \tau) \in \fof{\chi_s}{M_s}{\ttfin} \vee
        \big(\exists \tau_0, (q_s, \tau_0) \in \efof{\chi_s}{\chi_e}{M_s}{err} \wedge \tau_0 \leqslant_T \tau\big);
    \end{aligned}
    \\
    &\quad\qquad\begin{aligned}
        \forall w_B\ q_t\ \tau\ q_s,\
        &(q_t, \tau) \in \fof{\chi_t}{M_t}{\ttinf} \Rightarrow
        (q_s, q_t) \in \msal{R}{q}{B}(w_B) \Rightarrow 
        \\
        & (q_s, \tau) \in \fof{\chi_s}{M_s}{\ttinf} \vee
        \big(\exists \tau_0, (q_s, \tau_0) \in \efof{\chi_s}{\chi_e}{M_s}{err} \wedge \tau_0 \leqslant_T \tau\big).
    \end{aligned}
\end{align*}
Finally, for given Kripke interface relations \(A\) and \(B\), the behavior refinement between open modules \(M_s\) and \(M_t\), denoted by 
\(\sem{M_t} \precsim_{A \ttarrow B} \sem{M_s}\), is defined as follows.
\begin{align*}
    &\begin{aligned}
      \sem{M_t} \precsim_{A \ttarrow B} \sem{M_s} \DEF \forall \chi_t\,\chi_s\, \chi_e,\;
       \chi_t \precsim_A (\chi_s, \chi_e) \imply  \sem{M_t}_{\chi_t} \precsim_B \sem{M_s}_{\chi_s}^{\chi_e}.
    \end{aligned}
\end{align*}
We illustrate the above definitions (typically for terminating behaviors) in  Fig.~\ref{fig:compcerto_bref}.
This formalization enables classical rely–guarantee reasoning~\cite{DBLP:conf/popl/LiangFF12} by the accessibility relation between worlds.
Specifically, the evolution of worlds in the callee’s behavior refinement 
provides the rely condition for the caller’s behavior refinement. 
Conversely, assuming \(w_A \leadsto_A w_A'\), 
the evolution of worlds in the caller’s behavior refinement 
must respect the guarantee condition \(w_B \leadsto_B w_B'\).
When each module is proved to satisfy its guarantee condition, and these guarantee conditions entail their rely conditions, their compilation correctness becomes horizontally compositional.
\end{example}

\begin{theorem} [Horizontal compositionality]
For any source modules \(M_s,\, M_s'\), and target modules \(M_t,\, M_t'\),
if 
\(\sem{M_t} \precsim_{A \ttarrow A} \sem{M_s}\)
and
\(\sem{M_t'} \precsim_{A \ttarrow A} \sem{M_s'}\), then
\(\semlink{M_t}{M_t'} \precsim_{A \ttarrow A} \semlink{M_s}{M_s'}.\)
\end{theorem}

\begin{figure}[tbh]
    \begin{subfigure}[b]{.49\linewidth}
    \centering
    \includegraphics[width=0.98\linewidth]{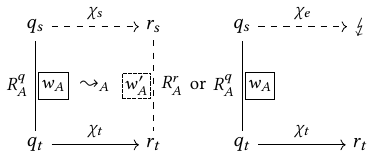}
    \caption{Refinement of callee's behavior}
    \label{subfig:bref-callee}
    \end{subfigure} 
    \begin{subfigure}[b]{.49\linewidth}
        \flushright
    \includegraphics[width=0.98\linewidth]{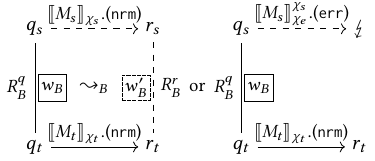}
    \centering
    \caption{Refinement of caller's behavior}
    \label{subfig:bref-caller}
    \end{subfigure}   
    \caption{Behavior refinement for CompCert open functions/modules.}
    \label{fig:compcerto_bref}
    \Description{This figure illustrates behavior refinement for open functions or modules in CompCert-style compositional verification. Subfigure (a) shows refinement of a callee’s behavior: when the target callee produces a terminating behavior under a target semantic oracle, there must be either a corresponding source callee behavior related by the reply relation, with the Kripke world advancing according to the accessibility relation, or a premature source-side error. This captures how external function behavior assumed by the target is matched by behavior available on the source side. Subfigure (b) shows refinement of a caller’s behavior: assuming related source and target call interfaces, each terminating behavior of the target caller must correspond either to a related terminating behavior of the source caller, again with suitable world evolution and related replies, or to a source-side error. Together, the two diagrams illustrate rely-guarantee reasoning for open modules, where refinement of callees provides assumptions for callers, and world evolution records the compatibility conditions needed for horizontal compositionality.}
\end{figure}

\begin{theorem} [Vertical compositionality] Given KIRs \(A_1, A_2, B_1, B_2\), and modules 
 \(M_1, M_2, M_3\),
   \begin{align*}
       \sem{M_3} \precsim_{A_2 \imply B_2} \sem{M_2}
        \imply 
        \sem{M_2} \precsim_{A_1 \imply B_1} \sem{M_1}
        \imply
        \sem{M_3} \precsim_{A_2 \bcirc A_1 \imply B_2 \bcirc B_1} \sem{M_1}.
   \end{align*}
\end{theorem}

Since the semantics of a function (e.g., Clight) is defined from the semantics of its function body (i.e., a statement), the correctness of function-level transformations can be derived from the correctness of statement-level transformations. We build the behavior refinement for open statements as follows, in a similar form to the behavior refinement for open modules.

\begin{example} [Behavior refinement for open statement of CompCert]
\label{ex:bref-open-stmt}
A statement of an open module \(M\) may call functions external to \(M\) (i.e., open). Suppose that 
a statement \(s\) in a source module \(M_s\) is compiled to a statement \(t\) in the compiled module \(M_t\). 
Given a Kripke relation  \(R_B: \mcal{K}{W_B}(\kstate_s, \kstate_t)\) between the source and target program states, and an accessibility relation \(\leadsto_B \;\subseteq\;\! W_B \times W_B\) between worlds. The behavior refinement between \(s\) and \(t\), denoted by \(\sem{t}_{\chi_t} \precsim_{R_B} \sem{s}_{\chi_s}^{\chi_e}\), is defined as follows, 
where the semantic oracles \(\chi_s, \chi_e\) and \(\chi_t\) are explained in Example~\ref{ex:bref-open-module}.
\begin{align*} %
&\sem{t}_{\chi_t} \precsim_{R_B}
\sem{s}_{\chi_s}^{\chi_e} \DEF
  \text{let }
    \efof{\chi_s}{\chi_e}{s}{err} \DEF
    \fof{\chi_s}{s}{err} \cup (\fof{\chi_s}{s}{cll} \circ \chi_{e})  
  \text{ in}
    \\
    &\quad\begin{aligned}
        \forall w_B\ \sigma_t\ \tau\ \sigma'_t\ \sigma_s,\
        &(\sigma_t, \tau, \sigma'_t) \in \fof{\chi_t}{t}{nrm} \Rightarrow
        (\sigma_s, \sigma_t) \in R_B(w_B) \Rightarrow 
        \\
        &\big(\exists w_B'\ \sigma'_s, (\sigma_s, \tau, \sigma'_s) \in \fof{\chi_s}{s}{nrm}
        \wedge w_B \leadsto_B w_B'
        \wedge (\sigma'_s, \sigma'_t) \in R_B(w_B')\big)
        \ \vee
        \\
        &\big(\exists \tau_0, (\sigma_s, \tau_0) \in \efof{\chi_s}{\chi_e}{s}{err} \wedge \tau_0 \leqslant_T \tau\big);
    \end{aligned}
    \\
    &\quad\begin{aligned}
        \forall w_B\ \sigma_t\ \tau\ \sigma_s,\:\!
        (\sigma_t, \tau) \in \fof{\chi_t}{t}{err} \Rightarrow
        (\sigma_s, \sigma_t) \in R_B(w_B) \Rightarrow 
        \big(\exists \tau_0, (\sigma_s, \tau_0) \in \efof{\chi_s}{\chi_e}{s}{err} \wedge \tau_0 \leqslant_T \tau\big);
    \end{aligned}
    \\
    &\quad\begin{aligned}
        \forall w_B\ \sigma_t\ \tau\ \sigma_s,\
        &(\sigma_t, \tau) \in \fof{\chi_t}{t}{\ttfin} \Rightarrow
        (\sigma_s, \sigma_t) \in R_B(w_B) \Rightarrow 
        \\
        & (\sigma_s, \tau) \in \fof{\chi_s}{s}{\ttfin} \vee
        \big(\exists \tau_0, (\sigma_s, \tau_0) \in \efof{\chi_s}{\chi_e}{s}{err} \wedge \tau_0 \leqslant_T \tau\big);
    \end{aligned}
    \\
    &\quad\begin{aligned}
        \forall w_B\ \sigma_t\ \tau\ \sigma_s,\
        &(\sigma_t, \tau) \in \fof{\chi_t}{t}{\ttinf} \Rightarrow
        (\sigma_s, \sigma_t) \in R_B(w_B) \Rightarrow 
        \\
        & (\sigma_s, \tau) \in \fof{\chi_s}{s}{\ttinf} \vee
        \big(\exists \tau_0, (\sigma_s, \tau_0) \in \efof{\chi_s}{\chi_e}{s}{err} \wedge \tau_0 \leqslant_T \tau\big).
    \end{aligned}
\end{align*}
Example~\ref{ex:bref_compcert} is a special case of Example~\ref{ex:bref-open-stmt} if we define $R_B \DEF \lambda w_B. R$ and $\leadsto_B \;\DEF W_B \times W_B$. In our \Coq implementation, we verify the statement transformation of the CompCert front-end with this general behavior refinement. 
We show a concrete nontrivial Kripke relation for Cminorgen in \S\ref{subsec:krel-cminorgen}.

\end{example}

\begin{example} [Behavior refinement for typical optimizations]
\label{ex:optimize}
Consider a standard constant propagation optimization on the PCALL language, where the variables known to have a constant value will be substituted with the constant.
At the optimization level, the algorithm first produces static analysis results, telling which variables hold constant values before and after every statement,
 and then generates the target code by replacing the variable with its constant. For example:
\begin{lstlisting}[language=Coq]
    source code $c$: x = 1; y = x + 1
    analysis: [$\dgcolor{\pi_0:= \_ \mapsto \bot}$] x =1; [$\dgcolor{\pi_1$:= x $\mapsto$ const 1}] y = x + 1 [$\dgcolor{\pi_2$:= x $\mapsto$ const 1; y $\mapsto$ const 2}]  
    target code $\mathcal{T}_\ttt{cp}(c)$: x = 1; y = 2
\end{lstlisting}
At the verification level, behavior refinement is based on the analysis results. We use $\Pi$ to denote the set of analysis results that map each variable to the constant value if the variable is constant, or $\bot$ if not.
We use $R_\pi$ to denote the state-matching relation between source and target program states. The relation $R_\pi$ is parameterized by the analysis result $\pi$, which requires that
the value of a variable on the source state respects its value on the analysis result $\pi$.
Assume that the source code $c$ is optimized to the target code $\mathcal{T}_\ttt{cp}(c)$ by standard constant propagation, generating analysis results $\pi_1$ at the end of $c$ with initially known result\footnote{
All the variables are mapped to the unknown value $\bot$ if we know nothing.
} $\pi_0$. We define $\sem{\mathcal{T}_\ttt{cp}(c)} \precsim_{(\pi_0, \pi_1)} \sem{c}$ as follows.
\begin{align*}
&\begin{aligned} %
    \forall \sigma_t\ \sigma'_t\ \sigma_s,\
    (\sigma_t, \sigma'_t) \in &\,\fof{}{\mathcal{T}_\ttt{cp}(c)}{nrm} \Rightarrow
    (\sigma_s, \sigma_t) \in R_{\pi_0} \Rightarrow 
    \\
    &\bigl(\exists \sigma'_s,\,\ (\sigma_s, \sigma'_s) \in \fof{}{c}{nrm}
    \wedge (\sigma'_s, \sigma'_t) \in R_{\pi_1}\bigr)
    \vee \sigma_s \in \fof{}{c}{err};
\end{aligned}
\\
&\begin{aligned} %
    \forall \sigma_t\ \sigma_s,\
    \sigma_t \in \fof{}{\mathcal{T}_\ttt{cp}(c)}{err} \Rightarrow
    (\sigma_s, \sigma_t) \in R_{\pi_0} \Rightarrow 
    \sigma_s \in \fof{}{c}{err};
\end{aligned}
\\
&\begin{aligned} %
    \forall \sigma_t\ \sigma_s,\
    \sigma_t \in \fof{}{\mathcal{T}_\ttt{cp}(c)}{dvg} \Rightarrow
    (\sigma_s, \sigma_t) \in R_{\pi_0} \Rightarrow \sigma_s \in \fof{}{c}{dvg} \ \vee\ \sigma_s \in \fof{}{c}{err};
\end{aligned}
\end{align*}

\paragbf{Summary and our solution}  
Examples~\ref{ex:bref_compcert} $\sim$~\ref{ex:optimize} have shown that nontrivial state-matching relations are required in realistic settings.
To extend the verification technique of the first two examples to such settings, we next propose the notion of Refinement Algebras (RAs). The key idea is to employ isomorphism functions to establish a structured correspondence between the denotational semantics of source and target programs, and then exploit the algebraic properties of semantic operators and RAs to achieve compositionality and to enable mechanized verification.

\end{example}

\subsection{Refinement Algebras}
\label{subsec:ref_algebra}

{ Our proof system is based on a notion of refinement algebras, whose definition is shown as follows.}

\begin{definition}
\label{def:refinement_algebra}
[Refinement algebras]
    A Refinement Algebra (RA) is a tuple $ (\mS, \mE, \mT, \gamma)$, where
    \begin{itemize}
        \item $\mS, \mE$, $\mT$ are behavior sets. Specifically, the set $\mS$ represents one kind of safe behaviors of a source program (e.g., termination), and the set $\mT$ denotes the corresponding behavior set of the target program. The set $\mE$ denotes the aborting behavior set of the source program;

        \item $\gamma: \mS \times \mE \rightarrow \mT$ is a \emph{monotonic} function transforming source behavior sets to the target:
        \begin{align}
        \forall N_1\ N_2\in \mS,\, \forall E_1\ E_2 \in \mE,\,
            N_1\subseteq N_2 \Rightarrow E_1\subseteq E_2 \Rightarrow \gamma (N_1, E_1) \subseteq \gamma(N_2, E_2)
        \end{align}
    \end{itemize}
\end{definition}

To illustrate how RAs encode various transformation correctness discussed in \S\ref{subsec:ref_examples}, we begin with the following RA instances\footnote{
For clarity of presentation, we use subscripts to distinguish gamma functions from different refinement algebra instances. In \Coq, the abstract refinement algebra (Def.~\ref{def:refinement_algebra}) is declared as a type class, and concrete refinement algebras can be instantiated for specific types. When applying an instance to define the refinement between behavior sets, Rocq’s type class resolution mechanism automatically selects the appropriate instance of gamma according to the type of its arguments.
} for Example~\ref{ex:bref_abort}, where $\pset{A}$ is the power set of a given set $A$.

\begin{example} [Termination instance 1]

\(\msal{\mA}{\kN}{1} \DEF 
(\msal{\mS}{\kN}{1},
 \msal{\mE}{\kN}{1},
 \msal{\mT}{\kN}{1},
 \gmm{\kN}{1}) \), where for a given state set \(\kstate\), 
 $\msal{\mE}{\kN}{1} \DEF \pset{\kstate}$,
 $\msal{\mS}{\kN}{1} = \msal{\mT}{\kN}{1} 
  \DEF
  \pset{\kstate \times \kstate}$.
 \(\forall N \in \msal{\mS}{\kN}{1}\),\,
 \(        E \in \msal{\mE}{\kN}{1}\),\, 
 \(
 \gmm{\kN}{1}(N, E) \DEF N \cup (E \times\kstate).
 \)
\end{example}

\begin{example} [Divergence instance 1]
\(\msal{\mA}{\kD}{1} \DEF 
(\msal{\mS}{\kD}{1},
 \msal{\mE}{\kD}{1},
 \msal{\mT}{\kD}{1},
 \gmm{\kD}{1}) \), where for a given state set \(\kstate\),
 $\msal{\mE}{\kD}{1} =
  \msal{\mS}{\kD}{1} = 
  \msal{\mT}{\kD}{1}  
  \DEF
  \pset{\kstate}$.
 \(\forall N \in \msal{\mS}{\kD}{1}\),\,
 \(        E \in \msal{\mE}{\kD}{1}\),\, 
 \(
 \gmm{\kD}{1}(N, E) \DEF N \cup E.
 \)   
\end{example}

With the RA instances \(\msal{\mA}{\kN}{1}\) and \(\msal{\mA}{\kD}{1}\), we can encode \(\sem{\compile(c)} \refines \sem{c}\) of Example~\ref{ex:bref_abort} as follows.
\begin{align}
    \sem{\mathcal{T}(c)} \sqsubseteq \sem{c} \DEF \left\{
    \begin{aligned}
    &\sem{\mathcal{T}(c)}.(\knrm) \subseteq \gmm{\kN}{1}(\sem{c}.(\knrm), \sem{c}.(\kerr)) \text{ and } \\
    &\sem{\mathcal{T}(c)}.(\kdvg) \subseteq 
    \gmm{\kD}{1}(\sem{c}.(\kdvg), \sem{c}.(\kerr)) \text{ and } \\
    &\sem{\mathcal{T}(c)}.(\kerr) \subseteq 
    \gmm{\kD}{1}(\oset, \sem{c}.(\kerr))
    \end{aligned}
    \right.
\end{align}

\begin{example} [Termination instance 2] 
\(\msal{\mA}{\kN}{2} \DEF 
(\msal{\mS}{\kN}{2},
 \msal{\mE}{\kN}{2},
 \msal{\mT}{\kN}{2},
 \gmm{\kN}{2}) \), where for a given state set \(\kstate\), 
 $\msal{\mE}{\kN}{2} \DEF
  \pset{\kstate \times \ktrace}$, 
 $\msal{\mS}{\kN}{2} =
  \msal{\mT}{\kN}{2} \DEF
  \pset{\nrmrel{\kstate}}$.
 \(\forall N \in \msal{\mS}{\kN}{2}\),\,
 \(        E \in \msal{\mE}{\kN}{2}\),\, 
\begin{align*}
&\begin{aligned} %
  \gmm{\kN}{2}(N, E) \DEF
    \Set{(\sigma, \tau, \sigma') |
    \big((\sigma, \tau, \sigma') \in N \big) \vee
    (\exists \tau_0, (\sigma, \tau_0) \in E \wedge \tau_0 \leqslant_T \tau)}
\end{aligned}
\end{align*}
\end{example}

\begin{example} [Finite diverging instance 2]
\(\msal{\mA}{\kF}{2} \DEF 
(\msal{\mS}{\kF}{2},
 \msal{\mE}{\kF}{2},
 \msal{\mT}{\kF}{2},
 \gmm{\kF}{2}) \), where for a given state set \(\kstate\), 
 $\msal{\mE}{\kF}{2} =
  \msal{\mS}{\kF}{2} =
  \msal{\mT}{\kF}{2}\DEF
  \pset{\kstate \times \ktrace}$.
 \(\forall N \in \msal{\mS}{\kF}{2}\),\,
 \(        E \in \msal{\mE}{\kF}{2}\),\, 
\begin{align*}
    &\begin{aligned} %
  \gmm{\kF}{2}(N, E) \DEF
    \Set{(\sigma, \tau) |
    (\sigma, \tau) \in N \vee
    \left(\exists \tau_0, (\sigma, \tau_0) \in E \wedge \tau_0 \leqslant_T \tau\right)}
\end{aligned}
\end{align*}

\end{example}

\begin{example} [Infinite diverging instance 2]
\(\msal{\mA}{\kI}{2} \DEF 
(\msal{\mS}{\kI}{2},
 \msal{\mE}{\kI}{2},
 \msal{\mT}{\kI}{2},
 \gmm{\kI}{2}) \), where for a given state set \(\kstate\), 
 $\msal{\mE}{\kI}{2} =
  \msal{\mS}{\kI}{2} =
  \msal{\mT}{\kI}{2}\DEF
  \pset{\kstate \times \kitrace}$.
 \(\forall N \in \msal{\mS}{\kI}{2}\),\,
 \(        E \in \msal{\mE}{\kI}{2}\),\, 
\begin{align*}
    &\begin{aligned} %
  \gmm{\kI}{2}(N, E) \DEF
    \Set{(\sigma, \tau) |
    (\sigma, \tau) \in N \vee
    \left(\exists \tau_0, (\sigma, \tau_0) \in E \wedge \tau_0 \leqslant_T \tau\right)}
\end{aligned}
\end{align*}
\end{example}

With the RA instances \(\msal{\mA}{\kN}{2}\), \(\msal{\mA}{\kF}{2}\) and \(\msal{\mA}{\kI}{2}\), we can encode \(\sem{t} \precsim \sem{s}\) of Example~\ref{ex:bref_Cshmgen} as follows.
\begin{align}
  \sem{t} \precsim \sem{s} \DEF \left\{
  \begin{aligned}
    &\fof{}{t}{nrm} \subseteq
      \gmm{\kN}{2}(\fof{}{s}{nrm}, 
         {\fof{}{s}{err}})
    \text{ and }
    \fof{}{t}{err} \subseteq \gmm{\kF}{2} (\oset, \fof{}{s}{err}) \text{ and }
    \\
    &\fof{}{t}{\ttfin} \subseteq \gmm{\kF}{2} (\fof{}{s}{\ttfin}, \fof{}{s}{err})
    \text{ and }
    \\
    &\fof{}{t}{\ttinf} \subseteq \gmm{\kI}{2} (\fof{}{s}{\ttinf}, 
    \fof{}{s}{err})
  \end{aligned}
  \right.
\end{align}

{In fact, we can provide more general RA instances to encode nontrivial behavior refinement.}
For instance, recall that we define the preservation of terminating behaviors in Example~\ref{ex:bref_compcert} as follows.
\begin{align}
\begin{aligned} %
    \sem{t}.(\knrm) \precsim_R\; & (\sem{s}.(\knrm), \sem{s}.(\kerr))
    \DEF
    \forall \sigma_t\ \tau\ \sigma'_t\ \sigma_s,\
    (\sigma_t, \tau, \sigma'_t) \in \fof{}{t}{nrm} \Rightarrow
    (\sigma_s, \sigma_t) \in R \Rightarrow 
    \\
    &\big(\exists \sigma'_s, (\sigma_s, \tau, \sigma'_s) \in \fof{}{s}{nrm}
    \wedge (\sigma'_s, \sigma'_t) \in R\big) \vee
    \left(\exists \tau_0, (\sigma_s, \tau_0) \in \fof{}{s}{err} \wedge \tau_0 \leqslant_T \tau\right);
\end{aligned}
\end{align}
We aim to define a termination instance (as given by Example~\ref{ins:compcert-terminate}) of refinement algebras such that
\[\sem{t}.(\knrm) \subseteq \gamma_R(\sem{s}.(\knrm), \sem{s}.(\kerr)) \iff \sem{t}.(\knrm) \precsim_R (\sem{s}.(\knrm), \sem{s}.(\kerr))\]  

\begin{example} [Termination instance 3]
\label{ins:compcert-terminate}
\(\msal{\mS}{\kN}{3} \DEF \pset{\kstate_s \times \ktrace \times \kstate_s} \), \(\msal{\mE}{\kN}{3} \DEF \pset{\kstate_s \times \ktrace}\), and \(\msal{\mT}{\kN}{3} \DEF \pset{\kstate_t \times \ktrace \times \kstate_t}\). For a given relation \(R \subseteq \kstate_s \times \kstate_t\), \(\forall N \in \msal{\mS}{\kN}{3},\, \forall E \in \msal{\mE}{\kN}{3},\)
\begin{align}
    &  \gmm{\kN}{3}(N, E) \DEF \Set{(\sigma_t, \tau, \sigma'_t) | 
     \forall \sigma_s,
     (\sigma_s, \sigma_t) \in R \Rightarrow\, 
    \begin{aligned}
     &(\exists \sigma'_s, (\sigma_s, \tau, \sigma'_s) \in N
     \wedge (\sigma'_s, \sigma'_t) \in R)
     \\
     &  \vee 
    \left(\exists \tau_0, (\sigma_s, \tau_0) \in E \wedge \tau_0 \leqslant_T \tau\right)   
    \end{aligned}
     }
\end{align}

\end{example}

Similarly, we can define RA instances for finite divergence and infinite divergence as follows. 
\begin{example} [Finite divergence instance 3] \(\msal{\mS}{\kF}{3} \DEF \pset{\kstate_s \times \ktrace } \), \(\msal{\mE}{\kF}{3} \DEF \pset{\kstate_s \times \ktrace}\), and \(\msal{\mT}{\kF}{3} \DEF \pset{\kstate_t \times \ktrace}\). For a given relation \(R \subseteq \kstate_s \times \kstate_t\), \(\forall N \in \msal{\mS}{\kF}{3},\, \forall E \in \msal{\mE}{\kF}{3},\)
\begin{align}
    &  \gmm{\kF}{3}(N, E) \DEF \Set{(\sigma_t, \tau) | 
     \forall \sigma_s,
     (\sigma_s, \sigma_t) \in R \Rightarrow\, 
     (\sigma_s, \tau) \in N \vee 
     (\exists \tau_0, (\sigma_s, \tau_0) \in E \wedge \tau_0 \leqslant_T \tau) }
\end{align}
\end{example}
\begin{example} [Infinite divergence instance 3] \(\msal{\mS}{\kI}{3} \DEF \pset{\kstate_s \times \kitrace } \), \(\msal{\mE}{\kI}{3} \DEF \pset{\kstate_s \times \ktrace}\), and \(\msal{\mT}{\kI}{3} \DEF \pset{\kstate_t \times \kitrace}\). For a given relation \(R \subseteq \kstate_s \times \kstate_t\), \(\forall N \in \msal{\mS}{\kI}{3},\, \forall E \in \msal{\mE}{\kI}{3},\)
\begin{align}
    &  \gmm{\kI}{3}(N, E) \DEF \Set{(\sigma_t, \tau) | 
     \forall \sigma_s,
     (\sigma_s, \sigma_t) \in R \Rightarrow\, 
     (\sigma_s, \tau) \in N \vee 
     (\exists \tau_0, (\sigma_s, \tau_0) \in E \wedge \tau_0 \leqslant_T \tau) }
\end{align}
\end{example}

With the RA instances \(\msal{\mA}{\kN}{3}=
(\msal{\mS}{\kN}{3}, \msal{\mE}{\kN}{3}, \msal{\mT}{\kN}{3}, \gmm{\kN}{3})\), \(\msal{\mA}{\kF}{3}=
(\msal{\mS}{\kF}{3}, \msal{\mE}{\kF}{3}, \msal{\mT}{\kF}{3}, \gmm{\kF}{3})\), and \(\msal{\mA}{\kI}{3}=
(\msal{\mS}{\kI}{3}, \msal{\mE}{\kI}{3}, \msal{\mT}{\kI}{3}, \gmm{\kI}{3})\), we can encode \(\sem{t} \precsim_R \sem{s}\) of Example~\ref{ex:bref_compcert} as follows, and obviously we have \(\sem{t} \refines_R \sem{s} \iff \sem{t} \precsim_R \sem{s}\).
\begin{align}
    \sem{t} \refines_R \sem{s} \DEF \left\{
    \begin{aligned}
        &\sem{t}.(\knrm) \subseteq
        \gmm{\kN}{3}(\sem{s}.(\knrm), \sem{s}.(\kerr)) \text{ and }
        \sem{t}.(\kerr) \subseteq
        \gmm{\kF}{3}(\oset, \sem{s}.(\kerr)) \text{ and }
        \\
        &\sem{t}.(\kfin) \subseteq
        \gmm{\kF}{3}(\sem{s}.(\kfin), \sem{s}.(\kerr)) \text{ and }
        \\
        &\sem{t}.(\kinf) \subseteq
        \gmm{\kI}{3}(\sem{s}.(\kinf), \sem{s}.(\kerr))
    \end{aligned}
    \right.
\end{align}

From these RA instances, we observe that the RA-based formulation, via the $\gamma$ function, provides an algebraic characterization of the correspondence between source and target behaviors. This formulation not only offers a uniform and compositional foundation for proving compiler correctness, but also enables the extension of algebraic verification methods from simple cases to more realistic compilation settings. Before turning to more complex RA instances, we first proceed to introduce enhanced RA structures to further consolidate the algebraic underpinnings of our framework.

\subsection{Enhanced Refinement Algebras}
\label{subsec:ra_enriched}
Since the union and composition operators are ubiquitous in the semantic definition of composite statements, we define the following enhanced RAs to facilitate structured compiler verification.

\begin{definition} [Union-preserved RAs]
    A refinement algebra \(\mathcal{A}=(\mS, \mE, \mT, \gamma)\) is \emph{union-preserved} iff:
    \begin{align}
        \forall N_1\,N_2\in \mS,\; \forall E_1\, E_2 \in \mE,\,
        \gamma (N_1, E_1) \cup \gamma (N_2, E_2)
        \subseteq \gamma (N_1 \cup N_2, E_1 \cup E_2) \label{prop:gunion}
    \end{align}
\end{definition}

\begin{definition} [Concatenation-preserved RAs]
    Refinement algebras \(\mcal{A}{1}=(\mS_1, \mE_1, \mT_1, \gamma_1)\) and \(\mcal{A}{2}=(\mS_2, \mE_2, \mT_2, \gamma_2)\) are \emph{concatenation-preserved} to refinement algebra \(\mcal{A}{3} = (\mS_3, \mE_1, \mT_3, \gamma_3) \) iff:
    given operators
    \(\circ: \mT_1 \times \mT_2 \arrow \mT_3\),
    \(\circ: \mS_1 \times \mS_2 \arrow \mS_3\), and
    \(\circ: \mS_1 \times \mE_2 \arrow \mE_1\) s.t. for any
    \( N_1,\, N_2,\, E_1, E_2\),
    \begin{align}
    \gamma_1(N_1, E_1) \circ
    \gamma_2(N_2, E_2) \subseteq
    \gamma_3\big(N_1 \circ N_2, E_1 \cup( N_1 \circ E_2)\big) \label{prop:gconcat}
    \end{align}
    Particularly, if RAs \(\mathcal{A}\) and \(\mathcal{A}\) are concatenation-preserved to \(\mathcal{A}\), \(\mathcal{A}\) is \emph{self-concatenation-preserved}. 
    {Note that for a given RA \((\mS, \mE, \mT, \gamma)\), writing \( \gamma(N, E)\) implicitly means that \(N \in \mS \) and \(E \in \mE \) . }
\end{definition}

\begin{lemma} [Instances of enhanced refinement algebras]
\label{lemma:simple-instance-prop}
     \( \msal{\mA}{\kN}{i}\), \(\msal{\mA}{\kF}{i}\), and \( \msal{\mA}{\kI}{i} \) (defined in \S\ref{subsec:ref_algebra}) are union-preserved refinement algebras;
    (ii)  \( \msal{\mA}{\kN}{i} \) is self-concatenation-preserved;
    (iii) \(\msal{\mA}{\kN}{i}\) and \(\msal{\mA}{\kF}{i}\) are
        concatenation-preserved to \(\msal{\mA}{\kF}{i}\); and
    (iv) \(\msal{\mA}{\kN}{i}\) and \(\msal{\mA}{\kI}{i}\) are
        concatenation-preserved to \(\msal{\mA}{\kI}{i}\) for \(i=1,2,3\).
    {Concrete operators for these RA instances are provided in \S\ref{subsec:while_lang} according to the type of behavior sets.}
\end{lemma}
\begin{figure}[t]
    \centering
    \includegraphics[width=0.8\linewidth]{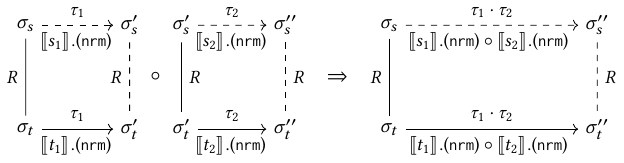}
    \caption{The behavior refinement for CompCert is compositional w.r.t. composition.}
    \label{fig:refine_concat}
    \Description{This figure schematically illustrates that CompCert-style behavior refinement is preserved by sequential composition for normal termination. The figure visually summarizes the concatenation-preservation property used to prove that refinement of sequential statements follows from refinement of their components.}
\end{figure}

We first provide a brief preview of how the enhanced RA instances from Lemma~\ref{lemma:simple-instance-prop} enable a uniform refinement proof (see \S\ref{sec:compose_correctness} for details). For a structural transformation of sequential statements (as in Example~\ref{ex:bref_while}), the proof proceeds as follows: the crucial step \((\ref{pr:ind_hypo}) \subseteq (\ref{pr:concat_incl})\) follows immediately from the concatenation-preservation property of the underlying RA instances. Fig.~\ref{fig:refine_concat} provides a schematic illustration of this property for normally terminating behaviors in CompCert.
\begin{align}
  &\fof{}{t_1;t_2}{nrm} =
   \fof{}{t_1}{nrm} \circ
  \fof{}{t_2}{nrm} \label{pr:def_tgt}
  \\
  &\begin{aligned} \subseteq
  \gmm{\kN}{i}(\fof{}{s_1}{nrm},\ \fof{}{s_1}{err}) \circ
  \gmm{\kN}{i}(\fof{}{s_2}{nrm},\ \fof{}{s_2}{err})
   \end{aligned} \label{pr:ind_hypo}
  \\
  &\begin{aligned} \subseteq
     \gmm{\kN}{i}(\fof{}{s_1}{nrm} \circ \fof{}{s_2}{nrm}, \
           \fof{}{s_1}{err} \cup
           (\fof{}{s_1}{nrm} \circ \fof{}{s_2}{err}))
    \end{aligned} \label{pr:concat_incl}
   \\
    &\begin{aligned} =
    \gmm{\kN}{i}(\fof{}{s_1;s_2}{nrm},
           \fof{}{s_1;s_2}{err})
  \end{aligned} \label{pr:src_def}
\end{align}

Since the star $*$ and infinity $\infty$ operators are ubiquitous in the semantic definition of loops and recursion, we define the following enhanced RAs to facilitate structured compiler verification.
\begin{definition} [Star-preserved RAs]  RAs \(\mcal{A}{1}=(\mS_1, \mE, \mT_1, \gamma_1)\) and \(\mcal{A}{2}=(\mS_2, \mE, \mT_2, \gamma_2)\) are \emph{star-preserved} to \(\mcal{A}{2} \) iff: given operators 
    \(\circ: \mT_1 \times \mT_2 \arrow \mT_2\),
    \(\circ: \mS_1 \times \mS_2 \arrow \mS_2\), and
    \(\circ: \mS_1 \times \mE \arrow \mE\),
    and star operators
    \(*: \mT_1 \arrow \mT_1\) and \(*: \mS_1 \arrow \mS_1\) such that for any
    \( N_1 \in \mS_1\),
    \( N_2 \in \mS_2\),
    \( E_1\) and \(E_2 \in \mE\),
    \begin{align}
        (\gamma_1(N_1, E_1))^* \circ \gamma_2(N_2, E_2)
        \subseteq \gamma_2 \big(N_1^* \circ N_2, N_1^* \circ (E_1 \cup E_2)
        \big) \label{prop:gstar}
    \end{align}
\end{definition}

\begin{definition} [Infinity-preserved RAs]
RAs \(\mcal{A}{1}=(\mS_1, \mE, \mT_1, \gamma_1)\) and \(\mcal{A}{2}=(\mS_2, \mE, \mT_2, \gamma_2)\) are \emph{infinity-preserved} to \(\mcal{A}{2} \) iff: given composition operators
    \(\circ: \mT_1 \times \mT_2 \arrow \mT_2\),
    \(\circ: \mS_1 \times \mS_2 \arrow \mS_2\), and
    \(\circ: \mS_1 \times \mE \arrow \mE\),
    star operators
    \(*: \mT_1 \arrow \mT_1\),
    \(*: \mS_1 \arrow \mS_1\),
    and infinity operators
    \(\infty: \mT_1 \arrow \mT_2\), 
    \(\infty: \mS_1 \arrow \mS_2\) 
    such that for any
    \( N_1 \in \mS_1\),
    \( N_2 \in \mS_2\), and
    \( E_1, E_2 \in \mE\),
    \begin{align}
        \big((\gamma_1(N_1, E_1))^* \circ \gamma_2(N_2, E_2)\big)
        \cup
        \bigl(\gamma_1(N_1, E_1)\bigr)^\infty
        \subseteq \gamma_2 \big((N_1^* \circ N_2) \cup N_1^\infty,
        N_1^* \circ (E_1 \cup E_2)\big)
    \label{prop:ginfty}
    \end{align}
\end{definition}
\paragbf{Connection to fixed points}
Given sets $X$ and $Y$, the form $X^*\circ Y$ is actually the Kleene (least) fixed point of function $f(x) = Y \cup (X \circ x)$, and $(X^* \circ Y) \cup X^\infty$ is the greatest fixed point of function $f$. They are frequently used in the semantic definition of loops and recursion (as shown in \S\ref{sec:semantics_nocall} and \S\ref{sec:semantics_pcall}), so we use enhanced RAs to facilitate the compilation correctness proof of loops and recursion.

\begin{lemma}[Instances of star/infinity-preserved refinement algebras]
\label{lemma:iteration-instance}
    (i) \(\msal{\mA}{\kN}{i}\) and \(\msal{\mA}{\kF}{i}\) (defined in \S\ref{subsec:ref_algebra}) are
        star-preserved to \(\msal{\mA}{\kF}{i}\);
     and
    (ii) \(\msal{\mA}{\kN}{i}\) and \(\msal{\mA}{\kI}{i}\) are
        infinity-preserved to \(\msal{\mA}{\kI}{i}\) for \(i=1,2,3\).
\end{lemma}

\subsection{Algebraic Behavior Refinement}
\label{subsec:ra_instances}
\subsubsection{RA Instances for Open Statement Transformation of CompCert}

We define the following RA instances to encode the behavior refinement for CompCert phases shown in Example~\ref{ex:bref-open-stmt}.
\begin{theorem}
\label{thm:compcert-stmt-RA}
Given a Kripke relation \(R_B: \mcal{K}{W}(\kstate_s, \kstate_t)\) and a preorder \(\leadsto_B \;\subseteq\;\! W \times W\).
    \( \msal{\mA}{\kN}{}=
      (\msal{\mS}{\kN}{},
       \mcal{\mE}{},
       \msal{\mT}{\kN}{},
       \gmm{\kN}{R_B}) \),
    \( \msal{\mA}{\kF}{}=
     (\msal{\mS}{\kF}{},
      \mcal{\mE}{},
      \msal{\mT}{\kF}{},
      \gmm{\kF}{R_B}) \), and
    \( \msal{\mA}{\kI}{}=
      (\msal{\mS}{\kI}{},
       \mcal{\mE}{},
       \msal{\mT}{\kI}{},
       \gmm{\kI}{R_B}) \), where
  \(\mcal{E}{} \DEF \pset{\kstate_s \times \ktrace}\),
\begin{align*}
   &\msal{\mS}{\kN}{} \DEF \pset{\kstate_s \times \ktrace \times \kstate_s}
   && \msal{\mS}{\kF}{} \DEF \pset{\kstate_s \times \ktrace}
   && \msal{\mS}{\kI}{} \DEF \pset{\kstate_s \times \kitrace}
   \\
   &\msal{\mT}{\kN}{} \DEF \pset{\kstate_t \times \ktrace \times \kstate_t}
   && \msal{\mT}{\kF}{} \DEF \pset{\kstate_t \times \ktrace}
   && \msal{\mT}{\kI}{} \DEF \pset{\kstate_t \times \kitrace}
\end{align*}
\begin{align*}
   &\begin{aligned}
    \gmm{\kN}{R_B}(N, E) \DEF \Set{(\sigma_t, \tau, \sigma'_t) | 
    \begin{aligned}
    \forall \sigma_s\, w_B,\, (\sigma_s, \sigma_t) & \in R_B(w_B) \Rightarrow 
        \big(\exists w_B'\ \sigma'_s, (\sigma_s, \tau, \sigma'_s) \in N
        \wedge w_B \leadsto_B w_B' \ \wedge
        \\
         (\sigma'_s, \sigma'_t) & \in R_B(w_B')\big)
        \vee
        \big(\exists \tau_0, (\sigma_s, \tau_0) \in E \wedge \tau_0 \leqslant_T \tau\big);
      \end{aligned} }
   \end{aligned}
   \\
   &\begin{aligned}
    \gmm{\kF}{R_B}(N, E) \DEF
    \Set{(\sigma_t, \tau) \mid 
      \forall \sigma_s\, w_B,\; (\sigma_s, \sigma_t) \in R_B(w_B) \Rightarrow
        (\sigma_s, \tau) \in N \vee
        (\exists \tau_0,\; (\sigma_s, \tau_0) \in E \wedge \tau_0 \leqslant_T \tau)
    }
   \end{aligned}
   \\
   &\begin{aligned}
    \gmm{\kI}{R_B}(N, E) \DEF
    \Set{(\sigma_t, \tau) \mid 
      \forall \sigma_s\, w_B,\; (\sigma_s, \sigma_t) \in R_B(w_B) \Rightarrow
        (\sigma_s, \tau) \in N \vee
        (\exists \tau_0,\; (\sigma_s, \tau_0) \in E \wedge \tau_0 \leqslant_T \tau) }
   \end{aligned}
\end{align*}
(i) \( \msal{\mA}{\kN}{} \),
    \( \msal{\mA}{\kF}{}\), and
    \( \msal{\mA}{\kI}{}\) are union-preserved refinement algebras; 
(ii)  \( \msal{\mA}{\kN}{} \) is self-concatenation-preserved;
(iii) \(\msal{\mA}{\kN}{}\) and \(\msal{\mA}{\kF}{}\) are
        concatenation-preserved to \(\msal{\mA}{\kF}{}\);
(iv) \(\msal{\mA}{\kN}{}\) and \(\msal{\mA}{\kI}{}\) are
        concatenation-preserved to \(\msal{\mA}{\kI}{}\);
(v) \(\msal{\mA}{\kN}{}\) and \(\msal{\mA}{\kF}{}\) are
    star-preserved to \(\msal{\mA}{\kF}{}\);
 and
(vi) \(\msal{\mA}{\kN}{}\) and \(\msal{\mA}{\kI}{}\) are
    infinity-preserved to \(\msal{\mA}{\kI}{}\).
\end{theorem}

\subsubsection{RA Instances for CompCert Open Modules}

\begin{theorem}
\label{thm:open-module-RA}
Given a KIR \(B = \pair{W_B, \mcal{\leadsto}{B}, \msal{R}{q}{B}, \msal{R}{r}{B}}\), 
    \( \msal{\mA}{\kN}{}=
      (\msal{\mS}{\kN}{},
       \mcal{\mE}{},
       \msal{\mT}{\kN}{},
       \gmm{\kN}{B}) \),
    \( \msal{\mA}{\kC}{}=
      (\msal{\mS}{\kC}{},
       \mcal{\mE}{},
       \msal{\mT}{\kC}{},
       \gmm{\kC}{B}) \),
    \( \msal{\mA}{\kI}{}=
      (\msal{\mS}{\kI}{},
       \mcal{\mE}{},
       \msal{\mT}{\kI}{},
       \gmm{\kI}{B}) \), and
    \( \msal{\mA}{\kF}{}=
     (\msal{\mS}{\kF}{},
      \mcal{\mE}{},
      \msal{\mT}{\kF}{},
      \gmm{\kF}{B}) \), where
\begin{align*}
&\begin{aligned}
   &\mcal{E}{} \DEF \pset{\kquery_s \times \ktrace}
   && \msal{\mS}{\kF}{} \DEF \pset{\kquery_s \times \ktrace}
   && \msal{\mT}{\kF}{} \DEF \pset{\kquery_t \times \ktrace}
   \\
   &\msal{\mS}{\kN}{} \DEF \pset{\kquery_s \times \ktrace \times \kreply_s}
   && \msal{\mS}{\kC}{} \DEF \pset{\kquery_s \times \ktrace \times \kquery_s}
   && \msal{\mS}{\kI}{} \DEF \pset{\kquery_s \times \kitrace}
   \\
   &\msal{\mT}{\kN}{} \DEF \pset{\kquery_t \times \ktrace \times \kreply_t}
   && \msal{\mT}{\kC}{} \DEF \pset{\kquery_t \times \ktrace \times \kquery_t}
   && \msal{\mT}{\kI}{} \DEF \pset{\kquery_t \times \kitrace}    
\end{aligned}
\end{align*}
\begin{align*}
    &\begin{aligned}
    \gmm{\kN}{B}(N, E) \DEF \Set{(q_t, \tau, r_t) | 
    \begin{aligned}
    \forall q_s\, w_B,\, (q_s, q_t) & \in \msal{R}{q}{B}(w_B) \Rightarrow 
        \big(\exists w_B'\ r_s, (q_s, \tau, r_s) \in N
        \wedge w_B \leadsto_B w_B' \ \wedge
        \\
         (r_s, r_t) & \in \msal{R}{r}{B}(w_B')\big)
        \vee
        \big(\exists \tau_0, (q_s, \tau_0) \in E \wedge \tau_0 \leqslant_T \tau\big);
      \end{aligned} }
    \end{aligned}
    \\
    &\begin{aligned}
    \gmm{\kC}{B}(N, E) \DEF \Set{\;\!\!(q_t, \tau, q_t')\:\!\! | 
    \begin{aligned}
    \forall q_s\, w_B,\, (q_s, q_t) & \in \msal{R}{q}{B}(w_B) \Rightarrow 
        \big(\exists w_B'\ q_s', (q_s, \tau, q_s') \in N
        \wedge w_B \leadsto_B w_B' \ \wedge
        \\
         (q_s', q_t') & \in \msal{R}{q}{B}(w_B')\big)
        \vee
        \big(\exists \tau_0, (q_s, \tau_0) \in E \wedge \tau_0 \leqslant_T \tau\big);
      \end{aligned} }
    \end{aligned}
    \\
    &\begin{aligned}
    \gmm{\kF}{B}(N, E) \DEF
    \Set{(q_t, \tau) \mid 
      \forall q_s\, w_B,\; (q_s, q_t) \in \msal{R}{q}{B}(w_B) \Rightarrow
        (q_s, \tau) \in N \vee
        (\exists \tau_0,\; (q_s, \tau_0) \in E \wedge \tau_0 \leqslant_T \tau)
    }
    \end{aligned}
    \\
    &\begin{aligned}
    \gmm{\kI}{B}(N, E) \DEF
    \Set{(q_t, \tau) \mid 
      \forall q_s\, w_B,\; (q_s, q_t) \in \msal{R}{q}{B}(w_B) \Rightarrow
        (q_s, \tau) \in N \vee
        (\exists \tau_0,\; (q_s, \tau_0) \in E \wedge \tau_0 \leqslant_T \tau) }
    \end{aligned}
\end{align*}
(i) \( \msal{\mA}{\kN}{} \),
    \( \msal{\mA}{\kC}{} \),
    \( \msal{\mA}{\kF}{}\), and
    \( \msal{\mA}{\kI}{}\) are union-preserved refinement algebras; 
(ii)  \( \msal{\mA}{\kC}{} \) is self-concatenation-preserved;
(iii) \(\msal{\mA}{\kC}{}\) and \(\msal{\mA}{\kN}{}\) are
        concatenation-preserved to \(\msal{\mA}{\kN}{}\);
(iv) \(\msal{\mA}{\kC}{}\) and \(\msal{\mA}{\kF}{}\) are
        concatenation-preserved to \(\msal{\mA}{\kF}{}\);
(v) \(\msal{\mA}{\kC}{}\) and \(\msal{\mA}{\kI}{}\) are
        concatenation-preserved to \(\msal{\mA}{\kI}{}\);
(vi) \(\msal{\mA}{\kC}{}\) and \(\msal{\mA}{\kN}{}\) are
    star-preserved to \(\msal{\mA}{\kN}{}\);
(vii) \(\msal{\mA}{\kC}{}\) and \(\msal{\mA}{\kF}{}\) are
    star-preserved to \(\msal{\mA}{\kF}{}\);
 and
(viii) \(\msal{\mA}{\kC}{}\) and \(\msal{\mA}{\kI}{}\) are
    infinity-preserved to \(\msal{\mA}{\kI}{}\).
\end{theorem}

\begin{definition} [Behavior refinement for open statements of Cshmgen] \label{def:Cshmgen_stmt_ref}
  Given natural numbers\footnote{
    Natural numbers $n_b$ and $n_c$ are arguments of the statement transformation of Cshmgen,
    which are used to count the exiting number from nested blocks (i.e., \textbf{exit} $n_b$ and \textbf{exit} $n_c$ are compiled from \textbf{break} and \textbf{continue} respectively). 
  } $n_b$, $n_c$ and the aborting behavior of external C functions $\chi_e$ $(\subseteq$ \kcquery $\times$ \ktrace). A denotation 
  $D_2$ of \CshmDenote is a refinement of a denotation $D_1$ of \ClitDenote, denoted by $D_2 \sqsubseteq_{(R_B, n_b, n_c, \chi_e)} D_1$, iff:
\begin{align*}
\begin{alignedat}{3}
  & \denof{D_2}{nrm} &\;\subseteq\;& 
     \gmm{\kN}{R_B}\big(\denof{D_1}{nrm}, &\; & 
     \denof{D_1}{err} \cup (\denof{D_1}{cll} \circ \chi_e)\big) \\
  & \denof{D_2}{err} &\;\subseteq\;& 
     \gmm{\kF}{R_B}\big(\oset, &\; & 
     \denof{D_1}{err} \cup (\denof{D_1}{cll} \circ \chi_e)\big) \\
  & \denof{D_2}{blk}_{n_b} &\;\subseteq\;& 
     \gmm{\kN}{R_B}\big(\denof{D_1}{brk}, &\; & 
     \denof{D_1}{err} \cup (\denof{D_1}{cll} \circ \chi_e)\big) \\
  & \denof{D_2}{blk}_{n_c} &\;\subseteq\;& 
     \gmm{\kN}{R_B}\big(\denof{D_1}{ctn}, &\; & 
     \denof{D_1}{err} \cup (\denof{D_1}{cll} \circ \chi_e)\big) \\
  & \denof{D_2}{fin\_dvg} &\;\subseteq\;& 
     \gmm{\kF}{R_B}\big(\denof{D_1}{fin\_dvg}, &\; & 
     \denof{D_1}{err} \cup (\denof{D_1}{cll} \circ \chi_e)\big) \\
  & \denof{D_2}{inf\_dvg} &\;\subseteq\;& 
     \gmm{\kI}{R_B}\big(\denof{D_1}{inf\_dvg}, &\; & 
     \denof{D_1}{err} \cup (\denof{D_1}{cll} \circ \chi_e)\big)
\end{alignedat}
\end{align*}
where the gamma instances are defined in Theorem~\ref{thm:compcert-stmt-RA}, the state-matching relation \(R_B \DEF \lambda w_B.\idrel\), and \(\leadsto_B \DEF W_B \times W_B\) as the representation of program states remains unchanged during Cshmgen.

\end{definition}

\begin{definition} [Refinement between open function/module denotations]
\label{def:ref_function}
  
  Let $\Phi_1$ and $\Phi_2$ be the semantics of open  functions/modules (e.g., that of Clight).
 The semantics $\Phi_2$ is said to be a refinement of $\Phi_1$, written as $\Phi_2 \refines_{A \ttarrow B} \Phi_1$, if and only if:
 (i) $\denof{D_2}{dom} = \denof{D_1}{dom}$,
 and (ii) for any callee behaviors $\chi_t$, $\chi_s$ and $\chi_e$ such that $\chi_t \subseteq \gmm{\kN}{A}(\chi_s, \chi_e)$, the following conditions hold, where $D_1  = (\Phi_1)_{\chi_s}$ of \(\FDenote_s\), $D_2 = (\Phi_2)_{\chi_t}$ of \(\FDenote_t\), and the gamma instances are defined in Theorem~\ref{thm:open-module-RA}:
\begin{align*}
\begin{alignedat}{3}
  & \denof{D_2}{nrm} &\;\subseteq\;& 
     \gmm{\kN}{B}\big(\denof{D_1}{nrm}, &\; & 
     \denof{D_1}{err} \cup (\denof{D_1}{cll} \circ \chi_e)\big) \\
  & \denof{D_2}{err} &\;\subseteq\;& 
     \gmm{\kF}{B}\big(\oset, &\; & 
     \denof{D_1}{err} \cup (\denof{D_1}{cll} \circ \chi_e)\big) \\
  & \denof{D_2}{cll} &\;\subseteq\;& 
     \gmm{\kC}{B}\big(\denof{D_1}{cll}, &\; & 
     \denof{D_1}{err} \cup (\denof{D_1}{cll} \circ \chi_e)\big) \\
  & \denof{D_2}{fin\_dvg} &\;\subseteq\;& 
     \gmm{\kF}{B}\big(\denof{D_1}{fin\_dvg}, &\; & 
     \denof{D_1}{err} \cup (\denof{D_1}{cll} \circ \chi_e)\big) \\
  & \denof{D_2}{inf\_dvg} &\;\subseteq\;& 
     \gmm{\kI}{B}\big(\denof{D_1}{inf\_dvg}, &\; & 
     \denof{D_1}{err} \cup (\denof{D_1}{cll} \circ \chi_e)\big)
\end{alignedat}
\end{align*}
\end{definition}

The following theorems illustrates the soundness of the above
algebraic behavior refinement in terms of \(\precsim_{R_B}\) defined in Example~\ref{ex:bref-open-stmt} and \(\precsim_{A \ttarrow B}\) defined in Example~\ref{ex:bref-open-module}.
\begin{theorem} [Soundness of algebraic behavior refinement 1] 
\label{thm:cshm-ref-sound}
For any \(D_1\) of\, \ClitDenote and \(D_2\) of\, \CshmDenote, for any nature numbers \(n_b, n_c\), and for any semantic oracle \(\chi_e\),
    \begin{align*}
       \textit{if } D_2 \refines_{(R_B, n_b, n_c, \chi_e)} D_1, \textit{ then } D_2 \precsim_{R_B} D_1.
    \end{align*}
\end{theorem} 

\begin{theorem} [Soundness of algebraic behavior refinement 2]
\label{thm:module-ref-sound}
Given KIRs \(A\) and \(B\). For any semantics of open functions/modules $\Phi_1$ and $\Phi_2$,
    \(\Phi_2 \refines_{A \ttarrow B} \Phi_1\) if and only if \(\Phi_2 \precsim_{A \ttarrow B} \Phi_1 \).
\end{theorem}

\subsubsection{RA Instances for Typical Optimizations} 
We define the following RA instances to encode the behavior refinement for standard constant propagation optimizations shown in Example~\ref{ex:optimize}.

\begin{theorem} 
\label{thm:optimize-RA}
Let \(\pi_0\), \(\pi_1\) \(\in \Pi\) be analysis results. Given a relation \(R: \Pi  \arrow \pset{\kstate_s \times \kstate_t}\), let
    \( \msal{\mA}{\pi_1}{\pi_0}=
      (\msal{\mS}{\kN}{},
       \mcal{\mE}{},
       \msal{\mT}{\kN}{},
       {\gmm{\pi_1}{ \pi_0}}) \),
    \( \msal{\mA}{\kD}{\pi_0}=
     (\msal{\mS}{\kD}{},
      \mcal{\mE}{},
      \msal{\mT}{\kD}{},
      \gmm{\kD}{\pi_0}) \), where
\begin{align*}
   \msal{\mS}{\kN}{} \DEF \pset{\kstate_s \times \kstate_s}
   \quad
   \msal{\mT}{\kN}{} \DEF \pset{\kstate_t \times \kstate_t}
   \quad 
   \mcal{E}{} = \msal{\mS}{\kD}{} \DEF \pset{\kstate_s}
   \quad
   \msal{\mT}{\kD}{} \DEF \pset{\kstate_t}
\end{align*}
\begin{align*}
    &\begin{aligned}
        \gmm{\pi_1}{\pi_0}(N, E) \DEF \Set{(\sigma_t, \sigma'_t) | 
    \forall \sigma_s,
            (\sigma_s, \sigma_t) \in R(\pi_0) \Rightarrow
            (\exists \sigma'_s, (\sigma_s, \sigma'_s) \in N
            \wedge (\sigma'_s, \sigma'_t) \in R(\pi_1)) 
            \vee \sigma_s \in E
        }
    \end{aligned}
    \\
    &\begin{aligned}
    \gmm{\kD}{\pi_0}(N, E) \DEF
    \Set{\sigma_t \mid
      \forall \sigma_s,\; (\sigma_s,\sigma_t)\in R(\pi_0) \Rightarrow
        (\sigma_s \in N \;\vee\;
          \sigma_s\in E)
    }
    \end{aligned}
\end{align*}
(i) For any \(\pi_0\) and \(\pi_1\), \( \msal{\mA}{\pi_1}{\pi_0} \) and
    \( \msal{\mA}{\kD}{\pi_0}\) are union-preserved refinement algebras;
(ii) for any \(\pi_0, \pi_1, \pi_2\), \(\msal{\mA}{\pi_1}{\pi_0}\) and \(\msal{\mA}{\pi_2}{\pi_1}\) are
        concatenation-preserved to \(\msal{\mA}{\pi_2}{\pi_0}\);
(iii) for any  \(\pi_0, \pi_1\),
 \(\msal{\mA}{\pi_1}{\pi_0}\) and \(\msal{\mA}{\kD}{\pi_1}\) are
        concatenation-preserved to \(\msal{\mA}{\kD}{\pi_0}\);
and (iv) for any \(\pi\),
 \(\msal{\mA}{\pi}{\pi}\) and \(\msal{\mA}{\kD}{\pi}\) are
    star-preserved and infinity-preserved to \(\msal{\mA}{\kD}{\pi}\).
\end{theorem}

Similarly, given analysis results \(\pi_0\) and \(\pi_1\), we can define \(\sem{\constprop(c)} \refines_{(\pi_0, \pi_1)} \sem{c}\) for typical optimizations in Example~\ref{ex:optimize} as follows, such that if \(\sem{\constprop(c)} \refines_{(\pi_0, \pi_1)} \sem{c}\), then \(\sem{\constprop(c)} \precsim_{(\pi_0, \pi_1)} \sem{c}\).  
\begin{align}
  \sem{\constprop(c)} \refines_{(\pi_0, \pi_1)} \sem{c} \DEF \left\{ \,
  \begin{aligned}
    & \fof{}{\constprop(c)}{nrm} \subseteq \gmm{\pi_1}{\pi_0}(\fof{}{c}{nrm}, \fof{}{c}{err}),
    \\ 
    &\fof{}{\constprop(c)}{err} \subseteq \gmm{\kD}{\pi_0}(\oset, \fof{}{c}{err}),  \text{ and }
    \\
    &\fof{}{\constprop(c)}{dvg} \subseteq 
    \gmm{\kD}{\pi_0}(\fof{}{c}{dvg},\fof{}{c}{err})
  \end{aligned} \right.
  \label{eq:bref_abort_opt}
\end{align}

\subsection{Kripke Relations for Cminorgen}
\label{subsec:krel-cminorgen}

CompCert employs a hierarchy of memory relations—ranging from simple identity to extension and, in the most intricate cases, memory injection. Each phase is thus verified with the weakest relation sufficient for capturing its memory effects, which keeps proofs as lightweight as possible.

Among phases of CompCert, Cminorgen is one of the most challenging phase: it eliminates  scalar local variables that are never addressed with the \& operator from activation records and coalesces separate memory blocks allocated to addressed local variables within a function activation into a single block representing the activation record as a whole, thus reshaping memory layout in a nontrivial way. In the original CompCert proof, this is justified using the memory injection relation. In our framework, we recast this verification by exposing memory injections in a Kripke relation. 
We show key definitions as follows, where the memory model of CompCert is explained in \S\ref{subsec:dom_front}.
\paragbf{Memory injection of CompCert}
CompCert introduces \emph{injection functions} to track the general correspondence between memory blocks. An injection function $f: \kmeminj$ ($\kmeminj$ $\triangleq \block \partfun \block \times \offset$) is a partial function mapping source memory blocks to target locations, in the following form:
  \[
     f(b_1) = \bot \quad \text{ or } \quad f(b_1) = (b_2, \delta).
  \]
 The former means that $b_1$ is unmapped, and the later means that a location in $b_1$ is mapped to a target location of block $b_2$ with $\delta$ offset increased. Given an injection function \(f\),
 we use $\vinj{f}{v_1}{v_2}$ to assert the correspondent of program values (naturally lifted to \(\vinj{f}{\Vec{v_1}}{\Vec{v_2}}\) for a list of values), and use $\minj{f}{m_1}{m_2}$ to represent \emph{memory injections}. Specifically,
a pointer value \(\ptr{(b_1, o_1)}\) is injected to \(\ptr{(b_2,o_2)}\) if there is a \(\delta \in \offset\) s.t. \(f(b_1) = {(b_2, \delta)}\) and \(o_2 = o_1 + \delta\); integers or floating numbers are injected if the source value is identical to the target value; and
    the special value \kundef can be injected to any values.
Memory injection \(\minj{f}{m_1}{m_2}\) holds if memory contents in \(m_1\) are preserved to \(m_2\) for all injected addresses (simplified for presentation purposes), i.e., 
\begin{align*}
\text{for any } b_1\app b_2\app o_1\app o_2,\;
    &\begin{aligned}
      \vinj{f}{\ptr(b_1,o_1)}{\ptr(b_2, o_2)} \imply
      \vinj{f}{m_1(b_1,o_1)}{m_2(b_2, o_2)}.        
    \end{aligned}
\end{align*}

We reuse the notion of Kripke memory relation for protection called \kinjp~\cite{DBLP:conf/pldi/KoenigS21,DBLP:journals/pacmpl/ZhangWWKS24} to formulate rely-guarantee reasoning~\cite{DBLP:conf/popl/LiangFF12} in our framework.
    \begin{definition} [The Kripke memory relation \kinjp]
    \label{def:injp}
        $\kwd{injp} \triangleq \pair{W_\sinjp, \leadsto_\sinjp, R_\sinjp}$, where  \(W_\sinjp \triangleq \kmeminj \times \kmem \times \kmem\), \(R_\sinjp:\mathcal{K}_{W_\iii{injp}}(\kmem, \kmem)\) is a Kripke relation between memory states such that
    \[
       \forall f\, m_1\, m_2,\;
       (m_1, m_2) \in R_\sinjp(f, m_1, m_2) \iff \minj{f}{m_1}{m_2}
    \] 
and \(\leadsto_\sinjp (\subseteq W_\sinjp \times W_\sinjp)\) is an accessibility relation s.t. for any worlds  \((f, m_1, m_2), (f', m_1', m_2')\),
\begin{align*}
    (f, m_1,  m_2) \leadsto_\sinjp (f', m_1', m_2') & \iff
    \begin{aligned}
        \left\{
        \begin{aligned}
        (1) &\ \kwd{unmapped}(f) \subseteq \kwd{unchanged\_on}(m_1, m_1'); \\
        (2) &\ \kwd{out\_of\_reach}(f, m_1) \subseteq \kwd{unchanged\_on}(m_2, m_2'); \\
        (3) &\ \kwd{separately\_increased}(f, f', m_1, m_2); \text{ and} \\
        (4) &\ \kwd{mem\_forward}(m_1, m_1') \land \kwd{mem\_forward}(m_2, m_2').
        \end{aligned}
        \right.                
    \end{aligned}
\end{align*}
Here \(\kwd{unchanged\_on}(m, m')\) denotes the set of memory locations where the memory contents remain unchanged between \(m\) and \(m'\).
Clause (1) asserts that values at locations unmapped by \( f \) remain unchanged. 
Clause (2) indicates that values at locations not mapped to by \( f \) from valid blocks of \( m_1 \) are also unchanged. 
Clause (3) states that \( f' \) extends \( f \) by adding mappings for new blocks while preserving the existing mappings, and new blocks are not permitted to map into the existing memory blocks of \( m_2 \). 
Clause (4) represents side conditions that are not essential to the core idea.

Intuitively, the relation \(\leadsto_\sinjp\) governs the evolution of worlds across internal executions and external calls while enforcing disciplined memory protection. Concretely, private and public memory locations are identified by an injection \(f\): public locations are those related by \(f\), whereas private locations are either unmapped in the source or out-of-reach in the target (having no preimage through \(f\)). During world evolution, memory cells at private locations must remain unchanged, ensuring that the environment cannot observe or modify them, while public locations are preserved up to the injection so that loads, stores, and pointer comparisons behave consistently between the source and target. These constraints enable \(\kinjp\) to account for complex memory transformations, such as allocating new stack slots for spilled register variables, reorganizing stack frames, and extending or coalescing memory blocks across compilation phases, while preserving behavioral consistency, as demonstrated in CompCertO~\cite{DBLP:conf/pldi/KoenigS21,DBLP:journals/pacmpl/ZhangWWKS24}. 
\end{definition}

The Kripke memory relation 
\kinjp can be naturally integrated into our framework as follows.

\begin{definition} [The Kripke relation for Statement Transformation of Cminorgen]
Given the set of program states \(\kstate_s, \kstate_t\) (respectively for Csharpminor and Cminor),  where \(\kstate_s = \kstate_t \DEF \ktmp \times \kmem\).
We define the Kripke relation \(R_B: \mcal{K}{W_\iii{injp}}(\kstate_s, \kstate_t) \) for Cminorgen as follows, for any $(f, m_1, m_2) \in W_\sinjp$, \((te_1, m_1) \in \kstate_s\), and \((te_2, m_2) \in \kstate_t\),  
\begin{align*}
    &\bigl((te_1, m_1), (te_2, m_2)\bigr) \in R_B(f, m_1, m_2) \iff 
            \minj{f}{m_1}{m_2} \land  \forall x,\;
       \vinj{f}{te_1(x)}{te_2(x)}
\end{align*}
\end{definition}

\begin{definition} [The Kripke interface relation for Cminorgen]
\label{def:cmin_kir}
    \(A_\sss{C} \DEF \pair{W_\sinjp, \leadsto_\sinjp, \msal{R}{q}{\kC}, \msal{R}{r}{\kC}}\), where
    for any C queries \(q_1 = (id_1, \Vec{v_1}, m_1)\), \(q_2 = (id_2, \Vec{v_2}, m_2)\),
    and C replies \(r_1 = ({v_1'}, m_1')\), \(r_2 = ({v_2'}, m_2')\),
    \begin{align*}
       \bigl((id_1, \Vec{v_1}, m_1), (id_2, \Vec{v_2}, m_2)\bigr) \in \msal{R}{q}{\kC}(f, m_1, m_2)
       &\iff id_1 = id_2 \land
         \vinj{f}{\Vec{v_1}}{\Vec{v_2}} \land
         \minj{f}{m_1}{m_2}
         \\
       \bigl(({v_1'}, m_1'), ({v_2'}, m_2')\bigr) \in \msal{R}{r}{\kC}(f', m_1', m_2') 
       &\iff \vinj{f'}{v_1'}{v_2'} \land
         \minj{f'}{m_1'}{m_2'}
    \end{align*}
\end{definition}

Using the above Kripke relations (with \(A = B \DEF A_\kC\) in Def.~\ref{def:ref_function}), we have successfully reverified Cminorgen---one of the most intricate memory-reorganizing phases in CompCert---within our framework, which indicates that our proof techniques are scalable to the remainder of CompCert.

\subsection{Design Choices}
\label{subsec:design_choice}

In this subsection, we explain why \(\gamma\) is designed to take two arguments rather than just one.
For illustration, recall that Example~\ref{ex:bref_compcert} defines the preservation of terminating behaviors as follows.
\begin{align}
\label{eq:compcert_nrm3}
\begin{aligned} %
    \sem{t}.(\knrm) {\precsim_R}\; & (\sem{s}.(\knrm), \sem{s}.(\kerr))
    \DEF
    \forall \sigma_t\ \tau\ \sigma'_t\ \sigma_s,\
    (\sigma_t, \tau, \sigma'_t) \in \fof{}{t}{nrm} \Rightarrow
    (\sigma_s, \sigma_t) \in R \Rightarrow 
    \\
    &\big(\exists \sigma'_s, (\sigma_s, \tau, \sigma'_s) \in \fof{}{s}{nrm}
    \wedge (\sigma'_s, \sigma'_t) \in R\big) \vee
    \left(\exists \tau_0, (\sigma_s, \tau_0) \in \fof{}{s}{err} \wedge \tau_0 \leqslant_T \tau\right);
\end{aligned}
\end{align}
where 
the termination behavior of a target program may correspond either to the terminating behavior of the source program or to its aborting behavior, with different requirements: starting from related initial states, the former requires the final states to be related, whereas the latter does not. 
{
We next show that designing \(\gamma\)-functions with a single parameter would be less flexible to meet the two  different requirements than that with two parameters by contraposition.
}
{Specifically, suppose we intended to find instances of \(\gamma: \mS \arrow \mT\) to encode Formula~(\ref{eq:compcert_nrm3}) for algebraic reasoning.} There may exist two choices, and they must satisfy the soundness conditions as  follows.
\begin{align}
  & \text{Chioce 1:  }\;
    N_t \subseteq \gmm{\kN}{R}(N_s) \cup \bigl(\gmm{\kE}{R}(E_s) \times \kstate_t\bigr);
    \qquad \text{Choice 2:  } \;
    N_t \subseteq \gmm{\kN}{R}\bigl(N_s \cup (\wE_s \times \kstate_s)\bigr)
  \\
  &\text{Soundness of chioce 1:} \qquad
    N_t \subseteq \gmm{\kN}{R}(N_s) \cup \bigl(\gmm{\kE}{R}(E_s) \times \kstate_t\bigr) 
    \;\imply\;
    N_t \precsim_R (N_s, E_s)
  \\
  &\text{Soundness of choice 2:} \qquad
    N_t \subseteq \gmm{\kN}{R}\bigl(N_s \cup (\wE_s \times \kstate_s)\bigr) 
   \;\imply\;
   N_t \precsim_R (N_s, E_s)
  \\
  &\begin{aligned}
    \text{where }
    &\wE_s \DEF
    \Set{(\sigma, \tau) \in \kstate_s \times \ktrace \mid
    \exists \tau_0\; \tau_1,\;(\sigma, \tau_0) \in E_s \land
    \tau = \tau_0 \cdot \tau_1 }
    \\
    &\gmm{\kN}{R} : \pset{\nrmrel{\kstate_s}} \arrow \pset{\nrmrel{\kstate_t}} 
    \\
    &\gmm{\kE}{R} : \pset{\kstate_s \times \ktrace} \arrow \pset{\kstate_t \times \ktrace} 
  \end{aligned}
\end{align}
where in Choice 2 we employ a trace-extension operator \(\wE_s\) for the source aborting behavior \(E_s\) such that the observable event trace of \(E_s\) can be a prefix of the target terminating behavior \(N_t\).
In fact, we can define the following instances of \(\gamma\) that satisfy the above soundness conditions.
\begin{align}
    &\gmm{\kN}{R}(N) \DEF \Set{(\sigma_t, \tau, \sigma'_t) | 
     \forall \sigma_s,
     (\sigma_s, \sigma_t) \in R \Rightarrow
     (\exists \sigma'_s, (\sigma_s, \tau, \sigma'_s) \in N
     \wedge (\sigma'_s, \sigma'_t) \in R)
     }
\\
    &\gmm{\kE}{R}(E) \DEF \Set{(\sigma_t, \tau) | 
     \forall \sigma_s,
     (\sigma_s, \sigma_t) \in R \Rightarrow
      \exists \tau_0, (\sigma_s, \tau_0) \in E \wedge \tau_0 \leqslant_T \tau
     }
\end{align}
\begin{figure}[t]
    \begin{subfigure}[b]{.48\linewidth}
    \centering 
     \includegraphics[width=0.82\linewidth]{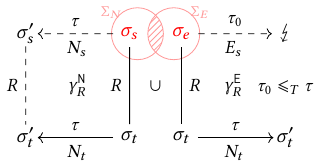}
    \caption{Behavior refinement of choice 1}
    \label{subfig:gamma-choice1}
    \end{subfigure} 
    \begin{subfigure}[b]{.48\linewidth}
        \flushright
    \includegraphics[width=0.82\linewidth]{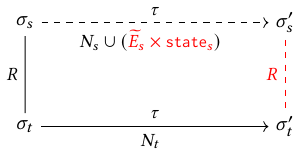}
    \centering
    \caption{Behavior refinement of choice 2}
    \label{subfig:gamma-choice2}
    \end{subfigure}   
\caption{Choice 1 separates correspondences for normal termination and abortion using distinct 
$\gamma$-instances, but they fail to cover full cases: \( \gmm{\kN}{R}\) fails for (\(\Sigma_E-\Sigma_N\)) and \( \gmm{\kE}{R}\) fails for (\(\Sigma_N-\Sigma_E\)). Choice 2 forces uniform correspondence between the final states, even when the target terminates normally but the source aborts.}
    \label{fig:design-choices}
    \Description{This figure compares two unsuccessful design choices for defining behavior refinement between target normal termination and source normal termination or abortion. Subfigure (a), Choice 1, uses two separate \(\gamma\)-instances: one for source normal termination and one for source abortion. The diagram highlights a mismatch between the set \(\Sigma_N\) of source initial states that may terminate and the set \(\Sigma_E\) of source initial states that may abort. States in \(\Sigma_E \setminus \Sigma_N\) are not properly handled by the normal-termination instance, and states in \(\Sigma_N \setminus \Sigma_E\) are not properly handled by the abortion instance, so the two cases cannot be cleanly combined. Subfigure (b), Choice 2, forces normal termination and abortion into a single combined relation. This gives uniform correspondence with the target’s normal termination, but it also requires assigning an artificial final source state even when the source aborts and has no real final state. The figure illustrates why both choices are inadequate for realistic compiler verification.}
\end{figure}

\emph{Choice~1} employs distinct instances of \(\gamma\) to separately capture the two correspondences between source and target behaviors. The instance \(\gmm{\kN}{R}\) characterizes cases where, from any related initial state, the source program may terminate normally with a matching final state. The instance \(\gmm{\kE}{R}\) characterizes cases where, from any related initial state, the source program may abort. 
However, there may exist related source states from which the source program does not terminate or does not abort. More specifically, we define \(\Sigma_N \DEF \Set{\sigma_s|\exists \tau\,\sigma_s',\, (\sigma_s, \tau, \sigma_s') \in N_s}\), and \(\Sigma_E \DEF \Set{\sigma_e|\exists \tau,\, (\sigma_e, \tau) \in E_s}\), which respectively denotes the set of initial states for termination and abortion.
As illustrated in Fig.~\ref{fig:design-choices},
those states in \(\Sigma_N\) but not in \(\Sigma_E\) are not properly quantified by \(\gmm{\kE}{R}\), and those states in \(\Sigma_E\) but not in \(\Sigma_N\) are not properly quantified by \(\gmm{\kN}{R}\).
Simply taking the union of \(\gmm{\kN}{R}\) and \(\gmm{\kE}{R}\) does not solve this problem.
\emph{Choice~2} avoids this issue by taking the union of \(N_s\) and \(\widetilde{E_s} \times \kstate_s\) first.
This enforces a uniform treatment of normal termination and abortion, aligning both with the target’s normal termination. 
As a result, whenever the source program aborts, one must provide an artificial final source state related to the actual target final state, even though aborted executions do not yield canonical final states. 
{That is, for any target state there must exist a related source state, which may not in fact exist.} 
This additional requirement is unnatural and often infeasible under general state-matching relations. 
Therefore, both choices are ill-suited for realistic compiler verification.
{By contraposition, we conclude that it would be more challenging to instantiate a single-parameterized \(\gamma\) for algebraic compiler verification, especially when the source semantics is nondeterministic.}
\section{Compositional  Correctness}
\label{sec:compose_correctness}

Refinement algebras further enhance the compositionality of compiler verification. The semantic operators in extended KATs enable the semantic composition of program substructures, while refinement algebras support the compositionality of behavioral refinement. Their strong algebraic properties allow us to design automated proof tactics to mechanize compiler verification.
In this section, taking the transformation of Clight statements in Cshmgen as an example, we first show the proof style of our framework in \S\ref{subsec:proof_style}, along with two essential automated tactics. We then present the correctness of compilation phases in \S\ref{subsec:compile_correct} and the module-level compositionality in \S\ref{subsec:module_compose}.

\subsection{Proof Style}
\label{subsec:proof_style}
\begin{figure}[ht]
\def\ksuc{\textbf{Suc.}}
\begin{gather*}
  \Infer[Seq]{%
    \tr{s_1} = \ksome(u_1) \;\land\; \tr{s_2} = \ksome(u_2)
  }{%
    \tr{s_1;s_2} = \ksome(u_1;u_2)
  }
  \\
  \Infer[Brk]{%
  }{%
    \tr{\textbf{break}} = \ksuc(\textbf{exit } n_b)
  }
  \quad
  \Infer[Ctn]{%
  }{%
    \tr{\textbf{continue}} = \ksuc(\textbf{exit } n_c)
  }
  \\
  \Infer[Loop]{%
    \trs{1}{0}{s_1} = \ksome(u_1) \;\land\;
    \trs{0}{n_c+1}{s_2} = \ksome(u_2)
  }{%
    \tr{\textbf{loop } s_1\ s_2}
    = \ksome\bigl(\textbf{block} \textcolor{blue}{\{} \textbf{Sloop } (\textbf{block } \{u_1\}; u_2) \textcolor{blue}{\}}\bigr)
  }
\end{gather*}
\caption{Original Cshmgen rules from CompCert (selectively listed), where \ksuc\xspace is short for \ksome.}
\label{fig:compilation_rules}
\Description{This figure lists selected statement-translation rules from the original Cshmgen phase of CompCert, written as inference rules. The Seq rule states that if statements \(s_1\) and \(s_2\) translate successfully to \(u_1\) and \(u_2\), then the sequential statement \(s_1;s_2\) translates to the sequential composition \(u_1;u_2\). The Brk rule translates \textbf{break} to a successful result containing \textbf{exit} \(n_b\), and the Ctn rule translates \textbf{continue} to a successful result containing \textbf{exit} \(n_c\). The Loop rule translates a Clight loop with body statements \(s_1\) and \(s_2\) into a Csharpminor term built from nested \textbf{block} constructs and an \textbf{Sloop}; the translated first body \(u_1\) is placed inside an inner block, followed by the translated second body \(u_2\). Overall, the figure illustrates how Cshmgen compiles structured control flow into lower-level block, loop, and exit constructs.}
\end{figure}
The proof style in our framework is demonstrated by the following proof sketch of Lemma \ref{lemma:clitgen}, where $\mathcal{T}_\ttt{Cshmgens}$ denotes the statement translation of Cshmgen, with selected compilation rules shown in Fig. \ref{fig:compilation_rules}.
We next omit the \(R_B\), defined in Def.~\ref{def:Cshmgen_stmt_ref}, for the behavior refinement \(\refines_{(R_B, n_b, n_c, \chi_e)}\).

\begin{lemma} \label{lemma:clitgen}
For any natural numbers
$n_b$, $n_c$, and callee function behaviors $\chi_t$, $\chi_s$, $\chi_e$, if $n_b \neq n_c$ and $\chi_t \subseteq \gmm{\kN}{R_B}(\chi_s, \chi_e)$, then for any Clight statement $s$ and Csharpminor statement $u$,
    \begin{align*}
    \text{if}\ \ \mathcal{T}_\ttt{Cshmgens}(n_b, n_c, s) = \ksome (u) , \text{ then } \bracket{u}_{\chi_t} \sqsubseteq_{(n_b, n_c, \chi_e)} \bracket{s}_{\chi_s}.
    \end{align*}
\end{lemma}
\begin{proof} [Proof sketch]
    By induction on the statement $s$, we derive proof obligations asserting that for every syntactic construct of $s$, if each substructure of $s$ satisfies the refinement relation with its compiled counterpart, then the entire statement $s$ does as well. To illustrate how the refinement algebra works in the proof, we demonstrate the core ideas behind three key proof obligations.
     \paragbf{Case 1: $s$ is a sequential statement}
    By induction, since $\mathcal{T}_\ttt{Cshmgens}(n_b, n_c, s_1; s_2) = \ksome(u)$, then
    $\mathcal{T}_\ttt{Cshmgens}(n_b, n_c, s_1) = \ksome(u_1)$,
    $\mathcal{T}_\ttt{Cshmgens}(n_b, n_c, s_2) = \ksome(u_2)$, and $u = u_1;u_2$. Prove that:
\begin{gather}
    \text{if }\bracket{u_1}_{\chi_t} \sqsubseteq_{(n_b, n_c, \chi_e)} \bracket{s_1}_{\chi_s}\! \text{ and }
    \bracket{u_2}_{\chi_t} \sqsubseteq_{(n_b, n_c, \chi_e)} \bracket{s_2}_{\chi_s}  
    \text{, then }
    \bracket{u_1;u_2}_{\chi_t} \sqsubseteq_{(n_b, n_c, \chi_e)} \bracket{s_1;s_2}_{\chi_s}.
    \label{eq:Cshmgen-nrm}
\end{gather}
For the sake of brevity, let $N_i$, $E_i$ and $C_i$ respectively represent
$\fof{\chi_s}{s_i}{nrm}$, $\fof{\chi_s}{s_i}{err}$ and $\fof{\chi_s}{s_i}{cll}$ for $i=$ 1, 2. The proofs for the refinement of terminating behaviors are shown as follows.
\begin{align}
  &\fof{\chi_t}{u_1;u_2}{nrm} =
   \fof{\chi_t}{u_1}{nrm} \circ
  \fof{\chi_t}{u_2}{nrm} \label{pf:def_tgt}
  \\
  &\begin{aligned} \subseteq
  \mathcolor{red}
  {\gmm{\kN}{R_B}(N_1,\ E_1 \cup C_1 \circ \chi_e)} \circ
      \mathcolor{red}
      {\gmm{\kN}{R_B}(N_2,\ E_2 \cup C_2 \circ \chi_e)}
   \end{aligned} \label{pf:ind_hypo}
  \\
  &\begin{aligned} \subseteq
     \gmm{\kN}{R_B}(\mathcolor{red}{N_1 \circ N_2}, \
         \mathcolor{red}
         {(E_1 \cup  C_1 \circ \chi_e) \cup}
        \mathcolor{red}
         {N_1 \circ \left(E_2 \cup C_2 \circ \chi_e
         \right)}
    )
   \end{aligned} \label{pf:concat_incl}
   \\
  &\begin{aligned} \subseteq
     \gmm{\kN}{R_B}(N_1 \circ N_2, \
        \mathcolor{red}
        {(E_1 \cup N_1 \circ E_2) \cup}
         \mathcolor{red}
         {\left(C_1 \cup N_1 \circ C_2
         \right) \circ \chi_e}
    )
    \end{aligned}  \label{pf:cup_circ_comm}
    \\
    &\begin{aligned} =
    \gmm{\kN}{R_B}(\fof{\chi_s}{s_1;s_2}{nrm},
       \fof{\chi_s}{s_1;s_2}{err} \cup
       \fof{\chi_s}{s_1;s_2}{cll} \circ \chi_e
       )
  \end{aligned} \label{pf:src_def}
\end{align}
The key to the above proofs is the inclusion (\ref{pf:ind_hypo}) $\subseteq$ (\ref{pf:concat_incl}), which is directly solved by the concatenation-preserved property of RAs.
The inclusion (\ref{pf:concat_incl}) $\subseteq$ (\ref{pf:cup_circ_comm}) is by the axioms of Kleene Algebras;
the inclusion (\ref{pf:def_tgt}) $\subseteq$ (\ref{pf:ind_hypo}) is by hypotheses;
and the equivalence in (\ref{pf:def_tgt}) and between (\ref{pf:cup_circ_comm}) $\sim$ (\ref{pf:src_def}) holds by definition.
We see that refinement algebras reduce refinement proofs into the proof of relational propositions in Kleene Algebras.
Other cases for sequential statements can be proved similarly.
\paragbf{Case 2: $s$ is a loop statement} By structural induction, the proof goal for the terminating case for example is to show:
 if $\bracket{u_1}_{\chi_t} \sqsubseteq_{(1, 0, \chi_e)} \bracket{s_1}_{\chi_s}\! \text{ and }
    \bracket{u_2}_{\chi_t} \sqsubseteq_{(0, n_c + 1, \chi_e)} \bracket{s_2}_{\chi_s}
$, then
\begin{align}
    &\begin{aligned} \label{eq:loop_refine}
        &\fof{\chi_t}{\textbf{block} \{ \textbf{Sloop } (\textbf{block } \{u_1\}; u_2) \}}{nrm} \ \subseteq  \\ &\quad \quad\gmm{\kN}{R_B}\big(\fof{\chi_s}{\textbf{loop } s_1\ s_2}{nrm},
            \fof{\chi_s}{\textbf{loop } s_1\ s_2}{err} \cup (\fof{\chi_s}{\textbf{loop } s_1\ s_2}{cll} \circ \chi_e)\big)
    \end{aligned}
\end{align}
Since the semantics of loops is defined with extended KATs operators $^*$ (for the reflexive transitive closure) and $^\infty$ (for the infinite iterations of composition), the behavior refinement between them can be efficiently verified using the instances of star-preserved RAs and infinity-preserved RAs.
\newcommand{\gms}[2]{\gmm{\kN}{R_B}(\fof{\chi_s}{s_#1}{#2}, E_#1)}
\newcommand{\Gms}[3]{\mathcolor{#3}{\gmm{\kN}{R_B}(\fof{\chi_s}{s_#1}{#2}, E_#1)}}

In our \Coq implementation, proof obligations for composition statements  (e.g., case 1 and case 2) is mechanically
solved by two automated new proof tactics: {\tts gamma\_simpl} and {\tts solvefix}. The {\tts gamma\_simpl} tactic automatically reduces the behavior refinement between source and target programs to relational propositions of extended KATs, by properties of refinement algebras. The {\tts solvefix} tactic mechanically verifies these relational propositions by properties of extended KATs. More specifically,
by unfolding the definition of the left-hand side of Formula~(\ref{eq:loop_refine}), the goal is to show:
\begin{align*}
     &\begin{aligned}
         & \big((\mathcolor{ForestGreen}{\fof{\chi_t}{u_1}{nrm}}
            \cup \mathcolor{OliveGreen}{\fof{\chi_t}{u_1}{blk}_0})
            \circ \mathcolor{Orange}{\fof{\chi_t}{u_2}{nrm}}\big)^*
            \ \circ \\
         &\quad\big(\mathcolor{Plum}{\fof{\chi_t}{u_1}{blk}_1} \cup
         ((\mathcolor{ForestGreen}{\fof{\chi_t}{u_1}{nrm}} \cup 
         \mathcolor{OliveGreen}{\fof{\chi_t}{u_1}{blk}_0}) \circ
         \mathcolor{Cyan}{\fof{\chi_t}{u_2}{blk}_0})\big)
     \end{aligned}
     \\
     &\begin{aligned}
        \quad \quad \quad \ \ \subseteq &\ \gmm{\kN}{R_B}(\fof{\chi_s}{\textbf{loop } s_1\ s_2}{nrm},
            \fof{\chi_s}{\textbf{loop } s_1\ s_2}{err} \cup (\fof{\chi_s}{\textbf{loop } s_1\ s_2}{cll} \circ \chi_e))
     \end{aligned}
\end{align*}
We then use the induction hypothesis to rewrite every set of the left-hand side of this goal, and remember 
    $\fof{\chi_s}{s_i}{err} \cup (\fof{\chi_s}{s_i}{cll}) \circ \chi_e$) as ${E_i}$ for brevity. The proof goal is reduced to: 
\begin{align}
    &\begin{aligned}
       &\big((\Gms{1}{nrm}{ForestGreen} \cup
        \Gms{1}{ctn}{OliveGreen}) \circ
        \Gms{2}{nrm}{Orange}\big)^* \ \circ \\
         &\quad\big( \Gms{1}{brk}{Plum} \cup
         (\Gms{1}{nrm}{ForestGreen} \cup
         \Gms{1}{ctn}{OliveGreen}) \circ
         \Gms{2}{brk}{Cyan}\big)    \nonumber
    \end{aligned} 
     \\
     &\begin{aligned}
        \quad \quad \quad 
        \quad \quad
        \subseteq &\ \gmm{\kN}{R_B}(\fof{\chi_s}{\textbf{loop } s_1\ s_2}{nrm},
            \fof{\chi_s}{\textbf{loop } s_1\ s_2}{err} \cup (\fof{\chi_s}{\textbf{loop } s_1\ s_2}{cll} \circ \chi_e))
     \end{aligned} \label{eq:gamma_equation}
\end{align}
By applying properties of underlying RA instances (using \tra), we reduce the proof goal to:
\begin{align}
    &\begin{aligned}
        &\big((\fof{\chi_s}{s_1}{nrm} \cup
        \fof{\chi_s}{s_1}{ctn}) \circ
        \fof{\chi_s}{s_2}{nrm}\big)^*
        \ \circ \\
        &\quad \quad \big(\fof{\chi_s}{s_1}{brk}\cup (\mathcolor{dkgreen}{(} \fof{\chi_s}{s_1}{nrm} \cup \fof{\chi_s}{s_1}{ctn}\mathcolor{dkgreen}{)} \circ
        \fof{\chi_s}{s_2}{brk})\big)
      \subseteq
      \fof{\chi_s}{\textbf{loop } s_1\ s_2}{nrm};
      \label{eq:loop_nrm_ka}
    \end{aligned}
    \\
    &\begin{aligned}
        &\big((\fof{\chi_s}{s_1}{nrm} \cup \fof{\chi_s}{s_1}{ctn}) \circ \fof{\chi_s}{s_2}{nrm}\big)^* \ \circ
        \\ 
        &\quad \quad \big(E_1 \cup
            (E_1 \cup E_1) \cup (\mathcolor{dkgreen}{(} \fof{\chi_s}{s_1}{nrm} \cup \fof{\chi_s}{s_1}{ctn}\mathcolor{dkgreen}{)} \circ E_2) \\
        &\quad \quad \quad \quad \cup 
            (E_1 \cup E_1) \cup (\mathcolor{dkgreen}{(} \fof{\chi_s}{s_1}{nrm} \cup \fof{\chi_s}{s_1}{ctn}\mathcolor{dkgreen}{)} \circ E_2)\big) \\
        &\quad \quad \quad \quad \quad \quad \quad \quad \quad \quad
         \quad \quad \quad \ \ 
        \subseteq\ \fof{\chi_s}{\textbf{loop } s_1\ s_2}{err} \cup (\fof{\chi_s}{\textbf{loop } s_1\ s_2}{cll} \circ \chi_e).
        \label{eq:loop_err_ka}
    \end{aligned}
\end{align}
Finally, by properties of extended KATs (using tactic {\tts solvefix}), proof goals (\ref{eq:loop_nrm_ka}) and (\ref{eq:loop_err_ka}) are solved.

Although these equations may appear lengthy, they are actually quite intuitive. For instance, the left-hand side of Formula~(\ref{eq:loop_nrm_ka})
  says that after repeatedly executing $s_1$ and $s_2$ several times---where each time either $s_1$ terminates normally, or $s_1$ exits early due to a \textbf{continue} and then $s_2$ terminates normally---iterated execution may either break in $s_1$ or $s_2$. This behavior is precisely the semantics of the loop on the right-hand side.
  Therefore, using automated proof tactics to solve them is suitable.

\begin{figure}[tbh]
    \includegraphics[width=0.6\linewidth]{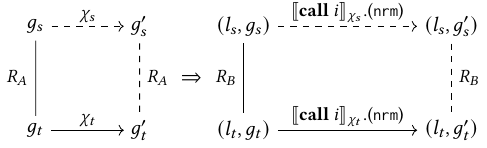}
    \caption{Proofs for the refinement of procedure calls. Here, \(R_A\) relates global states, and \(R_B\) relates pairs of local and global states. They satisfy that for any \(l_s, g_s\), and \(l_t, g_t\), if \(((l_s, g_s), (l_t, g_t)) \in R_B\), then \((g_s, g_t) \in R_A\).
    After procedure calls, if \(l_s\) and \(l_t\) is already related in \(R_B\), and for any $g_s', g_t'$, \((g_s', g_t') \in R_A\), then \(((l_s, g_s'), (l_t, g_t')) \in R_B\).
    }
    \label{fig:call_proof}
   \Description{This figure illustrates the proof idea for refinement of procedure calls. It shows how refinement of a caller’s behavior is reduced to refinement of the callee’s behavior through two state relations: \(R_A\), relating source and target global states, and \(R_B\), relating full states consisting of local and global components. The relation \(R_B\) is compatible with \(R_A\): whenever the full source and target states are related by \(R_B\), their global states are related by \(R_A\). During a procedure call, the local states of caller and callee remain unchanged, while the callee execution transforms only the global states. After the call returns, if the updated global states are related by \(R_A\), then combining them again with the unchanged local states yields full states related by \(R_B\). The figure highlights that caller refinement follows from callee refinement together with preservation of the local-state correspondence across procedure calls.}
\end{figure}

\paragbf{Case 3: $s$ is a call statement}
In our framework, the behavior refinement of procedure calls before and after the compilation depends on that of callee procedures. We use the PCALL language to illustrate the key idea of proofs for simplicity, as shown in Fig. \ref{fig:call_proof}. The refinement of the caller's behavior is derived from the refinement of the callee's behavior. The correspondence between local states is preserved since they remain unchanged before and after the procedure call.
The behavior refinement between Clight function calls and the compiled Csharpminor function calls is proved in a similar way, with additional considerations for arguments and return values.

It's worth noting that the behavior refinement of atomic statements has no algebraic structure and is verified directly according to their definitions.
We reuse some of CompCert original proofs for these cases, such as those for the transformation of expressions. 
\end{proof}
\textbf{Summary.} 
In our framework, almost all the proofs (including proofs of all the theorems in \S\ref{subsec:compile_correct} and \S\ref{subsec:module_compose}) have the following main proof steps (as shown in the proof of Lemma \ref{lemma:clitgen}):
\begin{itemize}
    \item By the algebraic properties of refinement algebras such as (\ref{prop:gunion}) $\sim$ (\ref{prop:gconcat}) and (\ref{prop:gstar}) $\sim$ (\ref{prop:ginfty}), the behavior refinement between substructures are composed into an overall refinement, without unfolding the concrete definition of $\gmm{\kN}{R_B}$ (except for proofs for atomic statements);
    \item By the properties of extended KATs, proofs are finished in a purely relational subsystem.
\end{itemize}
These two proof steps are automated in \Coq by our new {\tts gamma\_simpl} tactic and {\tts solvefix} tactic. For clarity, 
we illustrate the corresponding \Coq proofs to the first two proof cases of Lemma~\ref{lemma:clitgen} in Fig.~\ref{fig:rocq-refine-proofs},
where a predicate (\kwd{refines nbrk nctn dt ds Hree}) means
the denotation \(\kwd{dt}\) of \ClitDenote refines the denotation \(\kwd{ds}\) of \CshmDenote, namely \(\kwd{dt} \refines_{(\sss{nbrk}, \sss{nctn}, \sss{Hree})} \kwd{ds} \), corresponding to \(\sem{u}_{\chi_t} \refines_{(n_b, n_c, \chi_e)} \sem{s}_{\chi_s}\) defined in Def.~\ref{def:Cshmgen_stmt_ref}.
In other words, Fig.~\ref{fig:seq-sem-refines} means the proof obligation of Formula~(\ref{eq:Cshmgen-nrm});
and Fig.~\ref{fig:loop-sem-refines} means that if $\bracket{u_1}_{\chi_t} \sqsubseteq_{(1, 0, \chi_e)} \bracket{s_1}_{\chi_s}\! \text{ and }
    \bracket{u_2}_{\chi_t} \sqsubseteq_{(0, n_c + 1, \chi_e)} \bracket{s_2}_{\chi_s}
$, then
\begin{align}
    &\begin{aligned} \label{eq:loop_refines}
        &\sem{\textbf{block} \{ \textbf{Sloop } (\textbf{block } \{u_1\}; u_2) \}}_{\chi_t} 
        \refines_{(n_b, n_c, \chi_e)} \sem{\Sloop}_{\chi_s}.
    \end{aligned}
\end{align}

\begin{figure}[t]
  \centering
  \begin{subfigure}[t]{0.41\textwidth}
    \centering
    \begin{tabular}{|c|}
    \hline
\begin{lstlisting}[language=Coq]
 Lemma seq_sem_refines:
   forall ds1 ds2 dt1 dt2 nbrk ncnt Hree
   (EQ1: refines nbrk ncnt dt1 ds1 Hree)
   (EQ2: refines nbrk ncnt dt2 ds2 Hree),
   refines nbrk ncnt (seq_sem dt1 dt2) 
           (Clit.seq_sem ds1 ds2) Hree.
 Proof.
   preprocess;handle_special;
   try_rewrite EQ1;try_rewrite EQ2;
   gamma_simpl;fold_callee Hree;
   solvefix.
 Qed.
\end{lstlisting}
    \\ \hline
    \end{tabular}
    \caption{Lemma \code{seq\_sem\_refines}}
    \label{fig:seq-sem-refines}
  \end{subfigure}%
  \hfill
  \begin{subfigure}[t]{0.58\textwidth}
    \centering
    \begin{tabular}{|c|}
    \hline
\begin{lstlisting}[language=Coq]
 Lemma loop_sem_refines:
   forall nbrk ncnt Hree dt1 dt2 ds1 ds2 (NEQ: nbrk <> ncnt)
   (EQ1: refines 1 0 dt1 ds1 Hree)
   (EQ2: refines 0 (S ncnt) dt2 ds2 Hree),
   refines nbrk ncnt
     (Sblock_sem (loop_sem (seq_sem (Sblock_sem dt1) dt2)))
     (Clit.loop_sem ds1 ds2) Hree.
 Proof.
   preprocess; try_rewrite EQ1; try_rewrite EQ2;
   gamma_simpl;handle_special; repeat empty_simpl;
    gamma_simpl;fold_callee Hree;solvefix.
  Qed.
\end{lstlisting}
    \\ \hline
    \end{tabular}
    \caption{Lemma \code{loop\_sem\_refines}}
    \label{fig:loop-sem-refines}
  \end{subfigure}
  \caption{Refinement proofs for sequential and looping cases by automated Rocq tactics.}
  \Description{This figure contains two Rocq proof scripts that automate refinement proofs for sequential and loop constructs. Subfigure (a) shows the lemma \code{seq\_sem\_refines}. It states that if two target denotations \(dt_1\) and \(dt_2\) each refine the corresponding source denotations \(ds_1\) and \(ds_2\) under the same refinement parameters \code{nbrk}, \code{ncnt}, and \code{Hree}, then the sequential composition \code{seq\_sem dt1 dt2} refines the source sequential semantics \code{Clit.seq\_sem ds1 ds2}. The proof uses a short sequence of automated tactics, including preprocessing, rewriting with the hypotheses, simplifying \(\gamma\)-instances, folding callee relations, and solving fixed-point obligations. Subfigure (b) shows the lemma \code{loop\_sem\_refines}. It states that if \(dt_1\) refines \(ds_1\) under break and continue indices \(1\) and \(0\), and \(dt_2\) refines \(ds_2\) under indices \(0\) and \code{S ncnt}, then the translated loop built from \code{Sblock\_sem}, \code{loop\_sem}, and sequential composition refines the Clight loop semantics \code{Clit.loop\_sem ds1 ds2}. Its proof follows the same automated style, with additional simplification steps for empty sets. The figure illustrates that both sequential and looping refinement proofs can be discharged concisely by automated Rocq tactics.}
  \label{fig:rocq-refine-proofs}
\end{figure}

Three key tactics are employed in the proofs of Fig.~\ref{fig:rocq-refine-proofs}: \twrite, \tra and \tka.  
The tactic \twrite reduces proof goals by rewriting with the refinement conditions, typically of the form \(x \subseteq \gamma(y, z)\).  
The tactic \tra discharges \(\gamma\)-instances by exploiting the algebraic properties of the underlying RA instances.  
The tactic \tka resolves KAT (in-)equations by applying the equational laws of Kleene algebra with tests.  
Several auxiliary tactics are also used for supporting purposes, such as \kwd{preprocess} for unfolding definitions and simplifying by definitional equalities, \kwd{empty\_simpl} for eliminating empty sets, and \kwd{handle\_special} for handling special operators.

\paragbf{Remark} 
The \tka tactic is designed for KAT (in-)equations and can also handle simple equations with infinity operators (i.e., those grammatically equal), which are sufficient in our setting. More general cases with infinity operators are undecidable and must be handled manually. Atomic statements lack algebraic structure and are therefore verified directly by their definitions.

\subsection{Compilation Correctness}
\label{subsec:compile_correct}

The CompCert front-end compiles Clight to Cminor through the Cshmgen and Cminorgen phases. We use $\mathcal{T}_\ttt{Cshmgenf}$ for the transformation of Clight functions by Cshmgen, and $\mathcal{T}_\ttt{Cmingenf}$ for the transformation of Csharpminor functions by Cminorgen, whose correctness is shown as follows.

\subsubsection{Compilation Correctness from Clight to Cminor}
 We first define the identity Kripke interface relation as: \(A_\kU \DEF \pair{\kwd{Unit}, \leadsto_\kU, \lambda w. \idrel, \lambda w. \idrel}\), \(\leadsto_\kU \,\DEF \kunit \times \kunit\), and \kunit denotes a unit set. The KIR \(A_\kC\) is defined in Def.~\ref{def:cmin_kir}, and ``\(\refines_{\_ \ttarrow \_}\)'' is defined in Def.~\ref{def:ref_function}. We then have the following theorems.
\begin{theorem} [Cshmgen correctness] \label{thm:cshmgen}
    For any Clight function $F_s$ and Csharpminor function $F_u$,
    \begin{align*}
        \text{if }\, \mathcal{T}_\ttt{Cshmgenf} (F_s) = \ksome (F_u), \text{ then }
        \bracket{F_u} \sqsubseteq_{A_\kU \ttarrow A_\kU} \bracket{F_s}.
    \end{align*}
\end{theorem}

Remark that Thm. \ref{thm:cshmgen} is deduced from the translation correctness of Clight statements (i.e., Lemma \ref{lemma:clitgen}), and the original CompCert proof is reused for verifying the translation of initialization when entering the function and the translation of memory releasing when leaving the function.

\begin{theorem} [Cminorgen correctness]
\label{thm:cmingen}
    For any Csharpminor function $F_u$, Cminor function $F_t$,
    $$\text{if }\, \mathcal{T}_\ttt{Cmingenf} (F_u) = \ksome(F_t), \text{ then }
        \bracket{F_t} \sqsubseteq_{A_\kC \ttarrow A_\kC} \bracket{F_u}.
    $$
\end{theorem}

\subsubsection{Optimization Correctness}
To demonstrate the applicability of our approach beyond the CompCert front-end, we model typical back-end optimizations of CompCert (i.e., standard constant propagation and dead code elimination) on extended PCALL and CFG languages that additionally account for event traces. We refer to these as \emph{trace-enriched} PCALL/CFG languages. Their denotations have the signature \FDenote, and the associated correctness theorems are given as follows.
\begin{theorem} [Soundness of constant propagation on PCALL]
\label{thm:cp_pcall}
    If a trace-enriched PCALL program $p$ is optimized into $\mathcal{T}_\ttt{cp}(p)$ by constant propagation, 
    then $\bracket{\mathcal{T}_\ttt{cp}(p)} \sqsubseteq_{A_\kU \ttarrow A_\kU} \bracket{p}$.
\end{theorem}

\begin{theorem} [Soundness of dead code elimination on PCALL]
\label{thm:dce_pcall}
    If a trace-enriched PCALL program $p$ is optimized into $\mathcal{T}_\ttt{dce}(p)$ by dead code elimination, 
    then $\bracket{\mathcal{T}_\ttt{dce}(p)} \sqsubseteq_{A_\kU \ttarrow A_\kU} \bracket{p}$.
\end{theorem}
Similarly, we have optimization correctness on the trace-enriched control-flow-graph language.
\begin{theorem} [Correctness of constant propagation on CFGs]
\label{thm:cp_cfg}
    If a trace-enriched CFG $p$ is optimized into $\mathcal{T}_\ttt{cp}(p)$ by constant propagation, 
    then $\bracket{\mathcal{T}_\ttt{cp}(p)} \sqsubseteq_{A_\kU \ttarrow A_\kU} \bracket{p}$.
\end{theorem}

\begin{theorem} [Correctness of dead code elimination on CFGs]
\label{thm:dce_cfg}
    If a trace-enriched CFG $p$ is optimized into $\mathcal{T}_\ttt{dce}(p)$ by dead code elimination, 
    then $\bracket{\mathcal{T}_\ttt{dce}(p)} \sqsubseteq_{A_\kU \ttarrow A_\kU} \bracket{p}$.
\end{theorem}
Additionally, we have the correctness of generating CFGs from syntax trees of PCALL as follows.
\begin{theorem} [Correctness of CFGs generation]
\label{thm:cfg_gen}
    If a trace-enriched CFG $\mathcal{T}_\ttt{cfg}(p)$ is generated from a trace-enriched PCALL program $p$, then $\bracket{\mathcal{T}_\ttt{cfg}(p)} \sqsubseteq_{A_\kU \ttarrow A_\kU} \bracket{p}$.
\end{theorem}

\subsection{Compositional Correctness}
\label{subsec:module_compose}

Based on the definition of refinement between function denotations in Def. \ref{def:ref_function}, we then have the following results, where Lemma \ref{lemma:fun2module} helps derive the refinement of modules from the refinement of functions, Thm. \ref{thm:vertical} is used to compose the correctness of each compilation phase, and Thm. \ref{thm:horizon} shows the correctness of module-level compositionality.

Specifically,
in our implementation, we follow CompCertO's approach \cite{DBLP:conf/pldi/KoenigS21,DBLP:journals/pacmpl/ZhangWWKS24} to support vertical compositionality and use denotation-based semantic linking along with fixed-point-related theorems to support horizontal compositionality.
\begin{lemma} 
    \label{lemma:fun2module}
    For any open modules $M_s$ and $M_t$, and for any Kripke interface relations \(A\) and \(B\),
    \[
    \text{if } \parents{M_t} \sqsubseteq_{A \ttarrow B} \parents{M_s}, 
    \text{ then } 
    \bracket {M_t} \sqsubseteq_{A \ttarrow B} \bracket{M_s}.
    \]
\end{lemma}

\begin{theorem} [Vertical compositionality]
\label{thm:vertical}
For any modules $M_1$, $M_2$, $M_3$, and KIRs \(A_1, A_2, B_1, B_2\),
\[
  \text{if }
  \sem{M_3} \refines_{A_2 \ttarrow B_2} \sem{M_2}
  \;\text{and}\;
  \sem{M_2} \refines_{A_1 \ttarrow B_1} \sem{M_1},
  \text{ then }
  \sem{M_3} \refines_{A_2 \bcirc A_1 \ttarrow B_2 \bcirc B_1} \sem{M_1}.
\]

\end{theorem}

\begin{theorem} [Horizontal compositionality]
\label{thm:horizon}
    For any modules $S_1$, $S_2$, $T_1$ and $T_2$, and a KIR \(A\),
\[
  \text{if } 
  \bracket{T_1} \sqsubseteq_{A \ttarrow A} \bracket{S_1}
  \;\text{ and }\;
  \bracket{T_2} \sqsubseteq_{A \ttarrow A} \bracket{S_2},
  \text{ then }
  \bracket{T_1} \oplus \bracket{T_2}
  \sqsubseteq_{A \ttarrow A}
  \bracket{S_1} \oplus \bracket{S_2}.
\]
    
\end{theorem}
The horizontal compositionality (termination case) is elegantly proved by the following lemma, where \(\mu x.f(x)\) can be \((\semlink{T_1}{T_2})_{\chi_t}.(\knrm)\), \(\mu x.g(x)\) can be \((\semlink{S_1}{S_2})_{\chi_s}.(\knrm)\), and \(\mu x.h(x)\) can be \((\semlink{S_1}{S_2})^{\chi_e}_{\chi_s}(\kerr)\)\footnote{\((\semlink{S_1}{S_2})^{\chi_e}_{\chi_s}(\kerr)\) represents a fixed point equivalent to  \( (\semlink{S_1}{S_2})_{\chi_s}(\kerr) \cup (\semlink{S_1}{S_2})_{\chi_s}(\kcll) \circ \chi_e\).}. 
This lemma itself can be easily proved by induction.
Other proof obligations can be systematically solved by the \tra and \tka tactics akin to proofs for Formula~(\ref{eq:loop_refines}). 
\begin{lemma}
Given a refinement algebra $\mcal{A}{} = (\mS, \mE, \mT, \gamma)$,  
with CPOs $(\mS, \subseteq)$, $(\mE, \subseteq)$, and $(\mT, \subseteq)$,  
and Scott-continuous functions
\(
   f : \mT \to \mT,\,
   g : \mS \to \mS,\,
   h : \mE \to \mE.
\)
We have:
\begin{gather*}
  \bigl(
    \forall N_t \in \mT,\, N_s \in \mS,\, E_s \in \mE,\;
    N_t \subseteq \gamma(N_s, E_s)
    \imply
    f(N_t) \subseteq \gamma(g(N_s), h(E_s))
  \bigr) 
  \imply
  \\
   \mu x.\, f(x)
   \;\subseteq\;
   \gamma\bigl(\mu x.\, g(x),\; \mu x.\, h(x)\bigr).
\end{gather*}    
\end{lemma}

By the horizontal compositionality (Thm. \ref{thm:horizon}) and the equivalence between semantic linking and syntactic linking (Thm. \ref{thm:ss_equiv}), we have the following separate compilation correctness.
\begin{lemma} [Compositional correctness of the CompCert front-end]
    \label{corol:ccc-compcert}
    For any Clight open modules $S_1,\dots, S_n$ and Cminor open modules $T_1,\dots, T_n$, if $\ T_i$ is separately compiled from $S_i$ by the CompCert front-end for each $i$, then
    \(
     \sem{T_1} \oplus \cdots \oplus \sem{T_n}
      \sqsubseteq_{A_\kC \ttarrow A_\kC} 
     \sem{S_1} \oplus \cdots \oplus \sem{S_n}
    \).
\end{lemma}

\begin{theorem} [Final theorem]
    \label{thm:final-theorem}
    For any Clight open modules $S_1,\dots, S_n$ and Cminor open modules $T_1,\dots, T_n$, if $\ T_i$ is separately compiled from $S_i$ by the CompCert front-end for each $i$, then
    \[
     \sem{T_1 + \cdots + T_n}
      \sqsubseteq_{A_\kC \ttarrow A_\kC} 
     \sem{S_1 + \cdots + S_n}.
    \]
\end{theorem}

\paragbf{Handling global environments} 
The global environment $\gvar$ maps functions to their definitions and does not vary.
 Suppose that \(F\) is a function of program \(P\).
 For \(F\)'s semantics \(  \sem{F}^\gvar\),  the global environment \(\gvar\) is loaded from the program \(P\) following the original loading process of CompCert. 
 
When a compilation phase modifies the memory layout (e.g., Cminorgen), it is essential that the global environment be preserved under memory injection. Given a global environment $\gvar$, we extend the Kripke relation $R_B$ with respect to $\gvar$ (denoted $\msal{R}{\gvar}{B}$) as follows: for any world $(f, m_1, m_2)$,
 \begin{align}
    \forall a\,b,\;
    (a, b) \in \msal{R}{\gvar}{B}(f, m_1, m_2) \iff
    (a, b) \in R_B(f, m_1, m_2) \land
    \code{genv\_preserved}(\gvar, f, m_1, m_2).
    \label{eq:genv-extend}
 \end{align}
Similarly, a Kripke interface relation \(B = \pair{W_B, \leadsto_B, \msal{R}{q}{B}, \msal{R}{q}{B}}\) is extended under a global environment $\gvar$, if the Kripke relations \(\msal{R}{q}{B}\) and \( \msal{R}{r}{B}\) are extended under \(\gvar\) according to Formula~(\ref{eq:genv-extend}). 
Readers may safely skip these technical details 
without compromising the conceptual grasp of our framework.

\section{Comparison with small-step-based Compiler Verification}
\label{sec:comparison}
We compare our work with small-step-based approach in two aspects: 
 (1) how we prove behavior refinement of statements and procedures; and (2) how we support module-level compositionality.

\subsection{Comparison in Behavior Refinement}
\label{subsec:compare_bref}
The proof of behavior refinement relies on establishing an invariant (i.e., a matching relation) between source and target program states, which must be preserved throughout execution. The denotation-based approach, by focusing on overall outcomes rather than step-by-step processes, simplifies the construction of such invariants and alleviates the concern of accounting for intermediate and nondeterminism correspondence. We illustrate the two points using concrete examples.

\paragbf{Simplification of program states and invariants}

\begin{figure}[t]
\centering
\setlength{\fboxsep}{6pt}
\fbox{%
\begin{minipage}{0.95\linewidth}
\small
\textbf{Original invariant of \textsf{Cminorgen} (informal) includes:} \\[4pt]
\begin{tabular}{@{}ll@{}}
\kwd{match\_function} & The source and target functions currently executed are related. \\
\kwd{match\_statement} & The source and target statements currently executed are related. \\
\kwd{match\_continuation} & The source and target continuations (control contexts) are related. \\
\kwd{value\_inject} & Function arguments and return values are related by the injection $f$. \\
\kwd{match\_memory} & The source and target memories are related by injection 
  ($\minj{f}{m_s}{m_t}$). \\[2pt]
\kwd{match\_callstack} & Each stack frame corresponds with the following relations: \\
\quad -- \kwd{match\_tenv} & Temporary environments correspond under injection 
  ($\tinj{f}{te_s}{te_t}$). \\
\quad -- \kwd{match\_env} & Local environments are related, with pointers related by $f$. \\
\quad -- \kwd{padding\_freeable} & Padding target memory between local variables can be safely released. \\
\quad -- \kwd{match\_bounds} & Local variable addresses lie in valid memory blocks 
  (tracked by \kwd{nextblock}). \\
\end{tabular}
\end{minipage}}
\caption{Informal description of the simulation invariant for \textsf{Cminorgen} in CompCert.}
\Description{This figure summarizes, in informal terms, the simulation invariant used in CompCert’s \textsf{Cminorgen} proof. It lists the components that must be related between a source state and a target state during simulation. These include the currently executing source and target functions, the current statements, and the continuations or control contexts. It also requires that function arguments and return values are related by a memory injection \(f\), and that the source and target memories are related by the same injection. For each corresponding stack frame, the invariant additionally relates temporary environments, local environments with pointers mapped by \(f\), the ability to safely free padding memory between local variables in the target, and validity bounds for local variable addresses tracked by \texttt{nextblock}. Overall, the figure highlights that CompCert’s original front-end simulation invariant must simultaneously maintain correspondence across functions, statements, continuations, values, memories, and stack-frame environments.}
\label{fig:cminorgen-invariant}
\end{figure}

In our framework, we presented the definitions of program states for CompCert front-end languages in Table~\ref{tab:clight-sem-component}, define the \textsf{Cshmgen} proof invariant as the identity relation in Def.~\ref{def:Cshmgen_stmt_ref}, and describe the \textsf{Cminorgen} invariant in \S\ref{subsec:krel-cminorgen}.
By comparison, a front-end state in CompCert, illustrated here using Clight as a representative example, is originally defined with the following structure, with Csharpminor and Cminor adopting a similar pattern.

\begin{lstlisting}[language=Coq]
    Clight.state = | ItnlState: function -> statement -> cont -> env -> tenv -> mem -> state
                   | Callstate: fundef -> val -> list val -> cont -> mem -> state
                   | RtrnState: val -> cont -> mem -> state.
\end{lstlisting}
Here, each state represents either:
(i) an internal execution (\kwd{ItnlState}) of a function at a particular statement with its continuation, local environment, temporary environment, and the current memory; (ii)
a function call (\kwd{Callstate}) with a function definition, arguments, continuation, and the current memory; or (iii) a function return (\kwd{RtrnState}) with a return value, continuation, and the memory.
Consequently, establishing correspondences between source and target program states requires reasoning about all components of these enriched states, as summarized in Fig.~\ref{fig:cminorgen-invariant}.

\begin{wrapfigure}{r}{0.4\textwidth}
    \centering
    \includegraphics[width=0.95\linewidth]{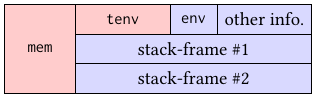}
    \caption{A front-end state in CompCert.}
    \label{fig:program_state}
    \Description{This figure illustrates the structure of a front-end program state in CompCert. The state contains several components used by the small-step semantics: the currently executing function, the current statement, the continuation or control context, the local environment, the temporary environment, and the memory state, together with related call and return states. The figure highlights that only the temporary environment and memory form the dynamic part of the state, shown in light red, because they can change during statement execution. The other components are static or control-related information. The figure is used to emphasize that the denotational approach keeps only the dynamic part in program states and treats the static local environment as a parameter instead.}
\end{wrapfigure}
In the original proofs of CompCert, all constraints over these components have to be maintained at each proof step, as illustrated in Fig.~\ref{subfig:compcert_sim}.
In comparison, the program state in denotational semantics defined in Table~\ref{tab:clight-sem-component} only needs to include the memory state and temporary environment (i.e., the real dynamic part that can be updated by store operations or variable assignments), as colored by light red in Fig.~\ref{fig:program_state}.
The local environment (i.e., the static part that is not modified during the execution of statements) is treated as a parameter of the semantic functions rather than as a component of the program state, and other information specific to small-step semantics is not required in our state definition. Furthermore, simplifying the state definition also reduces the complexity of establishing invariants. We quantify the simplification offered by our framework by comparing the \Coq definitions of program states and invariants in the CompCert front-end with that in our work in Table~\ref{tab:state-invariant}.

In proofs, we derive the function transformation correctness from the statement (i.e., function body) transformation correctness, as illustrated in Fig.~\ref{subfig:vcomp_sim}.
Our invariants for Cminorgen (defined in \S\ref{subsec:krel-cminorgen}) only involve the correspondence of dynamic parts (i.e., memory state and temporary environment), while the correspondence for static components (i.e., local environment and \kwd{padding\_freeable}), once established at function entry, is taken as a prerequisite for the refinement proof of statement transformation, and is preserved until function exit through the accessibility relation \(\leadsto_\sinjp\) (defined in \S\ref{subsec:krel-cminorgen}).
Thus, we lightened the proof burden of maintaining the invariant in each proof step by decomposing the refinement proof into smaller and simpler obligations.

\begin{figure}[t]
    \begin{subfigure}[b]{.48\linewidth}
    \centering
        \includegraphics[width=0.85\linewidth]{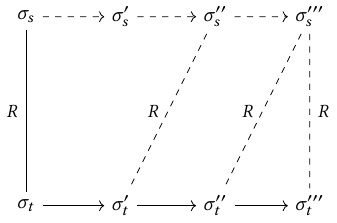}
    \caption{Small-step-based proof style}
    \label{subfig:compcert_sim}
    \end{subfigure} 
    \begin{subfigure}[b]{.48\linewidth}
        \flushright
        \includegraphics[width=0.85\linewidth]{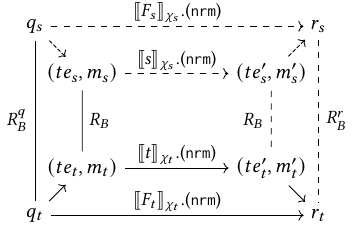}
        \centering
    \caption{Denotation-based proof style}
    \label{subfig:vcomp_sim}
    \end{subfigure}   
    \caption{Comparison between small-step-based approach and denotation-based approach in Cminorgen.}
    \label{fig:proof_decompose}
    \Description{This figure compares two proof styles for Cminorgen. Subfigure (a) shows the traditional small-step-based proof style used in CompCert. At each simulation step, the proof must maintain a large invariant relating many components of the source and target states, including the current function, statement, continuation, environments, memory, and call stack. Subfigure (b) shows the denotation-based proof style. The proof is decomposed into smaller obligations: function transformation correctness is derived from statement or function-body transformation correctness, and the invariant focuses only on the dynamic components, namely memory state and temporary environment. Static information such as the local environment is established at function entry and then preserved separately through an accessibility relation. The figure highlights that the denotation-based approach reduces the burden of maintaining complex invariants at every proof step.}
\end{figure}
  
\begin{table}[tbh]
\centering
\caption{Comparison of program state and invariant definitions in CompCert and our framework.}
\label{tab:state-invariant}
\begin{tabular}{|l|c|c|c|c|c|c|}
\hline
Definitions & Clit.state & Cshm.state & Cmin.state & \textsf{Cshmgen}.inv & \textsf{Cminorgen}.inv & Total \\
\hline
CompCert & 17 & 20 & 20 & 75 & 83 & 237 \\
Our work & 4 & 4 & 4 & 1 & 39 & 52 \\
\hline
\end{tabular}
\end{table}

\paragbf{Elimination of intermediate correspondence}
CompCert uses forward simulation to verify each compilation phase. This works well when each step of the source program indeed corresponds to multiple steps in the target program.
However, the correspondence between the source program and the target program can sometimes be many-to-one instead of one-to-many. In this case, verifying the forward simulation requires additional proof steps for the construction of the intermediate state correspondence.
In more complex compilation phases, no matter whether using forward or backward simulation, additional intermediate correspondences must be constructed. For instance, CompCertM~\cite{DBLP:journals/pacmpl/SongCKKKH20} requires mixed simulation to verify compilation correctness due to its modifications of interaction semantics. In comparison, the denotation-based approach can simplify such verification process by focusing on the overall execution of programs.
\begin{figure}[bht]
    \begin{subfigure}[b]{.47\linewidth}
    \centering
        \includegraphics[width=0.95\linewidth]{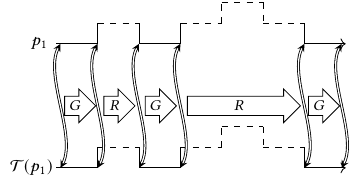}
    \caption{Refinement in small-step semantics}
    \label{subfig:small_ref}
    \end{subfigure} 
    \begin{subfigure}[b]{.52\linewidth}
        \flushright
        \includegraphics[width=0.95\linewidth]{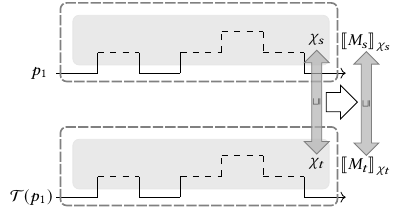}
    \centering
    \caption{Refinement in denotational semantics}
    \label{subfig:denote_ref}
    \end{subfigure}
    \caption{Comparison between the refinement process for small-step semantics and that for denotational semantics. In Fig. \ref{subfig:small_ref}, $R$ and $G$ respectively represent rely and guarantee conditions, vertical curved arrows denote the matching relation between the source and the target program states held throughout the execution,
    and vertical shaded arrows in Fig. \ref{subfig:denote_ref} denote the refinement between function denotations.}
    \Description{This figure compares refinement for open modules in small-step semantics and in denotational semantics. Subfigure (a) shows the small-step style. An incoming call to the target module is matched with a corresponding incoming call to the source module, and then the proof proceeds step by step through execution. Each outgoing external call of the target must match a corresponding outgoing external call of the source, and their returns must also match before the proof can continue. Throughout the execution, a state-matching relation is maintained between source and target states, while rely and guarantee conditions constrain how external calls and internal computations may affect the calling context. Subfigure (b) shows the denotational style. Instead of matching execution step by step, the semantics of an open module is viewed as mapping callee behavior to caller behavior through semantic oracles. Refinement is expressed vertically between the target and source function denotations: if the external source-function behavior is refined by the external target-function behavior, then the denotation of the compiled target module is refined by that of the source module. The figure emphasizes that the denotational approach captures rely and guarantee conditions through refinement between semantic oracles and module denotations, avoiding step-by-step matching of intermediate states and calls.}
    \label{fig:compare_ref}
\end{figure}

\def\kx{\kwd{x}} \def\ky{\kwd{y}}

We use the following examples (previously presented in \S\ref{subsec:compcert_bg}) to demonstrate that handling intermediate correspondence and nondeterminism is straightforward in our framework.
\begin{gather*}
    \begin{aligned}
      &\text{Source program } S:\ \ \kwd{print(x); y = f(x); print(y);} \\
      &\text{Transformed to } T: \ \ \kwd{y = g(x); print(x); print(y);}
    \end{aligned}
    \\
   \begin{aligned}
      &\text{Source program } S': \ \ \kwd{choice(x = 0, x = 1); print(0); print(x);} \\
      &\text{Transformed to } T': \ \ \kwd{print(0); choice(x = 0, x = 1); print(x);}
   \end{aligned}
\end{gather*}
Assume that they are type-safe programs. The following semantics is derived by definitions, where \(\sem{\kwd{print(x)}} \DEF \Set{(\sigma, \tau, \sigma) | \tau = \sigma(\kx) }\), and \(\sigma(\kx)\) means the value of a temporary variable \(\kx\) on state \(\sigma\).
\begin{gather*}
\begin{aligned}
    & \sem{S} = \Set{(\sigma, \tau, \sigma') | (\sigma,\knil,  \sigma') \in \sem{f} \land \tau = \kout(\sigma(\kx)) \kcons \kout(\sigma'(\ky)) \kcons \knil }
    \\
    &\sem{T} = \Set{(\sigma, \tau, \sigma') | (\sigma, \knil, \sigma') \in \sem{g} \land \tau = \kout(\sigma(\kx)) \kcons \kout(\sigma'(\ky)) \kcons \knil  }
\end{aligned}
\\
\begin{aligned}
    &\sem{S'} = \Set{(\sigma, \tau, \sigma') | \tau = \kout(\kwd{0}) \kcons \kout(\sigma'(\kx))\kcons \knil \land 
        \bigl( \sigma' = \sigma[x \mapsto \kwd{0}] \vee
               \sigma' = \sigma[x \mapsto \kwd{1}] \bigr)} 
    \\
    &\sem{T'} = \Set{(\sigma, \tau, \sigma') | \tau = \kout(\kwd{0}) \kcons 
    \kout(\sigma'(\kx))\kcons \knil \land 
        \bigl( \sigma' = \sigma[x \mapsto \kwd{0}] \vee
               \sigma' = \sigma[x \mapsto \kwd{1}] \bigr)}
\end{aligned}
\end{gather*}
As a result, we can straightforwardly conclude that \(\sem{S'} = \sem{T'}\) and if \(\sem{g} \refines \sem{f}\), then \(\sem{T} \refines \sem{S}\).

\subsection{Comparison in Supporting Verified Compositional Compilation}

Verified compositional compilers have been an important research topic for a long time, as evidenced by extensive research efforts such as 
\citeN{DBLP:conf/esop/BeringerSDA14},
\citeN{DBLP:conf/cpp/RamananandroSWKF15},
\citeN{DBLP:conf/popl/StewartBCA15},
\citeN{DBLP:journals/pacmpl/SongCKKKH20},
\citeN{DBLP:conf/pldi/KoenigS21}, and
\citeN{DBLP:journals/pacmpl/ZhangWWKS24}.

The common style of developing compositional semantics for supporting module-level compositionality is based on small-step semantics 
(including~\citeN{DBLP:conf/esop/BeringerSDA14},~\citeN{DBLP:conf/popl/StewartBCA15},~\citeN{DBLP:journals/pacmpl/SongCKKKH20},~\citeN{DBLP:conf/pldi/KoenigS21}, and~\citeN{DBLP:journals/pacmpl/ZhangWWKS24}),
in which the description of behavior refinement must introduce extra mechanisms to capture inter-module function calls. Because C functions in a module can both be called by other modules and call other modules, a C module’s compilation correctness described by small-step semantics, as shown in Fig. \ref{subfig:small_ref}, looks like: 
\begin{itemize}
    \item If an incoming call in the target language corresponds to a call in the source language (i.e., the arguments and programs states match with each other), then the first outgoing calls would match; and further, if the returned value and program states of these two first outgoing calls match with each other, then the second outgoing calls would match, etc. until the original incoming calls return.

    \item For horizontal compositionality, the compilation correctness of internal C functions relies on external
calls satisfying certain well-behavedness conditions (known as rely-conditions; e.g., external calls
do not modify the private memory of callers). In turn, the execution of internal C functions itself guarantees some well-behavedness conditions (known as guarantee-conditions, e.g., they do
not modify the private memory of their calling environments).
\end{itemize}

In comparison, the definition of our behavior refinement is much different.
  Since the semantics of open modules is modeled as a function from callee's denotation to caller's denotation, i.e., the behavior of every external call is interpreted by semantic oracle $\chi_s$ (for the source module) and $\chi_t$ (for the target module),
our compilation correctness is described as follows (shown in Fig. \ref{subfig:denote_ref}): 
  if the behavior of external source functions is refined by the behavior of external target functions, then the denotation of source module is refined by the denotation of the compiled target module.
The refinement of semantic oracles (i.e., callee's behavior) provides rely conditions for the refinement of caller's behavior, and the refinement of caller's behavior itself satisfies guarantee conditions.

 \citeN{DBLP:conf/cpp/RamananandroSWKF15} develop a compositional semantics on top of small-step semantics by
modeling the behavior of external calls as special events.
Such a special event records the procedure name and the memory state before and after the external call.
When composing the behavior of multiple procedures, one has to interpret these special traces, in which when an external procedure is called, instead of asking an oracle to obtain its semantics, the special event trace of the invoked procedure is used to interpret it, and so on.
For example, consider two procedures $p$ and $q$ that form two separate modules and do not use any memory state for simplicity.
If the behavior of $p$ is sequentially calling $q$ three times, and the behavior of $q$ is sequentially outputting zero for two times, then their semantics can be $\bracket{p}$ = {\tts Extcall}($q$)::{\tts Extcall}($q$)::{\tts Extcall}($q$)::{\tts nil}
and $\bracket{q}$ = {\tts OUT}(0)::{\tts OUT}(0)::{\tts nil}, where {\tts Extcall}($q$) is a special event and {\tts OUT}(0) is the original event of CompCert.
After the semantic linking, the behavior of $p$ would be outputting zero for six times, i.e., {\tts OUT}(0)::{\tts OUT}(0)::{\tts OUT}(0)::{\tts OUT}(0)::{\tts OUT}(0)::{\tts OUT}(0)::{\tts nil}.
That is, they compose semantics by replacing special events (based on the small-step semantics of CompCert), which can be viewed as progressing incrementally over the program's execution time. The key to proving behavior refinement for them is to establish the correspondence between events step by step. Therefore, their main proof structure is similar to that of small-step semantics. In comparison, our refinement proofs are built upon RA instances and are achieved in a more compositional way. 

\subsection{Comparison in vertical compositionality}

\begin{figure}[bht]
    \begin{subfigure}[b]{.53\linewidth}
        \flushright
        \includegraphics[width=0.9\linewidth]{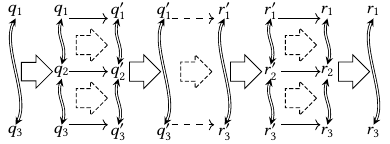}
    \centering
    \caption{Vertical composition in small-step semantics}
    \label{subfig:small_vert}
    \end{subfigure}
    \begin{subfigure}[b]{.46\linewidth}
    \centering
    \includegraphics[width=0.832\linewidth]{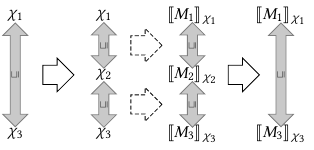}
    \caption{Vertical composition in our framework}
    \label{subfig:denote_vert}
    \end{subfigure}
    \caption{Comparison between vertical composition of open simulations and vertical composition of behavior refinement in our framework, where $q_i$ and $r_i$ represents the query and the reply of a function call in module $M_i$. The solid arrow $\rightarrow$ and dashed arrow $\dashrightarrow$ respectively mean internal steps and external steps. }
    \Description{This figure compares two styles of vertical composition for open-module correctness across two compilation phases, from module \(M_1\) to \(M_2\) and then to \(M_3\). Subfigure (a) shows the small-step style based on open simulations. Starting from an incoming query \(q_3\) to \(M_3\), the proof first splits its correspondence through an intermediate query \(q_2\) and then to \(q_1\). After internal execution steps, shown by solid arrows, each module may issue outgoing external queries \(q_1'\), \(q_2'\), and \(q_3'\); their pairwise correspondences must be merged into a direct correspondence between \(q_1'\) and \(q_3'\). The corresponding outgoing replies are then split again through the intermediate module, and after more internal steps the incoming replies \(r_1\), \(r_2\), and \(r_3\) are merged to obtain a direct correspondence between \(r_1\) and \(r_3\). Dashed arrows denote external steps, and the thick arrows emphasize the nontrivial splitting and merging of correspondences across phases. Subfigure (b) shows vertical composition in the denotational framework. Instead of following the interactive execution step by step, the proof has two higher-level steps: first, refinement of callee behavior between \(M_1\) and \(M_3\) is decomposed through the intermediate module \(M_2\); second, the caller-behavior refinements for \(M_3\) versus \(M_2\) and for \(M_2\) versus \(M_1\) are merged to derive refinement of \(M_3\) versus \(M_1\). The figure emphasizes that the denotational approach reasons about overall caller and callee behaviors, avoiding the repeated stepwise splitting and merging required by small-step simulations.}
    \label{fig:compare_vert}
\end{figure}

Consider three modules \(M_1\), \(M_2\), and \(M_3\) such that \(M_1\) is compiled into \(M_2\), and \(M_2\) is further compiled into \(M_3\) by two sequential compilation phases. The vertical compositionality is to show if \(\bracket{M_3} \sqsubseteq \bracket{M_2}\) and \(\bracket{M_2} \sqsubseteq \bracket{M_1}\), then \(\bracket{M_3} \sqsubseteq \bracket{M_1}\).  
We compare the vertical composition of behavior refinement in our framework with the vertical composition of open simulations between interactive semantics in Fig. \ref{fig:compare_vert}, where dashed fat arrows mean direct deduction by assumptions we know, and solid fat arrows mean that one needs to split the direct correspondence or merge the correspondence of two compilation phases (usually quite nontrivial). As shown in Fig. \ref{subfig:small_vert}, the vertical composition of open simulations typically proceeds as follows:
\begin{itemize}
  \item \emph{Split the initial correspondence of incoming queries}. The direct correspondence between incoming calls with initial queries \(q_1\) and \(q_3\) is decomposed. That means, there exists an incoming query \(q_2\) such that \(q_3\) corresponds to \(q_2\), and \(q_2\) corresponds to \(q_1\).

  \item \emph{Merge correspondences of outgoing queries}. After zero or more internal steps, an outgoing call may occur. By assumption, outgoing queries (\(q_i'\) for module \(M_i\)) are separately related, which are then merged into a direct correspondence between \(q_1'\) and \(q_3'\).  

  \item \emph{Split the correspondence of outgoing replies}. Following the outgoing calls with related queries \(q_1'\) and \(q_3'\), their corresponding replies \(r_1'\) and \(r_3'\) are assumed to be directly related. This correspondence is decomposed again for further deduction of internal steps.

  \item \emph{Merge correspondences of incoming replies}. After zero or more internal steps, the replies of initial incoming queries are separately related by assumptions. Finally, by merging relations of each phase, a direct correspondence between incoming replies \(r_1\) and \(r_3\) is established.
\end{itemize}
In comparison, the vertical composition in our framework, as shown in Fig. \ref{subfig:denote_vert}, goes through two steps:
  (i) decomposing the refinement of callee's behaviors between module $M_1$ and $M_3$; and (ii) merging the refinements of caller's behaviors between $M_1$ and $M_2$, and between $M_2$ and $M_3$. In other words, our framework approaches vertical composition from the perspectives of overall behaviors of the callee and the caller functions, whereas the small-step-based approach through the lens of the program's interactive process.
The connection between these two perspectives lies in that the first and last solid fat arrow in Fig. \ref{subfig:small_vert} corresponds to the last fat arrow in Fig. \ref{subfig:denote_vert}, and the second and third solid fat arrow in Fig.\ref{subfig:small_vert} corresponds to the first fat arrow in Fig. \ref{subfig:denote_vert}.

\begin{figure}[b]
  \begin{subfigure}[b]{0.33\textwidth}
  \includegraphics[width=0.642\linewidth]{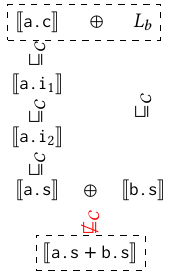}
  \centering
  \caption{CompComp}
  \label{fig:const-refinement}
  \end{subfigure}
  \begin{subfigure}[b]{0.28\textwidth}
  \includegraphics[width=0.96\linewidth]{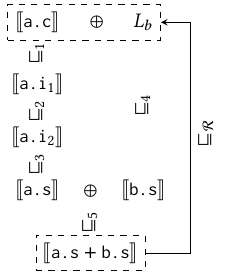}
  \centering
  \caption{CompCertM}
  \label{fig:sum-refinement}
  \end{subfigure}
  \begin{subfigure}[b]{0.33\textwidth}
  \includegraphics[width=0.817\linewidth]{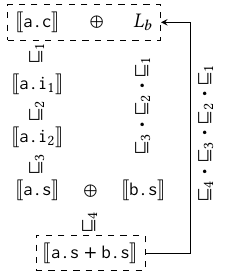}
  \centering
  \caption{CompCertO}
  \label{fig:prod-refinement}
  \end{subfigure}
  \caption{Refinements in existing VCC methods.
Adapted from~\citeN{DBLP:journals/pacmpl/ZhangWWKS24}.
\copyright\xspace 2024 the owner/author(s). Licensed under CC BY 4.0
(\url{https://creativecommons.org/licenses/by/4.0/}).
Changes: replaced $\precsim$ with $\sqsubseteq$.}
  \label{fig:refinements}
  \Description{This figure compares three refinement designs used in existing verified compositional compiler frameworks. Subfigure (a), labeled CompComp, depicts a single uniform refinement relation used across all compilation phases. This relation supports both vertical and horizontal composition, but it is heavyweight and tied to a restricted calling convention. Subfigure (b), labeled CompCertM, shows several phase-specific refinement relations feeding into one uniform refinement relation. The idea is that each individual compilation phase may be verified with its own refinement, and then lifted to a common refinement that supports composition. Subfigure (c), labeled CompCertO, shows refinement relations composed by concatenation along a compilation chain. This makes vertical composition straightforward, because consecutive refinements can be chained through intermediate languages, but it exposes the intermediate semantics and complicates horizontal composition with separately compiled environment modules. Overall, the figure contrasts the tradeoffs among a single uniform refinement, a family of refinements with a common target, and explicit concatenation of phase-specific refinements.}
\end{figure}

Though the above process of vertical composition seems intuitive, constructing interpolating program states to transitively connect evolving source and target states across external calls of open modules turns out to be a big challenge, as evidenced by a series of efforts~\cite{hur2012transitive,DBLP:conf/icfp/NeisHKMDV15,DBLP:journals/pacmpl/0001A19,DBLP:journals/pacmpl/SongCKKKH20,DBLP:journals/pacmpl/ZhangWWKS24}.
\citeN{DBLP:journals/pacmpl/ZhangWWKS24} make a comprehensive conclusion of challenges and approaches to
verified compositional compilers, as shown in Fig.~\ref{fig:refinements}.
Generally, an ideal VCC framework expects to find a refinement relation ``\(\refinesymb\)'' satisfying the following properties, where \(L\) (with subscripts) denotes the semantics of open modules. 
\begin{align}
    &\textbf{Vertical compositionality:} 
      &&L_3 \refinesymb L_2 \Rightarrow
        L_2 \refinesymb L_1 \Rightarrow
        L_3 \refinesymb L_1
    \\
   &\textbf{Horizontal compositionality:}
      &&L_2 \refinesymb L_1 \Rightarrow
        L_2' \refinesymb L_1' \Rightarrow
        L_2 \oplus L_2' \refinesymb
        L_1 \oplus L_1'
    \\
  &\textbf{Adequacy w.r.t. semantics:} \label{eq:adequacy}
     &&\sem{M_1 + M_2} \refinesymb \sem{M_1} \oplus \sem{M_2}
\end{align}
The final goal of VCC is to prove that, for example, if a C module \(\code{a.c}\) is compiled to assembly code
\(\code{a.s}\) going through \(\code{a.c} \rightarrow \code{a.i1} \rightarrow \code{a.i2} \rightarrow \code{a.s}\), and an environment module is separately compiled by a different compiler (typically, an assembly library \(\code{b.s}\) without any compilation), then \(\sem{\code{a.s} + \code{b.s}} \sqsubseteq \sem{\code{a.c}} \oplus L_b\), where \(L_b\) is the top-level semantic specification of \(\code{b.s}\). The proof of this goal usually proceeds as follows.

    (i) Show individual phases correctness, i.e., \(\sem{\code{a.s}} \refinesymb \sem{\code{a.i2}}\),
            \(\sem{\code{a.i2}} \refinesymb \sem{\code{a.i1}}\) and
            \(\sem{\code{a.i1}} \refinesymb \sem{\code{a.c}}\).
            
    (ii) Show end-to-end correctness by vertical compositionality, i.e., \(\sem{\code{a.s}} \refinesymb \sem{\code{a.c}}\) and \(\sem{\code{b.s}} \refinesymb L_b\).
    
    (iii) Show that \(\sem{\code{a.s}} \oplus \sem{\code{b.s}} \refinesymb \sem{\code{a.c}} \oplus L_b\) holds by horizontal compositionality.
    
    (iv) Show that \(\sem{\code{a.s} + \code{b.s}} \refinesymb \sem{\code{a.c}} \oplus L_b\) holds by adequacy (\ref{eq:adequacy}) and vertical compositionality.

\noindent
In practice, each compilation phase may use different refinement relations, since intermediate representations (or IRs) may have different ways of interaction with environment (so called \textit{language interface}). 
\textit{CompComp} uses a unified refinement relation for different compilation phases, which is both vertically and horizontally compositional but heavyweight and limited to interaction with C-like calling conventions. Thus CompComp does not support adequacy w.r.t. assembly semantics.
\textit{CompCertM} proposes a RUSC (Refinement Under Self-relation Context) theory featured with a fixed collection of refinement relations \(\{\refinesymb_1, ..., \refinesymb_n\}\) and a uniform refinement \(\refinesymb_\mathcal{R}\), so that individual compilation phases can be verified with one of \(\{\refinesymb_1, ..., \refinesymb_n\}\). The RUSC theory guarantees that given semantics of open modules \(L_1\) and  \(L_2\), if \(L_2 \refinesymb_i L_1\), then \(L_2\refinesymb_\mathcal{R} L_1\) for \(1 \leq i \leq n\), where \(\refinesymb_\mathcal{R}\) is both vertically and horizontally compositional.
\textit{CompCertO} introduces a naive product operator for concatenation of refinements, i.e.,
\(L_3 \refinesymb_2\compsymb \refinesymb_1 L_1\) iff there exists some \(L_2\) such that \(L_3 \refinesymb_2 L_2\) and \(L_2 \refinesymb_1 L_1\). Thus, vertical composition in CompCertO is trivial. However, as pointed out by Zhang et al., horizontal composition in CompCertO would expose the intermediate semantics of the compilation chain, in the sense that (see Fig. \ref{fig:prod-refinement}) 
for horizontal composition with the compilation correctness of \(\code{a.c}\), namely \(\sem{\code{a.s}} \refinesymb_3\compsymb \refinesymb_2 \compsymb \refinesymb_1 \sem{\code{a.c}}\), one has to show the environment module satisfying the same chain of refinement concatenation, i.e., \(\sem{\code{b.s}} \refinesymb_3\compsymb \refinesymb_2 \compsymb \refinesymb_1 L_b\). Zhang et al. address this problem by introducing a direct refinement of open modules (denoted as \(\refinesymb_\ttt{ac}\)). They show that the concatenation of refinements in CompCertO can be further refined to the direct refinement between C and assembly. For example, if \(\sem{\code{a.s}} \refinesymb_{3} \compsymb \refinesymb_{2} \compsymb \refinesymb_{1} \sem{\code{a.c}}\) and \(\refinesymb_{3} \compsymb \refinesymb_{2} \compsymb \refinesymb_{1}
    \ \preccurlyeq\ \refinesymb_\ttt{ac}\), then \(\sem{\code{a.s}} \refinesymb_\ttt{ac} \sem{\code{a.c}}\). In this perspective, the concatenation of refinements for CompComp and CompCertO can be refined to themselves, i.e.,  \(\refinesymb_\mathcal{C} \compsymb \refinesymb_\mathcal{C} 
    \ \preccurlyeq\ \refinesymb_\mathcal{C}\) and \(\refinesymb_\mathcal{R} \compsymb \refinesymb_\mathcal{R} 
    \ \preccurlyeq\ \refinesymb_\mathcal{R}\), while Zhang et al. achieve this with less restrictions by discovering a uniform and compositional Kripke relation for describing the evolution of memory injections across external calls (shown in Def.~\ref{def:injp} and used in our Cminorgen proof). 
    In our setting, given KIRs \(A_1, A_2\) and \(B_1, B_2\),
    \(
    \refines_{A_1 \ttarrow B_2} \compsymb \refines_{A_2 \ttarrow B_2}\) means  \(\refines_{A_2 \bcirc A_1 \ttarrow B_2 \bcirc B_1}
    \) (see Def.~\ref{def:kir}).
    {That is, the vertical compositionality shown in Example~\ref{ex:bref-open-module} follows the same approach of CompCertO and can achieve direct refinement following Zhang et al's approach.}

\section{Evaluation}
\label{sec:evaluation}
\begin{table}[hbt]
\caption{Significant lines of our code relative to CompCert v3.11.}
\centering
\label{tab:code_evaluate}
\begin{tabular}{|p{4.8cm}|r r|r r|}
\hline
 Portion & \multicolumn{2}{|c|}{Specs.} & \multicolumn{2}{|c|}{Proofs} \\
\hline
\textit{Front-end languages:} & old & new & old & new \\ 
Clight.v (\S\ref{sec:semantics_nocall}) 
    & 489 & 432 & 14 & 194 \\
Csharpminor.v (\S\ref{sec:semantics_nocall})
    & 345 & 305 & 0 & 97 \\
Cminor.v & 748 & 239 & 217 & 184 \\

\hline
\textit{Front-end compilation phases:} & \multicolumn{2}{c|}{} & \multicolumn{2}{c|}{} \\
Cshmgenproof.v & 550 & 453 & 1217 & 1322 \\
Cminorgenproof.v & 753 & 679 & 1149 & {1634} \\

\hline
\textit{Compositionality:} & \multicolumn{2}{c|}{} & \multicolumn{2}{c|}{} \\
Semantic linking (\S\ref{sec:semantics_pcall}, Thm. \ref{thm:ss_equiv})
& \multicolumn{2}{r|}{261} & \multicolumn{2}{r|}{561} \\
Refinement algebras (\S\ref{sec:refinement})
& \multicolumn{2}{r|}{362} & \multicolumn{2}{r|}{475} \\
Automated proof tactics (\S\ref{subsec:proof_style})
& \multicolumn{2}{r|}{1583} & \multicolumn{2}{r|}{1511} \\
Horizontal composition. (\S\ref{subsec:module_compose})
& \multicolumn{2}{r|}{103} & \multicolumn{2}{r|}{392} \\
Vertical composition. (\S\ref{subsec:module_compose})
& \multicolumn{2}{r|}{91} & \multicolumn{2}{r|}{146} \\

\hline
\textit{Demos on toy languages:} & \multicolumn{2}{c|}{} & \multicolumn{2}{c|}{} \\

Const. prop. on PCALL (Thm. \ref{thm:cp_pcall})
& \multicolumn{2}{r|}{650} & \multicolumn{2}{r|}{1173} \\

Const. prop. on CFG (Thm. \ref{thm:cp_cfg})
& \multicolumn{2}{r|}{196} & \multicolumn{2}{r|}{967} \\

DCE on PCALL (Thm. \ref{thm:dce_pcall})
& \multicolumn{2}{r|}{188} & \multicolumn{2}{r|}{519} \\

DCE on CFG (Thm. \ref{thm:dce_cfg})
& \multicolumn{2}{r|}{86} & \multicolumn{2}{r|}{456} \\

CFG generation (Thm. \ref{thm:cfg_gen})
& \multicolumn{2}{r|}{1031} & \multicolumn{2}{r|}{1988} \\

\hline
\textbf{Total}
& \textbf{2885} & \textbf{6659} & \textbf{2597} & \textbf{11619} \\
\hline
\end{tabular}
\end{table}

\begin{table}[htb]
    \caption{Comparison with existing compositional compilers derived from CompCert.}
    \centering
    \begin{tabular}{|p{4.2cm}|c|c|c|c|}
    \hline
    Portion  &   CompComp & CompCertM & CompCertO\footnotemark  & Our work\\
    \hline
    Based  CompCert version & v2.1 & v3.5 & v3.10 & v3.11 \\

    Front-end languages & 7042 & 2670 & 1634 & 1451 \\

    Front-end compilation phases & 11071 & 4306 & 3625 & 4088 \\

    Horizontal \& Vertical compos. & 15489 & 7656 & 948 & 732 \\
    
    \hline
    \textbf{Total}& \textbf{33602} & \textbf{14632} & \textbf{6207} & \textbf{6271} \\
    \hline
    \end{tabular}

    \label{tab:code_compare}
\end{table}

The significant lines of our \Coq development relative to CompCert v3.11 is shown in Table \ref{tab:code_evaluate}. 
  Our verification for the CompCert front-end is 5,539 lines in total (i.e., the \emph{front-end languages} and \emph{compilation phases} portion  in Table \ref{tab:code_evaluate}, including specifications and proofs), which increases only 57 lines relative to the CompCert v3.11's 5,482 lines. 
  The decrease in specification size is due to the reuse of common semantic operators (e.g., union and  composition of denotations) in different languages, as well as the simplification of program states and their invariants mentioned in \S\ref{subsec:compare_bref}.
  The increase in proof size is because we add deterministic proofs for the two phases, so as to apply the forward simulation proof of atomic statements in CompCert to our backward refinement. Furthermore, {the use of multiple behavior sets in semantic definition results in repeated proofs for atomic statements, indicating potential room for further code reduction.} In the \emph{compositionality} portion, our novel semantic linking operator takes 261 lines. The equivalence between semantic and syntactic linking (i.e., Thm. \ref{thm:ss_equiv}) takes 561 lines. We provide a full range of RA instances (including all the examples in \S\ref{subsec:ref_examples}$\sim$\S\ref{subsec:ra_instances}), taking 362 lines for definitions and 475 lines for axiom proofs. Our new proof tactic {\tts gamma\_simpl} and {\tts solvefix} take 1,583 lines for definitions and 1,511 lines for proofs. Finally, our demos on the verification of a typical CFG generation and optimizations (i.e., the \emph{demo} portion in Table \ref{tab:feature_comparison}) on toy languages take totally 2,151 lines for definitions and 5,103 lines for proofs.  

We compare our development with typical compositional compilers in Table \ref{tab:code_compare}. For the implementation of the CompCert front-end (i.e., the rows of front-end languages and compilation phases in Table \ref{tab:code_compare}), our approach does not significantly increase the code size, and yet significantly reduces the required code for supporting horizontal and vertical compositionality. As is well-known, the structured simulation approach employed by CompComp \cite{DBLP:conf/popl/StewartBCA15} extensively increases proof effort. CompCertM \cite{DBLP:journals/pacmpl/SongCKKKH20} supports compositional verification based on the RUSC theory, whose developments totally take about 7.6K code lines.
CompCertO \cite{DBLP:conf/pldi/KoenigS21} requires 948 lines for the horizontal and vertical composition of two simulations. Our compositional verification is as lightweight as CompCertO, but differently, our entire framework is directly based on backward refinement to provide better support for nondeterminism.

\footnotetext{Our statistics focuses on the latest version of CompCertO extended by Zhang et al. for supporting direct refinement.}

\begin{table}[th]
\centering
\caption{Increased significant lines of code for support new features.}
\label{tab:feature_comparison}
\setlength{\extrarowheight}{2pt}
\begin{tabular}{|l|l|c|c|r|}
\hline
Feature             & Portion                   & Before support & After support & Increased by       \\ \hline
\multirow{5}{*}{Procedure calls}
     & Const. Prop. on PCALL      & 1596            & 1823           & 227 (14.2\%)\ \     \\
     & Const. Prop. on CFG        & 970             & 1163           & 193 (19.9\%)\ \   \\
     & DCE on PCALL               & 606             & 707            & 101 (16.7\%)\ \   \\
     & DCE on CFG                 & 450             & 542            & 92 (20.4\%)\ \    \\
     & CFG generation             & 2607            & 3019           & 412 (15.8\%)\ \   \\ \hline
\multirow{2}{*}{Uncond. branch} 
     & Front-end languages         & 1250            & 1451           & 201 (16.1\%)\ \   \\
     & Front-end phases         & 3623            & 4088           & 465 (12.8\%)\ \   \\ \hline
\end{tabular}
\end{table}

We demonstrate the lightweight extensibility of our framework by presenting the incremental code required to support two new features that were not accommodated in our earlier development.
The first one is the support of procedure calls in toy languages. Adding the support of procedure calls to the WHILE and earlier CFG languages totally requires an additional 1,025 lines of code for definitions and proofs of the constant propagation and dead code elimination optimizations (corresponding to the first five portions in Table \ref{tab:feature_comparison}).
The second feature is unconditional branching introduced in \S\ref{subsec:goto_sem}. We now support this feature in the CompCert front-end (i.e., the last two portions in Table \ref{tab:feature_comparison}), with an increase of 201 lines for the CompCert front-end languages (i.e., Clight, Csharpminor and Cminor), and 465 additional lines for proofs of the Cshmgen and Cminorgen.

\section{Discussion}
\label{sec:discussion}
\paragbf{Denotational semantics versus small-step semantics}
There is no definitive consensus on whether denotational semantics or small-step operational semantics should serve as the gold standard for semantic characterization, as both are employed in different contexts.
Denotational semantics inherently supports compositionality and is often most convenient when expressed via least fixed points (e.g., in the semantics of open modules in \S\ref{sec:semantics_pcall}).
By contrast, small-step semantics is frequently preferred for low-level languages such as assembly.
As shown in \S\ref{subsec:compile_correct} and \S\ref{subsec:module_compose}, our results are stated in terms of denotations and the algebraic notion of behavior refinement presented in \S\ref{subsec:ra_instances}. We further proved the soundness of the algebraic behavior refinement with respect to CompCert-style behavior refinement in Thm.~\ref{thm:cshm-ref-sound} and Thm.~\ref{thm:module-ref-sound}.
Consequently, relating our denotation-based refinement results back to the original CompCert presentation reduces to establishing an equivalence theorem between our denotational semantics and CompCert's standard small-step semantics for Clight.
To illustrate this connection, we provide a supplementary mechanized proof of equivalence between denotational and small-step semantics for extended WHILE and CFG languages with I/O events considered.
{The proof relies on structural induction (on the syntax of the program) or rule induction (on the state transitions of small-step semantics).}
However, this proof does not yet cover three features of Clight---namely, procedure calls, unconditional branching (e.g., \textbf{goto}), and structured control flow (e.g., \textbf{break} and \textbf{continue}). 
Extending the equivalence proof to these remaining constructs is orthogonal to the refinement results developed in this article and is left to future work.

\paragbf{Denotational semantics versus big-step semantics}
Big-step operational semantics, widely used in compiler correctness and type-safety proofs, describes how a program evaluates in terms of computations that directly map initial states to final results. For example, the CakeML compiler~\cite{DBLP:conf/popl/KumarMNO14} employs a functional big-step semantics, where divergence is treated using a step-indexed construction rather than a coinductive relation.
Furthermore, coinductive approaches to big-step semantics have been investigated, such as the work of \citeN{DBLP:journals/scp/AnconaDRZ20}, which extends standard big-step semantics with inference systems and corules to capture infinite computations.
More recently, \citeN{DBLP:journals/toplas/ChargueraudCEG23} proposed omni-big-step semantics, which defines a single execution step from an initial state to a set of potential outcomes, demonstrating its applicability to compiler correctness, type-safety for lambda calculi, and reasoning about partial correctness. 
Nevertheless, despite these developments, big-step semantics and denotational semantics remain fundamentally distinct, particularly in their treatment of loops and recursive constructs. Big-step semantics captures the execution of such constructs through operational rules that describe how results are produced, while denotational semantics provides a mathematical interpretation based on abstract structures such as least or greatest fixed points. Furthermore, divergence in big-step semantics is typically modeled using a coinductive relation, while denotational semantics naturally captures nontermination by the use of least or greatest fixed points in suitable semantic domains.

\paragbf{Semantics in UTP} Unifying theories of programming (UTP)~\cite{DBLP:conf/RelMiCS/JifengH98,DBLP:conf/ifm/WoodcockC04} serves as a framework for unifying different programming theories and models, primarily based on denotational semantics.
To be brief, the semantics of recursion in UTP is handled using Tarski's fixed-point theory, leveraging the mathematical framework of complete lattices. In this context, divergence is often represented as the bottom element (\(\bot\)) of the lattice, signifying a program that fails to produce a meaningful or valid output. However, using the semantic framework of UTP to uniformly characterize the various program behaviors required for compiler verification faces a key challenge: it is difficult to establish an appropriate ordering on a complete lattice. If we consider the simple subset relation on sets, neither the least fixed point nor the greatest fixed point can correctly capture divergence with event traces. For instance, the silent divergence of program ``\kwd{\textbf{while} true \textbf{do} skip}'' should be characterized using the greatest fixed point, and its reactive (or nonsilent) divergence should be captured using the least fixed point.
        Conversely, the silent divergence of program ``\kwd{\textbf{while} true \textbf{do} print(0)}'' should be described by the least fixed point, while its reactive divergence should rely on the greatest fixed point.    
    This demonstrates that a conventional ordering cannot meet the requirements. Resolving this issue might necessitate constructing a highly counterintuitive ordering.
    In our framework, all definitions remain intuitive and straightforward, since we describe these behaviors by different behavior sets with relevant semantic operators.

\paragbf{Refinement calculus and Event-B} The refinement calculus~\cite{DBLP:journals/toplas/Morgan88,DBLP:journals/scp/Morris87,DBLP:books/refine/calculus} is a formal framework that focuses on transforming high-level specifications into concrete implementations (e.g., specific data structures or algorithms), through a series of mathematically verified refinement steps. Specifically, specifications in the refinement calculus are written as abstract programs, usually in the form of specification statements. 
Each refinement step would transform a certain component of the abstract program to a more concrete one by applying one of the refinement laws (i.e., inference rules for valid refinements). Event-B~\cite{DBLP:books/daglib/0024570} closely mirrors the conceptual framework of refinement calculus and extends related concepts to state-based systems and events.
Our framework shares a common trait with these pioneering efforts: both of us can derive the refinement of composite structures through the refinement of their substructures.
In comparison, the refinement algebras in our framework are used to algebraically and uniformly
build behavior refinement between the source and target program behaviors
and to enable mechanized verification,
while refinement calculus and Event-B focus on stepwise refinement of statements so as to bridge the gap between specifications and their implementations.

\paragbf{More compilation optimizations} Beyond presented optimizations, our denotation-based framework can be useful for a wide range of other optimizations, including tail recursion optimizations, inlining optimizations, context-sensitive optimization, and specification-based optimizations.

Specifically, {tail recursion optimizations} convert recursive procedure calls into iterative constructs (e.g., a while loop). Since the semantics of both recursive procedures and loops are defined by taking fixed points, proving semantic equivalence before and after the optimization is straightforward. For instance, the procedure
``{\tts src() \{ \textbf{if} $b$ \textbf{then} $c$; src() \textbf{else} skip \}}''
 can be optimized to
``{\tts  tgt() \{ \textbf{while} $b$ \textbf{do} $c$ \}}'' for given Boolean expression $b$ and statement $c$ in which no more recursive calls appear. 
Their denotational semantics\footnote{We here consider a module with the single procedure {\tts foo} for presentation purpose, and omit details about procedure names, arguments, and local variables in the semantic definition since they will not affect the final result.}, respectively shown in Formulae (\ref{eq:foo_src}) and (\ref{eq:foo_tgt}), are intuitively the two forms of loop's semantics, as illustrated by Example \ref{ex:domain_vs_ka}. We believe that more general cases can be proved similarly.
\begin{align}
    &\mu \chi.\fof{{\chi}}{\ttt{src}}{nrm} = \mu \chi.
    \big(\mathbf{test}(\bracket{b}.(\ttt{ffs})) \circ \idrel\big)
        \cup 
    \big(\mathbf{test}(\bracket{b}.(\ttt{tts})) \circ \fof{\chi}{c}{nrm} \circ \chi\big) \label{eq:foo_src}
    \\
    &\mu \chi.\fof{\chi}{\ttt{tgt}}{nrm}  =
     \fof{}{\ttt{while b do c}}{nrm} \label{eq:foo_tgt}
     = \big(\mathbf{test}(\bracket{b}.(\ttt{tts})) \circ \fof{\chi}{c}{nrm} \big)^* \circ
     \mathbf{test}(\bracket{b}.(\ttt{ffs}))
\end{align}

{Inlining optimizations} within a module replace procedure calls with the body of the callee procedures at the call sites. Before inlining, the denotations of all procedures in the module are first combined and then their semantics is derived by taking fixed points. After inlining, the denotations of callee procedures somehow are directly substituted into the call sites within the caller procedures, and then the semantics of caller procedures is derived by taking fixed points. In a procedural perspective, proving the correctness of inlining optimizations is conceptually similar to demonstrating the equivalence between semantic linking and syntactic linking.

{Context-sensitive optimizations} mean that
realistic compilers can make use of information from context-sensitive analyses, as demonstrated by \citeN{DBLP:conf/pldi/LattnerLA07}, \citeN{DBLP:conf/iwmm/LiCK13}, and \citeN{DBLP:journals/pacmpl/PoesiaP20}, so that a callee procedure can be optimized with the information that all calls to the callee satisfy certain properties. Typically, compilers can create multiple versions of a procedure optimized for different contexts (i.e., procedure cloning). Consider the example in Fig. \ref{fig:context_sensitive_opt},
\begin{figure}
    \centering
\begin{tabular}{|l || l|}
    \hline
      Before optimization &  After optimization\\
    {\tts callee(bool b) \{ \textbf{if} b \textbf{then} $c_1$ \textbf{else} $c_2$ \} }
         & {\tts callee\_tt() \{ $c_1$ \} \quad \quad callee\_ff() \{ $c_2$ \} }\\
    {\tts caller\_src() \{ callee(true); callee(false) \} }
         & {\tts caller\_tgt() \{ callee\_tt(); callee\_ff() \} }
         \\
    \hline
\end{tabular}
    \caption{A simple example of context-sensitive optimizations.}
    \label{fig:context_sensitive_opt}
    \Description{This figure shows a simple example of context-sensitive optimization by procedure specialization. The left column, labeled Before optimization, contains two procedures. The first is \texttt{callee(bool b)}, whose body branches to \(c_1\) when \(b\) is true and to \(c_2\) when \(b\) is false. The second is \texttt{caller\_src()}, which calls \texttt{callee(true)} and then \texttt{callee(false)}. The right column, labeled After optimization, shows the specialized version. The original callee is split into two procedures: \texttt{callee\_tt()}, whose body is just \(c_1\), and \texttt{callee\_ff()}, whose body is just \(c_2\). The caller is correspondingly rewritten as \texttt{caller\_tgt()}, which calls \texttt{callee\_tt()} and then \texttt{callee\_ff()}. The figure illustrates how a compiler can use calling-context information to replace a single conditional callee with multiple specialized versions.}
\end{figure}
where {\tts callee} is specialized according to the caller's context.
Verifying such an optimization in our framework can be reduced to 
first show the refinement of callee procedures under assumed conditions and then to show the refinement between caller procedures that guarantee the assumptions.

{Specification-based optimizations} leverage the specification of invoked procedures to optimize the caller procedure. For example, consider a procedure:
  {\tts  foo() \{ $c_1$; x = callee(x); $c_2$ \}}.
If the specification of {\tts callee} guarantees that it does not modify the variable {\tts x} and always returns a value identical to its input, the assignment {\tts x = callee(x)} can then be safely eliminated. Moreover, the information from the callee's specification can also enhance optimizations such as constant propagation and dead code elimination (DCE).
In our framework, the refinement processes for the caller and callee procedures are explicitly decoupled. This separation simplifies the implementation of specification-based optimizations, as it isolates the properties of the callee while preserving the modularity and correctness of the overall refinement.

In summary, these optimizations, particularly those involving procedure calls, are not as intuitively understood using small-step semantics. In contrast, they are much more straightforward when analyzed through the lens of denotational semantics. The latter provides a high-level, mathematical representation of program behavior, making it easier to reason about the transformations and their correctness, especially for higher-level constructs like procedure calls.

\paragbf{Liveness preservation} 
Safety properties guarantee that ``nothing bad will happen'' during program execution, such as preventing invalid states. Liveness properties, on the other hand, ensure that ``something good eventually happens'', asserting that every execution will lead to a desirable outcome within a program.
Although our refinement method primarily focuses on safety, it is important to clarify that this does not inherently preclude the verification of liveness preservation. In fact, our approach is capable of reasoning about liveness preservation if augmented with additional proof strategies or invariants tailored to liveness verification. For instance, \citeN{DBLP:journals/pacmpl/BartheBGHLPT20} develop a Constant-Time Secure (CTS) preserving compiler based on CompCert. The CTS property can be summarized as follows: given a safety program \( P \), if \( P \) is executed from any two ``indistinguishable'' initial states, then both of the two executions will produce the same leakage information (both executions terminate with the same leakage information, or they both diverge with the same leakage information), where \emph{leakage information} includes the truth values of conditions, memory access addresses, and the addresses of invoked functions. Their approach relies on small-step semantics and CT-simulation~\cite{DBLP:conf/csfw/BartheGL18}---a three-dimensional extension of CompCert’s flat simulation, forming a ``cube diagram''.
{We can also address the CTS preservation problem by enriching the event trace of each behavior set in our semantics with leakage information, defining the CTS property using denotational semantics, and showing the CTS preservation compositionally.  }

\paragbf{{Our limitation}}
Currently, our framework supports semantic linking for heterogeneous languages, provided they share the same memory model or employ memory models that can be mutually transformed. Moreover, we have not yet addressed concurrent settings or higher-order features (i.e., features of higher-order languages). We are particularly interested in exploring how our approach can be extended to these domains. Furthermore, unlike interaction trees or small-step semantics, our semantics is not directly executable.
Nevertheless, we have established formal equivalence with small-step operational semantics for simple imperative languages as mentioned above. We believe that extending the equivalence proof to realistic languages presents no fundamental difficulty.

\section{Related Work}
\label{sec:related_work}

\paragbf{Denotational semantics.}

In traditional denotational semantics, two fixed-point theorems are ubiquitous:
the Kleene fixed point theorem (for Scott-continuous operators on CPOs) and
the Knaster-Tarski fixed point theorem (for monotone operators on complete lattices).
Researchers have long tended to commit to one of these frameworks when handling recursion, e.g., Kleene’s theorem in domain theory~\cite{scott1970outline,scott1971toward}, and Knaster-Tarski's theorem in lattice-based frameworks such as UTP~\cite{DBLP:conf/RelMiCS/JifengH98,DBLP:conf/ifm/WoodcockC04}.
This commitment entails adopting a single global order over the semantic domain, endowed with (global) properties such as continuity or completeness.
However, a tension has long been observed: “trivial” domains that easily satisfy these properties often cannot express important program features, whereas richer domains make the required global completeness or continuity difficult to establish.
This tension has motivated increasingly elaborate fixed-point theories.
The phenomenon is particularly evident in the evolution of CSP’s denotational semantics~\cite{DBLP:series/txcs/Roscoe10} (following Hoare and subsequent work): from the \emph{trace model}, to the \emph{(stable) failures model}, to the failures-divergences model~\cite{DBLP:journals/jacm/BrookesHR84}, and eventually to more expressive extensions (e.g., incorporating infinite traces~\cite{DBLP:conf/birthday/Roscoe04}).
 Rather than binding ourselves to a single fixed-point method or a uniform global order, we advocate a staged (componentwise) approach to applying fixed-point theorems: equip each semantic component with its natural (simple) order, and apply the most suitable fixed-point principle per component. This staged and componentwise strategy sidesteps contrived global structures while preserving clarity, proof simplicity, and applicability.

More specifically, the foundational work of~\citeN{scott1970outline} and~\citeN{scott1971toward} in domain theory provides a mathematical framework for the denotational semantics of deterministic programs.
  Following their work, later efforts \cite{DBLP:conf/pc/BroyGW78,DBLP:conf/ac/Park79,apt1981cook,DBLP:journals/tcs/Back83} attempt to extend the framework for supporting nondeterministic programs, especially programs that can produce an infinite number of different results and yet be certain to terminate.
  Among them,
  \citeN{DBLP:journals/tcs/Back83} describes a path semantics for nondeterministic assignment  statements indicated by semantic models described for CSP programs \cite{DBLP:journals/jcss/FrancezHLR79} and data flow programs \cite{DBLP:conf/popl/Kosinki78} based on paths.
  The relational style of defining denotations in this article is inspired from the relational semantics of~\citeN{DBLP:conf/ac/Park79}, which originally aims to address unbounded nondeterminism in the fair scheduling problem of concurrency. See our comparison with this work in \S\ref{subsec:denote_bg} and \S\ref{subsec:while_lang}.

  Besides, \citeN{DBLP:conf/ifip/Kahn74} describes a class of asynchronous deterministic parallel programs, now known as Kahn networks. These networks are represented as systems of recursive equations based on stream transformations. Each process in the network consumes input streams and produces output streams.  
  Kahn also provides a denotational semantics for these networks by leveraging complete partial orders (CPOs), which focuses on stream-based communication between parallel processes. \citeN{paulin2009constructive} presents a formalization of denotational semantics for Kahn networks in the \Coq proof assistant. The formalization includes definitions of processes, channels, and their compositions.
In addition, denotational semantics has been employed to construct semantic models for probabilistic programming languages, as illustrated in works such as \citeN{DBLP:conf/esop/BartheEGGHS18}, \citeN{DBLP:journals/entcs/WangHR19}, and \cite{DBLP:journals/pacmpl/VakarKS19}. 
Among these works, \citeN{DBLP:conf/esop/BartheEGGHS18} adopt the traditional semantics of probabilistic distributions to develop an assertion-based program logic for reasoning about probabilistic programs. \citeN{DBLP:journals/entcs/WangHR19} tackle the interplay of nondeterminism and probability through domain theory. \citeN{DBLP:journals/pacmpl/VakarKS19} support higher-order functions by employing quasi-Borel spaces (QBS) and omega-complete quasi-Borel spaces  (\(\omega\)QBS).

\paragbf{Compiler verification with interaction trees} The Vellvm project~\cite{DBLP:journals/pacmpl/XiaZHHMPZ20,LLVM-IR} uses Interaction Trees (ITrees)~\cite{DBLP:journals/pacmpl/XiaZHHMPZ20} to model the semantics of LLVM
IRs and to verify compiler transformations. Specifically, interaction trees provide a denotational approach for modeling recursive, effectful computations that can interact with their environment.
Much different from the textbook denotational semantics and the relational denotational semantics used in this article, interaction trees are executable via code extraction, making them suitable for debugging, testing, and implementing software artifacts. 
To this end, the ITree-based denotational semantics is also used by~\citeN{CCR} for building executable module semantics in Conditional Contextual Refinement (CCR), and~\citeN{ITree24} further scales the semantics for higher-order languages.

  In comparison, a principal advantage of our approach is the smooth proofs for horizontal compositionality and for semantic adequacy (i.e., the equivalence between syntactic and semantic linking for open modules shown in Thm.~\ref{thm:ss_equiv}). The key observation is that the adequacy proofs reduce directly to the Beki\'{c} theorem~\cite{DBLP:conf/ibm/Bekic84e}.
   ITree-based frameworks model function calls and semantic composition via event handlers. To the best of our knowledge, within these frameworks the equivalence between syntactic and semantic linking for open modules has not been fully established, nor is it clear how such an equivalence could be derived from the Beki\'{c} theorem.

Moreover, ITrees use \kwd{Tau} to hide internal computation steps (i.e., steps with no observable events). Establishing refinement (e.g., strong bisimulation) between two ITrees may require additional treatment of \kwd{Tau} steps---that is, any finite number of \kwd{Tau} can be removed when determining equivalence. In comparison, our refinement directly relates the target event trace with \emph{strictly equal} (or \emph{prefixing} when the source aborts) source event trace. This is possible because our semantics distinguishes finite and infinite event traces via the behavior sets \kfin\ and \kinf straightforwardly.

\paragbf{{Advanced simulation techniques beyond CompCert}}
  In CompCert’s backward simulation, the central proof obligation is that: under a given relation $R$, for any pair of related states $(\sigma_t,\sigma_s)\in R$, if the target can take a step $\sigma_t \to \sigma_t'$, then the source can simulate this behavior by taking zero or more steps $\sigma_s \to^{*} \sigma_s'$ so that $(\sigma_t',\sigma_s')\in R$ and the same observable events are emitted. 
  A broad class of more advanced techniques relaxes this obligation by allowing the post‑step states to be related by a \emph{weaker} relation \(R'\), i.e., $(\sigma_t',\sigma_s')\in R'$ with $R'$ not necessarily equal to $R$ (e.g.,~\cite{Pous2016}). To ensure soundness, such approaches tend to impose additional \emph{recovery} simulation rules (for example, that $R$ be re‑established at designated synchronization points).
  {As one of the state-of-the-art efforts,} FreeSim~\cite{FreeSim} introduces \emph{progress indices} into the matching relation and provides new simulation rules for flexible simulation proofs.
    While the emphasis of FreeSim is put on extending stuttering simulations~\cite{StutterSim} to ensure soundness with only asynchronous progress (i.e., a simulation proof step that allows one side of the simulation to stay in place (stutter) while the other side to step forward) required,
    using progress indices also allows intermediate proof steps to proceed under weaker relations while ensuring that stronger invariants are recovered when required.
Consequently, in \emph{small-step semantics} based settings,
FreeSim can reduce the CompCert-style burden of preserving a single global matching relation by every step, thereby easing proofs involving \emph{intermediate correspondence}. This is illustrated in Fig.~\ref{fig:relaxed-bsim}: transforming ``\kwd{print(x); y = f(x); print(y);}'' into ``\kwd{y = f(x); print(x); print(y);}'' is justified when \kwd{f(x)} is a pure, deterministic computation that produces no observable behavior.
\begin{figure}[htb]
    \begin{subfigure}[b]{.48\linewidth}
    \centering
        \includegraphics[width=0.95\linewidth]{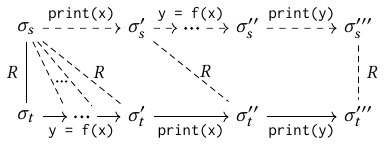}
    \caption{An instance of CompCert backward simulation}
    \label{subfig:compcert-bsim}
    \end{subfigure} 
    \begin{subfigure}[b]{.48\linewidth}
        \flushright
        \includegraphics[width=0.95\linewidth]{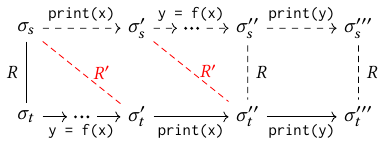}
    \centering
    \caption{An instance of advanced simulation diagrams}
    \label{subfig:relaxed-bsim}
    \end{subfigure}   
    \caption{Comparison between CompCert's simulation with advanced simulation techniques.}
    \label{fig:relaxed-bsim}
    \Description{This figure compares the standard backward simulation used in CompCert with a more flexible advanced simulation style. Subfigure (a) shows a CompCert-style backward simulation step: from a pair of related target and source states under a single relation \(R\), one target step must be matched by zero or more source steps, and the resulting states must again satisfy the same relation \(R\) while producing the same observable events. Subfigure (b) shows an advanced simulation diagram: after a target step and the matching source steps, the resulting states may satisfy a weaker intermediate relation \(R'\) instead of the original relation \(R\). Additional recovery or synchronization steps are then used to re-establish the stronger relation when needed. The figure illustrates how advanced simulation techniques relax the requirement that every proof step preserve one fixed global matching relation, making it easier to reason about intermediate correspondence in transformations such as reordering a pure computation around observable \texttt{print} statements.}
\end{figure}

For the example in Fig.~\ref{fig:f2bsim-eg1}, the flexibility of Freesim can make the proofs that the fragment code \kwd{g} simulates \kwd{f} and the proofs that the larger program $T$ simulates $S$ simpler and easier to carry out respectively. However, it remains unclear how to modularly reuse the former proofs for the latter (w.r.t. proofs based on small-step semantics). In comparison, our denotation-based approach addressed such issues illustrated in \S\ref{subsec:compare_bref}.

{Regarding \emph{nondeterminism correspondence},
FreeSim comfortably handles scenarios where nondeterministic choices occur entirely \emph{before} or entirely \emph{after} an event-synchronization step. However, when \emph{demonic} choices occur on \emph{opposite sides} of such a step (i.e., \textsc{step-event} in FreeSim) and influence subsequent observable events, a small-step semantics based proof in FreeSim generally fails: the later demonic step may not match a value chosen earlier on the other side. Still consider:
\[
\begin{aligned}
\text{Source program}:\quad 
&\mathcolor{gray}{[\sigma_s]}\;
\kwd{choice(x = 0, x = 1);}\; \kwd{print(0);}\;
\mathcolor{gray}{[\sigma_s']}\; \kwd{print(x);} \\
\text{Transformed to }:\quad 
&\mathcolor{gray}{[\sigma_t]}\;\kwd{print(0);}\;
\mathcolor{gray}{[\sigma_t']}\;\kwd{choice(x = 0, x = 1);}\; \kwd{print(x);}
\end{aligned}
\]
At the first event synchronization on \kwd{print(0)}, the proof obligation becomes to relate the \emph{post‑event} states \((\sigma_s',\sigma_t')\) under some derived relation \(R'\).
From there, the target would perform its internal demonic step \kwd{choice}, and the proof must succeed \emph{for all} outcomes \(x_t\in\{0,1\}\).
If \(x_t\neq x_s\) (where \(x_s\) was chosen earlier by the source), the subsequent visible actions \kwd{print(x)} do not match, so the required second event synchronization cannot be discharged and the simulation gets stuck.

Early research \cite{DBLP:journals/tcs/AbadiL91,DBLP:journals/iandc/LynchV95} proposes the use of \emph{prophecy variables} to enhance small-step semantics based simulations for handling such general nondeterminism correspondence. 
 They leverage prophecy variables to anticipate and reconcile nondeterministic choices between the source and target programs, ensuring that simulation aligns with their observable event traces.
 However, these approaches are often challenging to implement in practice \cite{DBLP:journals/toplas/LamportM22}. As far as we know, they have not yet been formalized and used for realistic compiler verification.
 {CompCert achieves a balance between theoretical complexity and practical applicability by eliminating nondeterminism from the outset.}
 
Besides, both relational denotational semantics and big-step operational semantics exhibit equal advantages when dealing with such a specific problem of nondeterminism; see \S\ref{subsec:compare_bref} for our solution.

\paragbf{Choice Trees (CTrees)} CTrees~\cite{ctrees23,ctrees25} extend ITrees with \emph{two} forms of internal choice to facilitate the reasoning about nondeterminism: \emph{stepping} choice, which materializes as $\kwd{Tau}$‑steps in the induced labeled transition system (LTS), and \emph{delayed} choice, which is structural and does not create LTS transitions. This extension enables CTrees to address the above problem: when a choice is modeled as \emph{delayed}, CTrees’ equational laws (congruent for composition) allow that choice to commute with the next observable event, aligning programs where one makes the choice before the event and the other after. In comparison, we model the whole event traces of a program in behavior sets (simpler than a tree structure), and hence the correspondence is immediate.

\paragbf{Compositional compiler verification.}
Many existing works have extended the small-step semantics of CompCert for supporting cross-module semantic linking. Among them, 
 \citeN{DBLP:conf/cpp/RamananandroSWKF15} model the behavior of an external call as a special event recording the function name and memory states before and after the external call so as to signal state transitions made by the environment and then ``big-step'' the small-step semantics to obtain the compositional semantics for semantic linking.
 \citeN{DBLP:conf/esop/BeringerSDA14}
 propose a novel interaction model, called core semantics (also known as interaction semantics), that describes the communication between a local thread with its environment, which is widely used by later efforts such as CompComp \cite{DBLP:conf/popl/StewartBCA15}, CompCertM \cite{DBLP:journals/pacmpl/SongCKKKH20} and CASCompCert \cite{DBLP:conf/pldi/JiangLXZF19}. 
Specifically, with the interaction semantics,
CompComp proves a general correctness result for semantic linking where programs can be written with heterogeneous languages such as C and assembly code. However, to achieve horizontal (module-level) and vertical compositionality, it introduces complicated \emph{structured simulations} which require a large number of changes to the original proofs of CompCert.

Compared to CompComp, SepCompCert \cite{DBLP:conf/popl/KangKHDV16} greatly reduces the complexity of proofs by limiting all source modules to be compiled by the same compiler. 
Furthermore, both CompCertM \cite{DBLP:journals/pacmpl/SongCKKKH20} and CompCertO \cite{DBLP:conf/pldi/KoenigS21} have developed a lightweight verification technique for supporting semantic linking of heterogeneous languages. CompCertM proposes a RUSC (Refinement Under Self-related Contexts) theory for composing two open simulations together, while CompCertO characterizes compiled program components directly in
terms of their interaction with each other and achieves compositionality through a careful
and compositional treatment of calling conventions. Recently, 
 \citeN{DBLP:journals/pacmpl/ZhangWWKS24} have proposed a fully composable and adequate approach to verified compilation with direct refinements between open modules based on CompCertO.
 We believe that their approach to handling calling conventions of heterogeneous languages can be applied to our framework.
\citeN{DBLP:journals/pacmpl/0001A19} make a comprehensive survey of compositional compiler correctness. They take existing compiler correctness theorems as a spectrum which ranges from CompCert, SepCompCert, CompCertX~\cite{DBLP:conf/popl/GuKRSWWZG15,DBLP:journals/pacmpl/WangWS19}, and CompComp up to compilers verified with multi-language techniques.

\paragbf{Other verified compiler or compilation passes}
 CertiCoq \cite{DBLP:journals/pacmpl/Paraskevopoulou21} and CakeML \cite{DBLP:conf/popl/KumarMNO14} are typical verified functional language compilers based on operational semantics. We believe that our method can also be extended to support typed lambda calculus and the verification of functional language compilation. Besides, the CompCert back-end is mostly implemented in control flow graphs. 
If procedure calls are not considered, there is no big difference between denotational semantics and small-step semantics in the back-end. However, we still believe that the proposed semantics using semantic oracle to interpret the behavior of callee functions is beneficial for verifying these back-end passes and then achieving module-level compositionality.
We are interested in these topics and leave them as future work.

\paragbf{Discussion on program verification}
The small-step semantics of Clight is also used in static analyzers like Verasco \cite{DBLP:phd/hal/Jourdan16} and program logics like VST-Floyd \cite{DBLP:journals/jar/CaoBGDA18}. The soundness proof of VST is connected to the proof of CompCert, ensuring that the correctness properties proved for Clight programs are preserved through compilation down to the target machine code.
We believe that for sequential programs, rebuilding Hoare-like logics on denotational semantics should be simple. For concurrent programs, more trace-related elements need to be incorporated into the semantics so as to provide semantic basis for related program logics.

\section{Conclusion}
\label{sec:conclusion}
We have proposed a denotation-based framework for compositional compiler verification. In this framework, we have scaled relational denotational semantics to realistic settings that capture diverse program features, and we have defined a novel semantic linking operator based on fixed-point theorems, reducing its equivalence with syntactic linking to concise fixed-point properties. In particular, we propose the notion of refinement algebras to define the refinement relations of interest and to facilitate the automation of the required proofs, and support verified compositional compilation in a language-independent way. We have applied this denotational semantics to the front-end languages of CompCert, reproved the correctness of compilation from Clight to Cminor, and demonstrated that the denotation-based approach is suitable for verifying typical optimizations.

\begin{acks}
We are grateful to Andrew W. Appel, Naijun Zhan, Zhenjiang Hu, Zhong Shao, Yuxi Fu, Yuxin Deng, Yuting Wang, and Huan Long for insightful suggestions and helpful discussions.
We also thank Ximeng Li, Shuling Wang, and Yiyuan Cao for careful feedback.
We thank the editors and the anonymous TOPLAS referees for their thoughtful and constructive comments.
This work was supported by the National Natural Science Foundation of China (Grant No.~62472274,~61902240).
\end{acks}

\newpage
\bibliographystyle{ACM-Reference-Format}
\bibliography{references.bib}

\end{document}